\documentclass[a4paper, 12pt, twoside]{article}
\usepackage[a4paper, margin=0.8in]{geometry}
\usepackage[utf8]{inputenc}
\usepackage[T1]{fontenc}
\usepackage[english]{babel}
\usepackage{textcomp} 
\usepackage{graphicx}
\DeclareGraphicsExtensions{.jpg,.mps,.pdf,.png,.gif}
\usepackage[french]{varioref} 
\usepackage[]{amsmath,amssymb,amsfonts}

\usepackage{float}
\usepackage{amsthm}
\usepackage{bbm}
\usepackage{dsfont}
\usepackage{color}
\usepackage{natbib}
\usepackage{titlesec}

\usepackage[hyperindex,breaklinks]{hyperref}
\usepackage{slashbox}

\newcommand{\Ben}{\begin{enumerate}}
\newcommand{\Een}{\end{enumerate}}
\newcommand{\Bit}{\begin{itemize}}
\newcommand{\Eit}{\end{itemize}}
\newcommand{\Beq}{\begin{equation}}
\newcommand{\Eeq}{\end{equation}}
\newcommand{\Ba}{\begin{align*}}
\newcommand{\Ea}{\end{align*}}
\newcommand{\Mb}{\mathbf}

\newcommand{\Bs}{\boldsymbol}
\newcommand{\Mbb}{\mathbb}
\newcommand{\Mr}{\mathrm}

\newtheorem{Th}{Theorem}

\newtheorem{Prop}{Proposition}

\newtheorem{Rq}{Remark}

\title{Stochastic derivative estimation for max-stable random fields}
\date{November 3, 2020}

\begin{document}

\author{
~~ Erwan Koch\footnote{EPFL, Institute of Mathematics, EPFL SB MATH MATH-GE, MA B1 457 (Bâtiment MA), Station 8, 1015 Lausanne,
Switzerland. \newline Email: erwan.koch@epfl.ch}
~~ Christian Y. Robert\footnote{Laboratory in Finance and Insurance - LFA CREST - Center for Research in Economics and Statistics, ENSAE, Paris, France. \newline Email: christian-yann.robert@ensae.fr}
}

\maketitle

\begin{abstract}
We consider expected performances based on max-stable random 
fields and we are interested in their derivatives with respect to the spatial dependence parameters of those fields. Max-stable fields, such as the Brown--Resnick and Smith fields, are very popular in spatial extremes. 
We focus on the
two most popular unbiased stochastic derivative estimation approaches: the
likelihood ratio method (LRM) and the infinitesimal perturbation analysis
(IPA). LRM requires the multivariate density of the
max-stable field to be explicit, and IPA necessitates the
computation of the derivative with respect to the parameters for each simulated value.
We propose convenient and tractable conditions ensuring the validity of LRM and IPA in the cases of the Brown--Resnick and Smith field, respectively. Obtaining such conditions is intricate owing to the very structure of max-stable fields. Then we focus on risk and dependence measures, which constitute one of the several frameworks where our theoretical results can be useful. We perform a simulation study which shows that both LRM and IPA perform well in various configurations, and provide a real case study that is valuable for the insurance industry.

\medskip

\end{abstract}

\section{Introduction}

Sensitivity analysis (SA) is essential in many fields and constitutes an active research area. It consists in the quantitative assessment of how
changes in a specific model parameter impact the so-called expected performance. The expected performance is the expectation of a function of an underlying stochastic model, and SA aims at computing its derivatives with respect to the model parameters. If the expected performance is excessively sensitive to some
critical parameters, it warns the decision maker not to be overly confident in the obtained value, especially if there is a huge
uncertainty on these parameters. It also informs him/her that energy should
be invested to find more reliable estimates of these parameters (e.g., by
developing new inference methods). By the delta method, the knowledge of the derivative of the expected performance with respect to a parameter of interest allows one to translate confidence bounds for this parameter into confidence bounds for the expected performance.
Finally, if the aim is to maximize the expected performance with respect to some parameter, knowing its derivative is often required as most optimization
algorithms are based on gradient methods.

To the best of our knowledge, SA was never considered for expected performances based on extreme-value models, although this is of high necessity. Indeed, such models are by essence generally fitted using a small amount of data, resulting in a large uncertainty in the parameter estimates. 
For a practical introduction about extreme-value theory, see, e.g., \cite{coles2001introduction}. Extreme-value models are valuable in many applications. Max-stable random fields \citep[e.g.,][]{haan1984spectral, de2007extreme}, which constitute an extension of univariate extreme-value theory to the infinite-dimensional (e.g., spatial) setting, are, for instance, well suited for the modelling of the pointwise maxima of variables having a spatial extent such as environmental ones. Such a modelling is useful for quantifying the impact of natural disasters, which is essential for insurance companies and civil authorities, especially in a context of climate change. However, owing to the complex structure of max-stable fields, their multivariate density is in general not
available 
for more than two or three sites, making estimation highly non-trivial. Commonly used estimators are not asymptotically efficient (in the sense of the Cram\'{e}r-Rao bound), which, combined with the uncertainty stemming from the scarcity of data mentioned above, can lead to estimates that are far from the true value.
It is therefore essential to assess the 
sensitivity of any expected performance based on max-stable fields with respect to the
parameters. The purpose of this paper is precisely to provide tractable conditions enabling such an assessment and to illustrate this in a case where the expected performance is a risk or dependence measure. 




We focus on expected performances that do not have any closed-form expression, entailing that assessment of their
derivative cannot be achieved using differentiation or finite differences
based on the analytical expression, but requires
simulation-based approaches. There are
essentially three classic stochastic derivative estimation methods. The
first one is based on finite-differences. Corresponding estimators
involve a bias-variance trade-off and require simulating at multiple
parameter values, making this methodology less powerful than the two others.
The second approach, referred to as likelihood ratio method (LRM), is based
on the derivatives of the density function associated with the simulation
model. The third one, called infinitesimal perturbation analysis (IPA),
relies on computing the derivative with respect to the parameter of each
simulated value (and then averaging them); IPA is thus also
referred to as a sample path differentiation 
method. Contrary to
IPA where the parameter is considered as purely structural, in LRM it is
viewed as a parameter of the probability measure. LRM requires the
density function to be explicit and differentiable and, on the other hand, IPA requires differentiability of the sample paths  with respect to the parameter, which is a quite strong condition. It is argued in \citet[][Section 7.4]{glasserman2013monte} that LRM estimators often have a larger variance than IPA estimators, explaining why IPA is generally
considered as the best derivative estimator. A huge literature is dedicated to LRM and IPA; excellent general references
are in particular \cite{asmussen2010}, Chapter VII, and \cite%
{glasserman2013monte}, Chapter 7. Both approaches are widely applied to many
fields, for instance to finance where many risk hedging strategies involve
computing sensitivities of option prices to the underlying assets' prices
and other parameters (the so-called Greeks). \cite{broadie1996estimating}
and \cite{chen2001efficient} use IPA for option pricing and mortgage-backed
securities, respectively. Motivated by financial applications, \cite%
{glasserman2010sensitivity} investigate LRM in the case where the relevant
densities are only known through their characteristic functions or Laplace
transforms.

To our knowledge, SA, and thus LRM\ and IPA, have not yet been considered for expected performances based on extreme-value models, a fortiori max-stable fields. Among both methods, none can be used for all max-stable fields systematically. LRM is often ruled out when the random performance involves values
of the field at more than two or three sites since the multivariate density is not explicitly available. However, there are notable exceptions such as the
Brown--Resnick random field \citep{brown1977extreme,
kabluchko2009stationary} or the extremal t random field \citep{opitz2013}.
The main technical challenge to apply LRM is to show that
interchange between derivative (with respect to the parameter) and
integration over the space of possible values of the field at the sites considered is feasible. On the other hand, IPA\ involves showing that the paths of the random field are differentiable with respect to the spatial dependence parameters, which is often arduous as max-stable fields arise as the pointwise maxima over an infinite number of latent random fields. Nevertheless, when the paths of these latent fields are sufficiently smooth, such as for the Smith random field \citep{smith1990max}, it is possible. The Smith field belongs to the class of Brown--Resnick fields, but, being a border case, its properties in terms of availability of a closed-form multivariate density and smoothness of the paths are very different than for most Brown--Resnick fields.


Our main theoretical contribution consists in proposing convenient conditions ensuring the validity of LRM and IPA for estimating the derivatives of expected performances based on the Brown--Resnick field and the Smith field, respectively. The Brown--Resnick field is very flexible and is one of the most appropriate models for environmental data among currently known max-stable models \citep[e.g.,][]{davison2012statistical}. We focus on the derivatives with respect to the spatial dependence parameters, since showing the validity of LRM or IPA for estimating the derivatives with respect
to the marginal parameters (those of the generalized extreme-value
distribution) is easy (under mild assumptions). Our conditions are as
tractable as possible; e.g., in the case of IPA for the Smith field, the condition involves the derivative of the expected performance with
respect to the field values but not to the spatial dependence parameters. Obtaining such practical conditions is difficult owing to the complex structure of max-stable fields.

Our second contribution pertains to the context of risk assessment in insurance or finance. Our results are insightful for studying the sensitivity of risk or dependence measures based, e.g., on insured losses triggered by extreme events having a spatial extent (typically such as weather events), or the sensitivity of
prices of event-linked securities such as catastrophe bonds. Indeed, many risk or dependence measures as well as prices can be written as expected performances. After providing a general discussion, we focus on a specific dependence measure for insured losses due to extreme wind speeds. This measure has analytical derivatives for both the Brown--Resnick and the Smith fields, and hence allows us to compare the LRM and IPA estimates with the true values. We implement LRM and IPA (adapting when necessary existing simulation algorithms) and perform a thorough numerical study which shows that both estimation methods are very accurate. Finally, we propose a concrete application which is valuable for actuarial practice. Using reanalysis wind speed data, we study the sensitivities of the aforementioned dependence measure in an area centred over the Ruhr region in Germany. This application highlights the importance of assessing the sensitivity of risk or dependence measures in real practice. 

The remainder of the paper is organized as follows. Section \ref
{Subsec_EstimationSensMaxstable} first recalls useful results about max-stable fields and presents the concept of expected performance. Then it introduces LRM and IPA with their associated examples of max-stable fields (Brown--Resnick and Smith, respectively) and states our conditions guaranteeing the validity of both methods in these cases. Section \ref{Sec_RiskAssessment} is devoted to risk assessment applications. After a general discussion, we introduce the aforementioned specific
dependence measure, we present the simulation study and we finally expose our real case study. Section \ref{Sec_Conclusion} briefly summarizes our contribution and presents some perspectives. Throughout the paper, we shall use the following notations. Let ${}^{\prime
} $ denote transposition. For $\mathbf{x}=(x_{1},\dots
,x_{d})^{\prime }\in \mathbb{R}^{d}$, $\left\Vert \mathbf{x}\right\Vert ^{2}=%
\mathbf{x}^{\prime }\mathbf{x}=\sum_{i=1}^{d}x_{i}^{2}$, for a
matrix $A\in \mathbb{R}^{d\times d}$, $\left\Vert A\right\Vert =\sup
\{\left\Vert A\mathbf{x}\right\Vert :\mathbf{x}\in \mathbb{R}^{d}$ such that 
$\left\Vert \mathbf{x}\right\Vert =1\}$ and, for a positive definite
symmetric matrix $\Sigma \in \mathbb{R}^{d\times d}$, $\left\Vert \mathbf{x}%
\right\Vert _{\Sigma ^{-1}}^{2}=\mathbf{x}^{\prime }\Sigma ^{-1}\mathbf{x}$. Furthermore, $\bigvee$ denotes the supremum when the latter is taken over a countable set. Finally, $\overset{d}{=}$ denotes equality in distribution. In the case of random fields, by distribution we mean the set of all finite-dimensional distributions. All proofs can be found in the appendix.

\section{Stochastic derivative estimators}

\label{Subsec_EstimationSensMaxstable}

\subsection{Max-stable random fields and associated expected performance}

A random field $X$ on $\mathbb{R}^{d}$ with non-degenerate marginals is
called max-stable if there are continuous functions $a_{n}(\mathbf{\cdot }%
)>0 $ and $b_{n}(\mathbf{\cdot })$ on $\mathbb{R}^{d}$ such that if $%
X_{1},\ldots, X_n $ are independent copies of $X$ then
\begin{equation*}
\bigvee_{i=1}^{n}\frac{X_{i}-b_{n}}{a_{n}}\overset{d}{=}X,\qquad n=1,2,\ldots,
\end{equation*}%
(pointwise maxima). A random field $Y$ on $\mathbb{R}^{d}$ with standard Fr\'{e}chet margins (i.e., $\Mbb{P}(Y(\Mb{x}) \leq y)=\exp(-1/y)$, $y>0$, $\Mb{x} \in \Mbb{R}^d$) is max-stable if 
\begin{equation*}
\bigvee_{i=1}^{n}n^{-1}Y_{i}\overset{d}{=}Y,\qquad n=1,2,\ldots,
\end{equation*}
where $Y_{1},\ldots, Y_n$ are independent copies of $Y$; such a field is said to be simple max-stable. It is well-known that
there exist continuous functions $\eta (\mathbf{\cdot })$, $\tau (\mathbf{%
\cdot })>0$ and $\xi (\mathbf{\cdot })$ on $\mathbb{R}^{d}$, called the
location, scale and shape functions such that%
\begin{equation}
X(\mathbf{x})\overset{d}{=}\left\{ 
\begin{array}{ll}
\left( \eta (\mathbf{x})-\tau (\mathbf{x})/\xi (\mathbf{x})\right) +\tau (%
\mathbf{x})Y(\mathbf{x})^{\xi (\mathbf{x})}/\xi (\mathbf{x}), & \quad \xi (%
\mathbf{x})\neq 0 \\ 
\eta (\mathbf{x})+\tau (\mathbf{x})\log (Y(\mathbf{x})), & \quad \xi (%
\mathbf{x})=0%
\end{array}%
\right. .  \label{Eq_Link_Maxstb_Simple_Maxstab1}
\end{equation}

Any simple max-stable random field $Y$ on $\mathbb{R}^{d}$ can be written %
\citep[e.g.,][]{haan1984spectral} as 
\begin{equation}
Y\overset{d}{=}\bigvee_{i=1}^{\infty }U_{i}Z_{i},
\label{Eq_Spectral_Representation_Stochastic_Processes}
\end{equation}%
where the $(U_{i})_{i\geq 1}$ are the points of a Poisson point process on $%
(0,\infty )$ with intensity function $u^{-2}\mathrm{d}u$ and the $%
Z_{i},i\geq 1$, are independent copies of a non-negative random field $Z$
such that, for all $\mathbf{x}\in \mathbb{R}^{d}$, $\mathbb{E}[Z(\mathbf{x}%
)]=1$. Conversely, any field of the form %
\eqref{Eq_Spectral_Representation_Stochastic_Processes} is simple max-stable. We assume that the distribution of $Z$ depends on a spatial dependence
parameter $\Bs{\theta }\in \mathbb{R}^{L}$, and we now denote $Z$ by $Z_{%
\Bs{\theta}}$ (as well as $Y$ by $Y_{\Bs{\theta}}$). The
dependence of $Z_{\Bs{\theta}}$ with respect to $\Bs{\theta}$ will
be explicitly described below.

We consider $M$ sites $\mathbf{x}_{1},\ldots ,\mathbf{x}_{M}\in \mathbb{%
R}^{d}$, and let $\mathbf{Y}_{\Bs{\theta}}=\left( Y_{\Bs{\theta}}(%
\mathbf{x}_{1}),\ldots ,Y_{\Bs{\theta}}(\mathbf{x}_{M})\right) ^{\prime
}$. It is known that the distribution function of $\mathbf{Y}_{\Bs{\theta}}$ is
\begin{equation*}
F_{\Bs{\theta}}\left( \mathbf{y}\right) =\exp \left( -V_{\mathbf{\theta 
}}(\mathbf{y})\right) ,\quad \mathbf{y}\in (0,\infty )^{M},
\end{equation*}%
where $V_{\Bs{\theta}}$ is the so-called exponent measure function given by
\begin{equation*}
V_{\Bs{\theta}}(\mathbf{y})=\mathbb{E}\left[ \bigvee_{i=1}^{M}\frac{Z_{%
\Bs{\theta}}(\mathbf{x}_{i})}{y_{i}}\right] ,\quad \mathbf{y}\in
(0,\infty )^{M}.
\end{equation*}%
If the $M$-th order partial derivatives of $V_{\Bs{\theta}}$ exist,
then the density function of $\mathbf{Y}_{\mathbf{\theta }}$ may be derived
by the Fa\`{a} di Bruno's formula for multivariate functions. Let $\mathcal{I}%
=\{1,\ldots ,M\}$. We denote by $\Pi $ the set of all partitions of $%
\mathcal{I}$ and, for a partition $\pi \in \Pi $, $B\in \pi $\ means that $B$ is
one of the blocks\ of the partition $\pi $. Moreover, for any set $B\subset \mathcal{I}$ and $\mathbf{y}\in
(0,\infty )^{M}$, we let $\mathbf{y}_{B}=\left( y_{j}\right) _{j\in B}$. The density
function of $\mathbf{Y}_{\Bs{\theta}}$ is then given by%
\begin{equation}
f_{\Bs{\theta}}\left( \mathbf{y}\right) =\exp \left( -V_{\mathbf{\theta 
}}(\mathbf{y})\right) \sum_{\pi \in \Pi }\left( -1\right) ^{|\pi
|}\prod\limits_{B\in \pi }\frac{\partial ^{|B|}}{\partial \mathbf{y}_{B}}V_{%
\Bs{\theta}}(\mathbf{y}), \label{Eq_Density_MaxStable}
\end{equation}%
where $|\pi |$ denotes the number of blocks of the partition $\pi $, $|B|$
the cardinality of the set $B$ and $\partial ^{|B|}/\partial \mathbf{y}_{B}$
the partial derivative with respect to $\mathbf{y}_{B}$.

The random performance (sometimes also called model output) is defined by $H\left( \mathbf{Y}_{\Bs{\theta}
}\right)$, where $H$ is a function from $\mathbb{
R}^{M}$ to $\mathbb{R}$, and the expected performance is its expectation, i.e., 
\begin{equation}
R(\Bs{\theta})=\mathbb{E}\left[ H\left( \mathbf{Y}_{\Bs{\theta}%
}\right) \right],  \label{Eq_GenQuantityToDerive}
\end{equation}%
provided that $\mathbb{E}\left[ |H\left( \mathbf{Y}_{\Bs{\theta}}\right) |%
\right] <\infty $. 
Below we give conditions ensuring the applicability of LRM and IPA to estimate $\partial
R\left( \Bs{\theta}\right) /\partial \Bs{\theta}$ at a specific $\Bs{\theta}_0 \in \Mbb{R}^L$, for $\Mb{Y}_{\Bs{\theta}}$ built from the Brown--Resnick and the Smith fields, respectively. The quantity in
\eqref{Eq_GenQuantityToDerive} is pretty general and the results developed
in this section can thus be applied for various purposes. One of them,
particularly important in finance and insurance, is to estimate the
sensitivities of risk and dependence measures; this is done in Section \ref
{Sec_RiskAssessment}.

When fitted to data, a max-stable field is not simple but written as in \eqref{Eq_Link_Maxstb_Simple_Maxstab1} with location, scale and shape functions not uniformly equal to unity. Hence, denoting by $X_{\Bs{\theta}}$ such a field and letting $\mathbf{X}_{\Bs{\theta}}=\left( X_{\Bs{
\theta }}(\mathbf{x}_{1}),\ldots ,X_{\Bs{\theta}}(\mathbf{x}
_{M})\right) ^{\prime }$, the expected performance should be written 
\Beq
\label{Eq_GenQuantityToDerive_G}
\mathbb{E}\left[ G\left( \mathbf{X}_{\Bs{\theta}
}\right) \right],
\Eeq 
where $G$ is a function from $\mathbb{R}^{M}$ to $\mathbb{R}$ such that $\mathbb{E}\left[ |G\left( \mathbf{X}_{\Bs{\theta}}\right) |%
\right] <\infty$. However, we see by \eqref{Eq_Link_Maxstb_Simple_Maxstab1} that the representation \eqref{Eq_GenQuantityToDerive} encompasses the case of \eqref{Eq_GenQuantityToDerive_G} by letting $H$ depend on the location, scale and shape parameters of $X_{\Bs{\theta}}$ at the $M$ sites. As technical issues arise for the derivatives with respect to the spatial dependence parameters only, we do not consider differentiation with respect to the marginal parameters, (although this is of practical interest), and focus on \eqref{Eq_GenQuantityToDerive} for convenience.

\subsection{Likelihood ratio method}

\label{Subsec_LRM}

\subsubsection{General methodology} 

We first introduce LRM to the context of expected performances based on max-stable fields.
Let $\mathbf{Y}_{\Bs{\theta}}$ be a simple max-stable random vector, $H$ be a function from $\Mbb{R}^M$ to $\Mbb{R}$ such that $\mathbb{E}\left[ |H\left( \mathbf{Y}_{\Bs{\theta}}\right) |
\right] <\infty$, $R\left( \Bs{\theta}\right)= \Mbb{E}[H\left( \mathbf{Y}_{\Bs{\theta}}\right)]$, and $\Bs{\theta}_0$ be a possible value of the parameter $\Bs{\theta}$.
LRM requires $\mathbf{Y}_{\Bs{\theta}}$
to have a density function $f_{\Bs{\theta}}$ that can be
differentiated with respect to $\Bs{\theta}$. In that case, the expected performance satisfies
\begin{equation*}
R\left( \Bs{\theta}\right) =\mathbb{E}\left[ H\left( \mathbf{Y}_{%
\Bs{\theta}}\right) \right] =\int_{(0,\infty )^{M}}H\left( \mathbf{y}%
\right) f_{\Bs{\theta}}\left( \mathbf{y}\right) \text{d}\mathbf{y}\text{%
.}
\end{equation*}%
We assume that there exists some non-random neighbourhood of $\Bs{\theta}%
_{0}$, $\mathcal{V}_{\Bs{\theta}_{0}}$, such that
\Bit
\item for each $\Bs{\theta}\in \mathcal{V}_{\Bs{\theta}_{0}}$, $%
\mathbb{E}\left[ |H\left( \mathbf{Y}_{\Bs{\theta}}\right) |\right]
<\infty $,
\item for almost all $\mathbf{y}\in (0,\infty )^{M}$, $\partial f_{%
\Bs{\theta}}\left( \mathbf{y}\right) /\partial \Bs{\theta}$ exists
for all $\Bs{\theta}\in \mathcal{V}_{\Bs{\theta}_{0}}$,
\item there is an integrable function $\Psi :\mathbb{R}^{M}\rightarrow 
\mathbb{R}$ such that $|H\left( \mathbf{y}\right)
|\sup_{j=1,...,M}\left\vert \partial f_{\Bs{\theta}}\left( \mathbf{y}%
\right) /\partial \theta _{j}\right\vert \leq \Psi \left( \mathbf{y}\right) $
for all $\Bs{\theta}\in \mathcal{V}_{\Bs{\theta}_{0}}$ and almost
every $\mathbf{y}\in (0,\infty )^{M}$.
\Eit
Then, by the dominated convergence theorem,
differentiation and integration can be interchanged, giving
\begin{equation}
\left. \frac{\partial R\left( \Bs{\theta}\right) }{\partial \mathbf{%
\theta }}\right\vert _{\Bs{\theta}=\Bs{\theta}_{0}}=\int_{(0,%
\infty )^{M}}H\left( \mathbf{y}\right) \frac{\partial f_{\Bs{\theta}%
}\left( \mathbf{y}\right) }{\partial \Bs{\theta}}\text{d}\mathbf{y}=%
\mathbb{E}\left[ H\left( \mathbf{Y}_{\Bs{\theta}_{0}}\right) \left. 
\frac{\partial \log f_{\Bs{\theta}}\left( \mathbf{Y}_{\Bs{\theta}%
}\right) }{\partial \Bs{\theta}}\right\vert _{\Bs{\theta}=\mathbf{%
\theta }_{0}}\right], \label{Eq_ExpressionDerivativeLRM}
\end{equation}
where, provided that, for all $B\subset \mathcal{I}$, $\partial \left( \partial
^{|B|}V_{\Bs{\theta}}(\mathbf{y})/\partial \mathbf{y}_{B}\right)
/\partial \Bs{\theta}$ exists for all $\Bs{\theta}\in \mathcal{V}_{%
\Bs{\theta}_{0}}$,
\begin{equation*}
\frac{\partial \log f_{\Bs{\theta}}\left( \mathbf{y}\right) }{\partial 
\Bs{\theta}}=-\frac{\partial V_{\Bs{\theta}}(\mathbf{y})}{\partial 
\Bs{\theta}}+\frac{\exp \left( -V_{\Bs{\theta}}(\mathbf{y})\right) 
}{f_{\Bs{\theta}}\left( \mathbf{y}\right) }\sum_{\pi \in \Pi }\left(
-1\right) ^{|\pi |}\sum_{B\in \pi }\frac{\partial }{\partial \Bs{\theta}%
}\frac{\partial ^{|B|}}{\partial \mathbf{y}_{B}}V_{\Bs{\theta}}(\mathbf{%
y})\prod\limits_{B^{\prime }\in \pi ,B^{\prime }\neq B}\frac{\partial
^{|B^{\prime }|}}{\partial \mathbf{y}_{B}}V_{\Bs{\theta}}(\mathbf{y}).
\end{equation*}
Then LRM consists in computing $\partial
R\left( \Bs{\theta}\right) /\partial \Bs{\theta} \vert_{\Bs{\theta}=\Bs{\theta}_0}$ by estimating the expectation in the right-hand side of \eqref{Eq_ExpressionDerivativeLRM} by Monte Carlo. Note that the above assumptions are quite
usual in the theory of maximum likelihood estimation where $\Bs{\theta}%
\mapsto \partial \log f_{\Bs{\theta}}\left( \mathbf{y}\right) /\partial 
\Bs{\theta}$ is called the score function.

\subsubsection{The case of the Brown--Resnick random field}

We now focus on the case of the Brown--Resnick field.
Let $W_{\Bs{\theta}}$ be a centred Gaussian random field on $\mathbb{R}%
^{d}$ with stationary increments and with semivariogram $\gamma _{\Bs{
\theta }}$, and let us define
$Z_{\Bs{\theta}}(\mathbf{x})=\exp \left( W_{\Bs{\theta}}(\mathbf{x}
)-\mathrm{Var}(W_{\Bs{\theta}}(\mathbf{x}))/2\right)$, $\Mb{x} \in \Mbb{R}^d$, where $\mathrm{Var}$ denotes the variance. Then the field $Y_{\Bs{
\theta }}$ defined by \eqref{Eq_Spectral_Representation_Stochastic_Processes}
with that $Z_{\Bs{\theta}}$ is referred to as the Brown--Resnick
random field associated with the semivariogram $\gamma _{\Bs{\theta}}$ %
\citep{brown1977extreme, kabluchko2009stationary}. It is stationary\footnote{%
Throughout the paper, stationarity refers to strict stationarity.} and its
distribution only depends on the variogram 
\citep[][Theorem 2 and
Proposition 11, respectively]{kabluchko2009stationary}. The case where $W_{%
\Bs{\theta}}$ is a fractional Brownian motion leads to the commonly
used semivariogram 
\begin{equation}
\gamma _{\Bs{\theta}}(\mathbf{x}_{1},\mathbf{x}_{2})=\left( \Vert 
\mathbf{x}_{1}-\mathbf{x}_{2}\Vert /\kappa \right) ^{\psi },\quad \mathbf{x%
}_{1},\mathbf{x}_{2}\in \mathbb{R}^{d},  \label{Eq_Power_Variogram}
\end{equation}%
where $\kappa >0$ and $\psi \in (0,2)$ are the range and the smoothness
parameters, respectively, and $\Bs{\theta}=\left( \psi ,\kappa
\right) ^{\prime }$.
Generally, $\psi=2$ is allowed in \eqref{Eq_Power_Variogram}, but we exclude that value here as it makes the multivariate density unavailable under a closed-form expression for $M > d +1$, and thus LRM inapplicable in this case.

For any semivariogram $\gamma_{\Bs{\theta}}$, let $\lambda _{\Bs{\theta}}(\mathbf{x}_{i},\mathbf{x}_{j})=\left(
\gamma _{\Bs{\theta}}(\mathbf{x}_{i},\mathbf{x}_{j})/2\right) ^{1/2}$, $\mathbf{x}_{i},\mathbf{x}_{j} \in \Mbb{R}^d$, $i, j=1, \ldots, M$. Let also $\lambda _{\Bs{\theta}}(\mathbf{x}_{i},\mathbf{x}_{-i})=\left(
\lambda _{\Bs{\theta}}(\mathbf{x}_{i},\mathbf{x}_{j})\right)' _{j\neq i}$
, $\log \left( \mathbf{y}_{-i}/y_{i}\right) =\left( \log \left(
y_{j}/y_{i}\right) \right)_{j\neq i}'$, and $\Omega _{\Bs{\theta}}^{(i)}$
be the matrix with $(j,m)$-th entry 
$2[\lambda _{\Bs{\theta}}^{2}(\mathbf{x}_{i},\mathbf{x}_{j})+\lambda _{
\Bs{\theta}}^{2}(\mathbf{x}_{i},\mathbf{x}_{m})-\lambda _{\mathbf{
\theta }}^{2}(\mathbf{x}_{j},\mathbf{x}_{m})]$, $j,m\neq i$. We denote, for $p\in \mathbb{N}_{\ast }$, by $\Phi _{p}\left( \cdot ;\Omega
\right) $ and $\varphi _{p}\left( \cdot ;\Omega \right) $ the $p$-dimensional Gaussian distribution and density functions with
covariance matrix $\Omega $, respectively. Provided that the matrices $\Omega _{\Bs{\theta}}^{(i)}$, $i=1, \ldots, M$, are positive definite, then the exponent measure function of the
Brown--Resnick random field is given by \citep[e.g.,][]{huser2013composite}
\begin{equation}
\label{Eq_ExponentMeasBR}
V_{\Bs{\theta}}(\mathbf{y})=\sum_{i=1}^{M}y_{i}^{-1}\phi _{i}\left( \mathbf{y},\Bs{\theta}\right),
\end{equation}
where 
\begin{equation*}
\phi _{i}\left( \mathbf{y},\Bs{\theta}\right) =\Phi _{M-1}\left(
\lambda _{\Bs{\theta}}(\mathbf{x}_{i},\mathbf{x}_{-i})+\log \left( 
\mathbf{y}_{-i}/y_{i}\right) ;\Omega _{\Bs{\theta}}^{(i)}\right).
\end{equation*}
Combined with \eqref{Eq_Density_MaxStable}, \eqref{Eq_ExponentMeasBR} shows the existence of a closed-form expression for the density $f_{\Bs{\theta}}$.

In the following theorem, we provide convenient conditions ensuring the applicability of LRM for estimating $\partial
R\left( \Bs{\theta}\right) /\partial \Bs{\theta} \vert_{\Bs{\theta}=\Bs{\theta}_0}$.
\begin{Th}
\label{Th_Interchange_LRM} 
For a non-random neighbourhood of $\Bs{\theta}_{0}$, $\mathcal{V}_{%
\Bs{\theta}_{0}}$, let 
\begin{equation}
B_{\mathcal{V}_{\Bs{\theta}_{0}}}=\inf_{i=1,\ldots,M}\inf_{\Bs{\theta}%
\in \mathcal{V}_{\Bs{\theta}_{0}}}\inf_{\mathbf{y}\in (0,\infty )^{M}}
\left \{ \phi _{i}\left( \mathbf{y},\Bs{\theta}\right) -\phi _{i}\left( 
\mathbf{y},\Bs{\theta}_{0}\right) \right \}.  \label{Eq_B}
\end{equation}%
Assume that\ there exist a non-random neighbourhood $\mathcal{V}_{\mathbf{%
\theta }_{0}}$ and a constant $\alpha >0$ such that 
\begin{equation}
\mathbb{E}\left[ |H\left( \mathbf{Y}_{\Bs{\theta}_{0}}\right) |\left( 1+%
\frac{1}{M}\sum_{i=1}^{M}Y_{\Bs{\theta}_{0},i}^{-1}\right) \left(
\sum_{i=1}^{M}Y_{\Bs{\theta}_{0},i}^{-\alpha }\right) \exp \left( -B_{%
\mathcal{V}_{\Bs{\theta}_{0}}}\sum_{i=1}^{M}Y_{\Bs{\theta}%
_{0},i}^{-1}\right) \left\Vert \mathbf{Y}_{\Bs{\theta}_{0}}\right\Vert
^{\alpha }\right] <\infty,   \label{Eq_Gen_Bound}
\end{equation}
where $Y_{\Bs{\theta}_{0},i}=Y_{\Bs{\theta}_{0}}\left( \mathbf{x}%
_{i}\right) $. Assume moreover that 
\begin{equation}
\sup_{k=1,\ldots,L}\sup_{1\leq i,j\leq m}\sup_{\Bs{\theta}\in \mathcal{V}_{%
\Bs{\theta}_{0}}}\left\vert \frac{\partial \lambda _{\Bs{\theta}}(%
\mathbf{x}_{i},\mathbf{x}_{j})}{\partial \Bs{\theta}_{k}}\right\vert
<\infty .  \label{Eq_bound_der_lambda}
\end{equation}
Then
\begin{equation}
\label{Eq_FinalResultTheorem}
\left. \frac{\partial R\left( \Bs{\theta}\right) }{\partial \Bs{%
\theta }}\right\vert _{\Bs{\theta}=\Bs{\theta}_{0}}=\mathbb{E}%
\left[ H\left( \mathbf{Y}_{\Bs{\theta}_{0}}\right) \left. \frac{%
\partial \log f_{\Bs{\theta}}\left( \mathbf{Y}_{\Bs{\theta}%
}\right) }{\partial \Bs{\theta}}\right\vert _{\Bs{\theta}=\Bs{
\theta }_{0}}\right] .
\end{equation}
\end{Th}
These technical but tractable conditions are pretty natural ones to guarantee the validity of the interchange between differentiation and integration.

Realizations of the term inside the expectation in \eqref{Eq_FinalResultTheorem} can be obtained since the Brown--Resnick field can be simulated at a finite number of sites and the score can be computed. The simulation methods for the Brown--Resnick field are either exact \citep[e.g.,][]{dombry2016exact, oesting2018exact} or approximate \citep[][Theorem 4]{schlather2002models}. 

Nonetheless, the computation of the density (and thus of the score) may be challenging
as the number of terms in the sum equals the $M$-th Bell number, which
grows super-exponentially in the dimension $M$. To circumvent this issue, a solution is to approximate \eqref{Eq_Density_MaxStable} by Monte Carlo
simulation; i.e., for $N\geq 1$,
\begin{equation*}
\hat{f}_{\Bs{\theta}}\left( \mathbf{y}\right) =\exp \left( -V_{\mathbf{%
\theta }}(\mathbf{y})\right) \frac{1}{N}\sum_{i=1}^{N}\prod\limits_{B\in \pi
_{i}}\left( -\frac{\partial ^{|B|}}{\partial \mathbf{y}_{B}}V_{\mathbf{%
\theta }}(\mathbf{y})\right),
\end{equation*}
where the partitions $\pi _{1},\ldots,\pi _{N}$ form an ergodic sequence with
stationary distribution given by%
\begin{equation}
g_{\Bs{\theta}}(\pi |\mathbf{y})=\frac{\prod\limits_{B\in \pi }\left( -%
\frac{\partial ^{|B|}}{\partial \mathbf{y}_{B}}V_{\Bs{\theta}}(\mathbf{y%
})\right) }{\sum_{\pi \in \Pi }\prod\limits_{B\in \pi }\left( -\frac{%
\partial ^{|B|}}{\partial \mathbf{y}_{B}}V_{\Bs{\theta}}(\mathbf{y}%
)\right) }\propto \prod\limits_{B\in \pi }\left( -\frac{\partial ^{|B|}}{%
\partial \mathbf{y}_{B}}V_{\Bs{\theta}}(\mathbf{y})\right) .
\label{Eq_Cond_Dens_Partition}
\end{equation}%
\cite{dombry2013conditional} design a Gibbs sampler to generate approximate
simulations $\pi _{1},\ldots,\pi _{N}$ without explicitly computing the
constant factor in the denominator of \eqref{Eq_Cond_Dens_Partition}. We
refer to that paper for more details about the practical implementation. Theoretically, the ergodicity of the resulting Markov
chain implies that the precision of the approximation is arbitrarily high as 
$N\rightarrow \infty $. In practice, the number of iterations of the Gibbs
sampler, $N$, is typically much smaller than the cardinality of $\Pi$, but the approximation is reasonably good even for moderate values of $N$ because generally only a few partitions $\pi \in \Pi $ are compatible
with the data; see, e.g., \cite{huser2019full}. For the Brown--Resnick field associated with the
semivariogram \eqref{Eq_Power_Variogram}, they
conclude that the Gibbs sampler converges quickly and that about $10\times M$
iterations are enough for the algorithm to converge for a large number of parameter configurations.

Using the Monte Carlo based idea, the approximation of the score function is then given by
\begin{equation*}
\frac{\partial \log \hat{f}_{\Bs{\theta}}\left( \mathbf{y}\right) }{%
\partial \Bs{\theta}}=-\frac{\partial V_{\Bs{\theta}}(\mathbf{y})}{%
\partial \Bs{\theta}}+\frac{\exp \left( -V_{\Bs{\theta}}(\mathbf{y}%
)\right) }{\hat{f}_{\Bs{\theta}}\left( \mathbf{y}\right) }\frac{1}{N}%
\sum_{i=1}^{N}\left( -1\right) ^{|\pi _{i}|}\sum_{B\in \pi _{i}}\frac{%
\partial }{\partial \Bs{\theta}}\frac{\partial ^{|B|}}{\partial \mathbf{%
y}_{B}}V_{\Bs{\theta}}(\mathbf{y})\prod\limits_{B^{\prime }\in \pi
_{i},B^{\prime }\neq B}\frac{\partial ^{|B^{\prime }|}}{\partial \mathbf{y}%
_{B}}V_{\Bs{\theta}}(\mathbf{y}).
\end{equation*}%
Analytical expressions of $\partial (\partial ^{|B|}V_{\Bs{\theta}}(%
\mathbf{y})/\partial \mathbf{y}_{B})/\partial \Bs{\theta}$ are known
for the Brown--Resnick random field \citep[e.g.,][Section B.4]{dombry2017asymp}. 

\subsection{Infinitesimal perturbation analysis}

\subsubsection{General methodology}

We first introduce IPA to the context of expected performances based on max-stable fields.
Let $\mathbf{Y}_{\Bs{\theta}}$ be a simple max-stable random vector, $H$ be a function from $\Mbb{R}^M$ to $\Mbb{R}$ such that $\mathbb{E}\left[ |H\left( \mathbf{Y}_{\Bs{\theta}}\right) |
\right] <\infty$, $R\left( \Bs{\theta}\right)= \Mbb{E}[H\left( \mathbf{Y}_{\Bs{\theta}}\right)]$, and $\Bs{\theta}_0$ be a possible value of the parameter $\Bs{\theta}$. IPA requires the random performance $H\left( \mathbf{Y}_{\Bs{\theta}}\right)$  to be differentiable with respect to $\Bs{\theta}$. However, the density $f_{
\Bs{\theta}}$ does not need to be explicit, making IPA a potential solution when LRM is invalid. When both LRM and IPA can be applied, IPA may be preferable, although the optimal choice depends on the specific form of the random performance.
For stochastic
processes and random fields, IPA is also called pathwise derivative
estimation because it uses differentiation of sample path functionals. 

The essence of IPA is to assess the derivative of interest using 
\begin{equation}
\left. \frac{\partial R\left( \Bs{\theta}\right) }{\partial \Bs{%
\theta }}\right\vert _{\Bs{\theta}=\Bs{\theta}_{0}}=\mathbb{E}%
\left[ \left. \frac{\partial H\left( \mathbf{Y}_{\Bs{\theta}}\right) }{%
\partial \Bs{\theta}}\right\vert _{\Bs{\theta}=\Bs{\theta}%
_{0}}\right] ,  \label{Eq_Interchange}
\end{equation}%
provided that the derivative and the expectation can be interchanged. The derivative can then be computed through estimation of the right-hand side of \eqref{Eq_Interchange} by Monte Carlo. A
sufficient condition for the interchange to hold in the general case is given, e.g., in \cite%
{asmussen2010}, Chapter VII, Proposition 2.3.

\begin{Prop}
\label{Prop_Interchange_Asmussen} Assume that $\Bs{\theta}\mapsto
H\left( \mathbf{Y}_{\Bs{\theta}}\right) $ is an almost surely (a.s.)
differentiable function at $\Bs{\theta}_{0}$ and that a.s. $\Bs{
\theta }\mapsto H\left( \mathbf{Y}_{\Bs{\theta}}\right) $ satisfies the
Lipschitz condition 
\begin{equation*}
\left\vert H\left( \mathbf{Y}_{\Bs{\theta}_{1}}\right) -H\left( \mathbf{%
Y}_{\Bs{\theta}_{2}}\right) \right\vert \leq \left\Vert \Bs{\theta}%
_{1}-\Bs{\theta}_{2}\right\Vert B_{\Bs{\theta}_{0}}
\end{equation*}%
for $\Bs{\theta}_{1}$, $\Bs{\theta}_{2}$ in a non-random
neighbourhood of $\Bs{\theta}_{0}$, where $\mathbb{E}\left[ B_{\mathbf{%
\theta }_{0}}\right] <\infty $. Then \eqref{Eq_Interchange} holds.
\end{Prop}

If $g:I\rightarrow \mathbb{R}$ is differentiable on an open set $I\subset 
\mathbb{R}^{d}$, and satisfies $\left\Vert \partial g(\mathbf{x})/\partial 
\mathbf{x}\right\Vert \leq K$ for all $\mathbf{x}$ in $I$, then $g$ is
Lipschitz continuous with Lipschitz constant at most $K$ over $I$.
Therefore, we immediately deduce that, if there exists a random variable $B_{%
\Bs{\theta}_{0}}$ satisfying $\mathbb{E}\left[ B_{\Bs{\theta}_{0}}%
\right] <\infty $ and such that a.s. 
\begin{equation*}
\sup_{\Bs{\theta}\in \mathcal{V}_{\Bs{\theta}_{0}}}\left\Vert 
\frac{\partial H\left( \mathbf{Y}_{\Bs{\theta}}\right) }{\partial 
\Bs{\theta}}\right\Vert \leq B_{\Bs{\theta}_{0}},
\end{equation*}%
where $\mathcal{V}_{\Bs{\theta}_{0}}$ is a non-random neighbourhood of $%
\Bs{\theta}_{0}$, then \eqref{Eq_Interchange} holds.

\subsubsection{The case of the Smith random field}

We illustrate IPA in the case where the random performance is based on the Smith random field \citep{smith1990max}, for which a closed-form expression for the density is only known
when $M$, the number of stations in $\mathbb{R}^{d}$, is smaller or equal
than $d+1$ \citep{genton2011likelihood}. Letting $(U_{i},\mathbf{C}_{i})_{i\geq 1}$
be the points of a Poisson point process on $(0,\infty )\times \mathbb{R}^{d}
$ with intensity function $u^{-2} \mathrm{d}u \times \mathrm{d}\mathbf{c}$, the Smith random field with covariance matrix $%
\Sigma =\left( \sigma _{ij}\right) _{ij}$  is defined
by 
\begin{equation}
Y_{\Sigma}(\mathbf{x})=\bigvee_{i=1}^{\infty }U_{i}\varphi _{M}(\mathbf{x}-\mathbf{C}%
_{i},\Sigma ),\qquad \mathbf{x}\in \mathbb{R}^{d},
\label{Eq_M3_Representation}
\end{equation}%
i.e., by taking $Z_{i}\left( \mathbf{x}\right) =\varphi _{M}(\mathbf{x}-%
\mathbf{C}_{i},\Sigma )$ in 
\eqref{Eq_Spectral_Representation_Stochastic_Processes}. It is stationary and simple max-stable, and corresponds to the Brown--Resnick field associated with the
semivariogram $\gamma (\mathbf{x})=\mathbf{x}^{\prime }\Sigma ^{-1}\mathbf{x}%
/2$, $\mathbf{x}\in \mathbb{R}^{d}$ \citep[e.g.,][]{huser2013composite}. Such a semivariogram leads to the
impossibility of characterizing the density for $M > d + 1$, as mentioned above. As was the case of $\Bs{
\theta }$ for the Brown--Resnick random field, $\Sigma $ completely
characterizes the dependence structure of the Smith field. For ease,
the vector $\Bs{\theta}$ is replaced by the positive definite
matrix $\Sigma$ in the following. The derivative of an expected performance $R\left( \Sigma \right) $
written as in \eqref{Eq_GenQuantityToDerive} with respect to $\Sigma $ at
some positive definite matrix $\Sigma _{0}$ is defined by %
\citep[e.g.,][]{dwyer1967some} 
\begin{equation}
\left. \frac{\partial R\left( \Sigma \right) }{\partial \Sigma }\right\vert
_{\Sigma =\Sigma _{0}}=\left( \left. \frac{\partial R\left( \Sigma \right) }{%
\partial \sigma _{ij}}\right\vert _{\Sigma =\Sigma _{0}}\right) _{ij}.
\end{equation}%
Note that results concerning differentiation with respect to a scalar or a
vector also hold in the case of differentiation with respect to a matrix.

Assume now that $\mathbf{y}\mapsto H\left( \mathbf{y}\right) $ is
differentiable. The differentiability of the function $\Sigma \mapsto 
\mathbf{Y}_{\Sigma }$ in a neighbourhood of $\Sigma _{0}$ will be shown in
the proof of Theorem \ref{Lem_Derivability} (Appendix \ref%
{Sec_Proof_Th_MainResultIPA}). Then the chain rule gives
\begin{equation}
\frac{\partial H\left( \mathbf{Y}_{\Sigma }\right) }{\partial \Sigma }%
=\sum_{j=1}^{M}\frac{\partial H\left( \mathbf{Y}_{\Sigma }\right) }{\partial
y_{j}}\frac{\partial Y_{\Sigma }(\mathbf{x}_{j})}{\partial \Sigma }.
\label{Eq_Chain_Rule}
\end{equation}%
We shall prove (Theorem \ref{As_M1} in Appendix \ref%
{Sec_Proof_Th_MainResultIPA}) that there exists some non-random
neighbourhood of $\Sigma _{0}$, $\mathcal{V}_{\Sigma _{0}}$, such that, for
any $q>1$, there exists a random variable $C_{\Sigma _{0}}(\mathbf{x},q)$
satisfying a.s. 
\begin{equation*}
\sup_{\Sigma \in \mathcal{V}_{\Sigma _{0}}}\left\Vert \frac{\partial \log
Y_{\Sigma }(\mathbf{x})}{\partial \Sigma }\right\Vert ^{q}\leq C_{\Sigma
_{0}}(\mathbf{x},q)\quad \mbox{and}\quad \mathbb{E}\left[ C_{\Sigma _{0}}(%
\mathbf{x},q)\right] <\infty .
\end{equation*}%
This technical outcome will allow us to derive our main result.
\begin{Th}
\label{Th_MainResultIPA} Assume that $\mathbf{y}\mapsto H\left( \mathbf{y}%
\right) $ is a differentiable function and that there exists $p>1$ such that 
\begin{equation}
\sup_{j=1,...,M}\mathbb{E}\left[ \sup_{\Sigma \in \mathcal{V}_{\Sigma
_{0}}}\left\vert Y_{\Sigma }(\mathbf{x}_{j})\frac{\partial H\left( \mathbf{Y}%
_{\Sigma }\right) }{\partial y_{j}}\right\vert ^{p}\right] <\infty ,
\label{Eq_Condition_Main_Theorem}
\end{equation}%
where $\mathcal{V}_{\Sigma _{0}}$ is a non-random neighbourhood of $\Sigma
_{0}$. Then
\Beq
\label{Eq_MainResultIPASmith}
\left. \frac{\partial R\left( \Sigma \right) }{\partial \Sigma }\right\vert
_{\Sigma =\Sigma _{0}}=\mathbb{E}
\left[ \left. \frac{\partial H\left( \mathbf{Y}_{\Sigma}\right) }{
\partial \Sigma}\right\vert_{\Sigma=\Sigma
_{0}} \right].
\Eeq
\end{Th}

This non-trivial theorem provides a sufficient condition to use IPA to
compute the derivatives of $R\left( \Sigma \right) $. This condition is much
more tractable and easier to check than that in Proposition \ref%
{Prop_Interchange_Asmussen}. The simplification stems from the fact that we
take care in Theorems \ref{Lem_Derivability} and \ref{As_M1} of the
intricate term of \eqref{Eq_Chain_Rule}, $\partial Y_{\Sigma }(\mathbf{x}%
_{j})/\partial \Sigma $, which involves the sample path properties with
respect to differentiation of the Smith random field. Theorems \ref%
{Lem_Derivability} and \ref{As_M1} are delicate to establish precisely due
to the inherent structure of max-stable fields.

In practice we have to simulate realizations of the random matrix inside the expectation in \eqref{Eq_MainResultIPASmith}, which can be done by simulating the Smith field and using \eqref{Eq_Chain_Rule}.   
The expression of $\partial Y_{\Sigma }(\mathbf{x}%
_{j})/\partial \Sigma $ appearing in \eqref{Eq_Chain_Rule} is given in %
\eqref{Eq_Deriv_Realis_Sig} (Appendix \ref{Sec_Proof_Th_MainResultIPA}). As the Brown--Resnick field, the Smith field can be simulated exactly \citep[e.g.,][]{oesting2018exact, dombry2016exact} or approximately \citep[][Theorem 4]{schlather2002models}; the latter approach is very accurate.


\section{Application to risk assessment}
\label{Sec_RiskAssessment}

We now focus on one framework (among several others) where the results of
Section \ref{Subsec_EstimationSensMaxstable} are useful, which is the
context of risk and dependence measures. After detailing the link with %
\eqref{Eq_GenQuantityToDerive}, we consider a dependence measure which is
particularly suited to insurance of damage triggered by extreme wind speeds.
We show that the conditions of Theorems \ref{Th_Interchange_LRM} and \ref%
{Th_MainResultIPA} hold in that case and confirm through a simulation study
that both LRM and IPA perform very well. Finally, we consider concrete data
and show that the sensitivity of our dependence measure can be very high,
highlighting the practical importance of studying the sensitivity of
functions of max-stable fields in the context of risk assessment.

\subsection{Risk measures based on max-stable fields}

Here we show that the quantity defined in \eqref{Eq_GenQuantityToDerive}
encompasses many risk and dependence measures, and our focus is mainly on
actuarial applications. A univariate risk measure is a mapping from a set of
random variables to the real numbers. A dependence measure summarises the
strength of dependence between several elements of such a set of random
variables. In finance, these random variables often represent portfolio
returns, and in an insurance context, they might be the claims associated
with insurance policies. When claims are triggered by environmental events,
a possible model for the insured cost field is 
\citep[][Section
2.3]{koch2017spatial} 
\begin{equation}
C(\mathbf{x})=E(\mathbf{x}) D_{\mathbf{x}}(X(\mathbf{x})),\qquad \mathbf{x}%
\in \mathbb{R}^{d},  \label{Eq_Cost_Field}
\end{equation}%
where $E$ is the (deterministic) insured exposure (i.e., insured value) field, $D_{\mathbf{x}}$ the damage
function at site $\mathbf{x}$ and $X$ the random field of the
environmental variable generating risk. The application of the damage
function $D_{\mathbf{x}}$ to $X(\mathbf{x})$ yields the insured cost ratio (i.e., the insured cost divided by the insured value) at site $\mathbf{x}$, which, multiplied by the insured exposure, gives the
corresponding insured cost.

Among the risk measures that can be written as in %
\eqref{Eq_GenQuantityToDerive}, one can find many sophisticated examples,
e.g., in insurance/reinsurance pricing or regulation. For $j=1,\dots ,M$,
let $C_{j}$ denote the claim of an insurance company at $%
\mathbf{x}_{j}$ and assume that $C_{j}$ can be written as a function of the
max-stable field (e.g., as in the right-hand side of \eqref{Eq_Cost_Field}).
Premium loadings that are proportional to specific moments of the sum of the 
$C_{j}$ at two or more sites constitute excellent examples in insurance
pricing. In reinsurance, the premium is sometimes based on order statistics
of the claims, as in the case of the ``excédent du coût moyen relatif'' (ECOMOR) or largest claims reinsurance (LCR) treaties. Let us
consider, e.g., $M\geq 3$ sites and assume that each of those is associated
with an insurance policy whose corresponding claim is $C_{j}$ as above. We
consider the ordered values of those claims, $C^{(1:M)}\geq C^{(2:M)}\geq
\cdots \geq C^{(M:M)}$. For instance, the risk measure $\mathbb{E}%
[(C^{(1:M)}-C^{(3:M)})+(C^{(2:M)}-C^{(3:M)})]$ would be involved in the
pricing of an ECOMOR reinsurance treaty having the third largest claim as
priority. These quantities can be written under the form \eqref{Eq_GenQuantityToDerive}
but behave in a non-linear way and do not have any analytical expression. A valuable example of dependence measure, which will be considered until the end, is presented in the next section.

\subsubsection{A specific dependence measure for wind speed}

\label{Subsec_Depend_Wind}

We present in this section a dependence measure for the costs due to high
wind speeds that can be written as in \eqref{Eq_GenQuantityToDerive}.
We
model the cost by \eqref{Eq_Cost_Field} in the case $d=2$, where the field of wind
speed extremes, $X_{\Bs{\theta}}$, is assumed to be max-stable. Moreover, for any site $\mathbf{x}\in \mathbb{R}^{2}$, we choose as
damage function $D_{\mathbf{x}}(x)=(x/u)^{\beta (\mathbf{x})}$, $x \leq u$, where $u>0$ and $\beta (%
\mathbf{x})\in \mathbb{N}_{\ast }$, which is utterly appropriate in the case
of wind. Indeed, as the force exerted by the wind and the corresponding rate
of work are proportional to the second and third powers of wind speed,
respectively, the total cost for a specific structure is expected to
increase as the square or the cube of the maximum wind speed. For studies
supporting the use of the square, see, e.g., 
\citet[][Equations (4.7.1),
(8.1.1) and (8.1.8)]{simiu1996wind}, and for the cube, see 
\citet[][Chapter
2, p.7]{lamb1991historic}, \cite{emanuel2005increasing} and \cite%
{kantha2008tropical}. However, in the case of insured costs, several authors
have recently found power-laws with much higher exponents; e.g., \cite%
{prahl2012applying} obtained exponents spanning from $8$ to $12$ for insured
losses on residential buildings in Germany. Indeed, the presence of a
deductible in the insurance contract increases the exponent from $2$ or $3$
to a larger value depending on the deductible %
\citep[e.g.,][]{prahl2015comparison, koch2018spatialpowers}. The level $u$ corresponds to that value of wind speed which triggers an insured cost ratio equal to unity, and is thus assumed to be larger than the wind speeds observed in practice. For a more
detailed review of wind damage functions, see, e.g., \cite%
{koch2018spatialpowers}.

We consider two sites ($M=2$), $\mathbf{x}_{1},\mathbf{x}_{2}\in \mathbb{R}%
^{2}$. Recall that the link between $X_{\Bs{\theta}}$ and the associated simple max-stable field $Y_{\Bs{\theta}}$ is given by \eqref{Eq_Link_Maxstb_Simple_Maxstab1}, and let $\eta _{i}=\eta (\mathbf{x}_{i}),\tau _{i}=\tau (\mathbf{x}%
_{i}),\xi _{i}=\xi (\mathbf{x}_{i})$ and $\beta _{i}=\beta (\mathbf{x}_{i})$, $i=1,2$, such that $\beta _{i}\xi _{i}<1/2$. The shape parameter for wind speed maxima is often slightly negative, implying that the event $X_{\Bs{\theta}}(\mathbf{x}_{i}) < 0$ occurs with probably not exactly $0$ (although extremely close owing to the values of the location and scale parameters). Anyway, this is not a problem as $\beta_{i}\in \mathbb{N}_{\ast }$. 
Moreover, we take $Y_{\Bs{\theta}}$ to be the Brown--Resnick or the Smith random field and we consider as dependence measure the correlation between the costs due to extreme winds at the two sites, i.e., 
\begin{equation}
R(\Bs{\theta})=\mathrm{Corr}\left( C(\Mb{x}_1), C(\Mb{x}_2) \right)=\mathrm{Corr}\left( X_{\Bs{\theta}}^{\beta _{1}}(%
\mathbf{x}_{1}),X_{\Bs{\theta}}^{\beta _{2}}(\mathbf{x}_{2})\right),
\label{Eq_Ex_Relevant_Dependence_Measure}
\end{equation}%
where $\Bs{\theta}$ must be replaced by $\Sigma$ when the Smith field is considered. The condition $\beta_{i}\xi _{i}<1/2$, $i=1,2$, ensures the existence of
the correlation in \eqref{Eq_Ex_Relevant_Dependence_Measure}. Correlation is
used a lot by practitioners in the finance/insurance industry and the
measure $R(\Bs{\theta})$ is of practical interest for any
insurance/reinsurance company handling the risk of damage caused by extreme
wind speeds; among others, it provides insight about potential spatial diversification.
For details, see \cite{koch2018spatialpowers} where %
\eqref{Eq_Ex_Relevant_Dependence_Measure} is thoroughly studied.

We now explain why the measure in \eqref{Eq_Ex_Relevant_Dependence_Measure}
is of the form \eqref{Eq_GenQuantityToDerive} with $M=2$. It is well-known
that, if $Y$ is a random variable following the standard Fr\'{e}chet
distribution, then $%
\Mbb{E}[Y^{\tilde{\beta}}]=\Gamma (1-\tilde{\beta})$ for any $\tilde{\beta}<1$, where $\Gamma $
denotes the gamma function. We now also assume that $\xi _{i}\neq 0$ for $%
i=1,2$. Consequently, using \eqref{Eq_Link_Maxstb_Simple_Maxstab1} and the
binomial theorem, we obtain, for $i=1,2$, 
\begin{equation}
\Mbb{E}\left[ X_{\Bs{\theta}}^{\beta _{i}}(\mathbf{x}_{i})\right]
=C_{\beta _{i},\eta _{i},\tau _{i},\xi _{i}},  \label{Eq_Exp_Zbeta}
\end{equation}%
where 
\begin{equation*}
C_{\beta _{i},\eta _{i},\tau _{i},\xi _{i}}=\sum_{k=0}^{\beta _{i}}{\binom{%
\beta _{i}}{k}}\left( \eta _{i}-\frac{\tau _{i}}{\xi _{i}}\right) ^{k}\left( 
\frac{\tau _{i}}{\xi _{i}}\right) ^{\beta _{i}-k}\Gamma (1-[\beta _{i}-k]\xi
_{i}).
\end{equation*}%
Moreover, Corollary 1 in \cite{koch2018spatialpowers} gives 
\begin{equation}
\Mr{Var}\left[ X_{\Bs{\theta}}^{\beta _{i}}(\mathbf{x}_{i})\right]
=D_{\beta _{i},\eta _{i},\tau _{i},\xi _{i}},  \label{Eq_Variance_Zbeta}
\end{equation}%
where 
\begin{equation}
D_{\beta _{i},\eta _{i},\tau _{i},\xi _{i}}=\sum_{k_{1}=0}^{\beta
_{i}}\sum_{k_{2}=0}^{\beta _{i}}B_{k_{1},k_{2},\beta _{i},\eta _{i},\tau
_{i},\xi _{i}}\left\{ \Gamma (1-\xi _{i}[2\beta _{i}-k_{1}-k_{2}])-\Gamma
(1-[\beta _{i}-k_{1}]\xi _{i})\Gamma (1-[\beta _{i}-k_{2}]\xi _{i})\right\} ,
\label{Eq_DefCoeffD}
\end{equation}%
with 
\begin{equation*}
B_{k_{1},k_{2},\beta _{i},\eta _{i},\tau _{i},\xi _{i}}={\binom{\beta _{i}}{%
k_{1}}}{\binom{\beta _{i}}{k_{2}}}\left( \eta _{i}-\frac{\tau _{i}}{\xi _{i}}%
\right) ^{k_{1}+k_{2}}\left( \frac{\tau _{i}}{\xi _{i}}\right) ^{2\beta
_{i}-(k_{1}+k_{2})}.
\end{equation*}%
Therefore, \eqref{Eq_Exp_Zbeta} and \eqref{Eq_Variance_Zbeta} yield 
\begin{equation*}
R(\Bs{\theta})=\mathbb{E}\left[ H\left( \mathbf{Y}_{\Bs{\theta}%
}\right) \right] ,
\end{equation*}%
with%
\begin{equation}
H\left( y_{1},y_{2}\right) =\frac{x_{1}^{\beta
_{1}}\left( y_{1}\right) x_{2}^{\beta _{2}}\left( y_{2}\right) -C_{\beta _{1},\eta _{1},\tau
_{1},\xi _{1}}C_{\beta _{2},\eta _{2},\tau _{2},\xi _{2}}}{\sqrt{D_{\beta
_{1},\eta _{1},\tau _{1},\xi _{1}}D_{\beta _{2},\eta _{2},\tau _{2},\xi _{2}}%
}}  \label{Eq_Exp_R_CorrWind_H}
\end{equation}%
where $x_{i}\left( y_{i}\right) =\left( \eta _{i}-\tau _{i}/\xi _{i}\right)
+\tau _{i}y_{i}^{\xi _{i}}/\xi _{i}$ since $\xi _{i}\neq 0$, $i=1,2$.

On top of being useful for actuarial applications, the measure %
\eqref{Eq_Ex_Relevant_Dependence_Measure} has a closed-form derivative with
respect to the dependence parameters of the max-stable field, allowing us to
compare the values of $\partial R(\Bs{\theta})/\partial \Bs{\theta}%
|_{\Bs{\theta}=\Bs{\theta}_{0}}$ obtained using LRM or IPA with
their true values. The analytical formulas of these derivatives are given in
Appendix \ref{Sec_Anal_Deriv}.

\subsubsection{Validity of the assumptions}

We now show that the assumptions of Theorems \ref{Th_Interchange_LRM} and 
\ref{Th_MainResultIPA} are valid in the case of
\eqref{Eq_Ex_Relevant_Dependence_Measure}, for $Y$ being
the Brown--Resnick and the Smith random field, respectively. Consequently, LRM and IPA can be used to assess the derivative of \eqref{Eq_Ex_Relevant_Dependence_Measure} with respect to $\Bs{\theta}$ and $\Sigma$, respectively.

\begin{Prop}
\label{Prop_Validity_Assumptions}Let $H$ be defined as in %
\eqref{Eq_Exp_R_CorrWind_H}.

i)\ Assume that $Y_{\Bs{\theta}}$ is the Brown--Resnick random field
with semivariogram given by \eqref{Eq_Power_Variogram}. There exist a
non-random neighbourhood $\mathcal{V}_{\Bs{\theta}_{0}}$ of $\Bs{
\theta }_{0}$, $\mathcal{V}_{\Bs{\theta}_{0}}$, and a constant $\alpha
>0$ such that \eqref{Eq_Gen_Bound} and \eqref{Eq_bound_der_lambda} hold true.

ii) Assume that $Y_{\Sigma}$ is the Smith random field. Then there
exist $p>1$ and a non-random neighbourhood of $\Sigma _{0}$, $\mathcal{V}%
_{\Sigma _{0}}$, such that \eqref{Eq_Condition_Main_Theorem} holds true.
\end{Prop}

\subsection{Simulation study}

\label{Subsec_Numerical_Study}

We numerically assess the accuracy of the LRM and IPA for computing the
sensitivities of the dependence measure $R\left( \Bs{\theta}\right) $
in \eqref{Eq_Ex_Relevant_Dependence_Measure}, where $Y$
is the Brown--Resnick field associated with the semi-variogram \eqref{Eq_Power_Variogram}, and the Smith field, respectively. The number of simulations used to approach the expectations in \eqref{Eq_FinalResultTheorem} and \eqref{Eq_MainResultIPASmith} is denoted by $S$. We display the
boxplots of the relative errors of $100$ estimates in different
configurations. We take $\beta _{1}=\beta _{2}=\beta $ with $\beta =2,3$, $%
S=10^{4},10^{5}$, and we look at different combinations of sites $\mathbf{x}%
_{1},\mathbf{x}_{2}$. We recall that the relative errors can be calculated
since the sensitivities of $R\left( \Bs{\theta}\right) $ can be
obtained analytically with an integral form. The
integrals in \eqref{Eq_DetailedExprDepenMeasWind} and
\eqref{Eq_Deriv_Risk_Measure} were computed using adaptive quadrature with a
relative tolerance of $10^{-7}$ to allow an accurate approximation.

\subsubsection{LRM for the Brown--Resnick field}

\label{Subsubsec_SimLMRBR}

We examine three combinations of sites: we set $\mathbf{x}_{1}=(0,0)^{\prime
}$ and $\mathbf{x}_{2}=(1,1)^{\prime },(3,2)^{\prime },(9,9)^{\prime }$.
These have been chosen in order to cover a wide range of sensitivities and
relative sensitivities. Moreover, we take $\Bs{\theta}_{0}=\left(
\psi_{0},\kappa _{0}\right) ^{\prime }=(3.05,0.86)^{\prime }$, $\eta
_{1}=\eta _{2}=26.11$, $\tau _{1}=\tau _{2}=2.90$ and $\xi _{1}=\xi
_{2}=-0.11$, which are the estimates obtained on the data used in the
application below (Section \ref{Subsec_Application}). We simulated the
Brown--Resnick field using the \texttt{rmaxstab} function of the
SpatialExtremes R package by \cite{PackageSpatialExtremes}.

Table \ref{Table_Sensitivity_BR} shows that, for $\kappa$ and $%
\psi$, the theoretical values of $R\left( \Bs{\theta}_{0}\right) 
$, its sensitivity and its relative sensitivity do not evolve much when
increasing $\beta $ from $2$ to $3$; the absolute values of the relative
sensitivities slightly decrease whereas $R\left( \Bs{\theta}_{0}\right) 
$ weakly increases. The absolute values of the sensitivities and of their
relative counterparts are the highest for $\mathbf{x}_{2}=(9,9)^{\prime }$;
in that case, the relative sensitivities are very high (about $30\%$ for $%
\kappa$ and more than $150\%$ for $\psi$), probably owing to the
fairly low value of $R\left( \Bs{\theta}_{0}\right) $. Such large
values typically explain the practical importance of properly assessing
sensitivities, e.g., in an insurance context.

\begin{table}[H]
\center
\begin{tabular}{c||cc|cc||cc|cc||cc}
& \multicolumn{2}{c}{$\kappa$} & \multicolumn{2}{c||}{$\psi$} & 
\multicolumn{2}{c}{$\kappa$} & \multicolumn{2}{c||}{$\psi$} & 
\multicolumn{2}{c}{$R\left( \Bs{\theta}_{0}\right) $} \\ \hline
\backslashbox{$\Mb{x}_2$}{$\beta$} & 2 & 3 & 2 & 3 & 2 & 3 & 2 & 3 & 2 & 3
\\ \hline
$(1,1)^{\prime }$ & 0.048 & 0.046 & 0.131 & 0.126 & 0.061 & 0.058 & 0.167 & 
0.158 & 0.784 & 0.797 \\ 
$(3,2)^{\prime }$ & 0.074 & 0.074 & -0.044 & -0.044 & 0.122 & 0.117 & -0.072
& -0.070 & 0.610 & 0.626 \\ 
$(9,9)^{\prime }$ & 0.087 & 0.089 & -0.439 & -0.452 & 0.306 & 0.302 & -1.552
& -1.529 & 0.283 & 0.296 \\  
\end{tabular}
\newline
\caption{The left panel gives the values of $\partial R(\mathbf{\Bs
\theta })/\partial \protect\kappa |_{\Bs{\protect\theta }=\mathbf{%
\protect\theta }_{0}}$ and $\partial R(\Bs{\protect\theta })/\partial 
\protect\psi |_{\Bs{\protect\theta }=\Bs{\protect\theta }_{0}}$ in
the various configurations. The middle one displays the previous values
normalized by $R(\Bs{\protect\theta }_{0})$. The right one gives $R(%
\mathbf{\Bs\theta }_{0})$. All values are the theoretical ones.}
\label{Table_Sensitivity_BR}
\end{table}

As $M=2$, the density of the Brown--Resnick field is explicitly known and thus the
score function has a closed-form expression (available upon request) with a small number of terms; it is
therefore needless to apply the MCMC methodology through Gibbs sampler
described above.

It is important to remark that $\partial \log f_{\Bs{\theta}%
}(y_{1},y_{2})/\partial \psi $ and $\partial \log f_{\Bs{\theta}%
}(y_{1},y_{2})/\partial \kappa $ are proportional (Appendix \ref%
{Sect_Proport_BR}), which entails that the LRM estimates of $\partial R(\Bs{
\theta })/\partial \psi $ and $\partial R(\Bs{\theta})/\partial
\kappa $ are also proportional. Moreover, by \eqref{Eq_Deriv_Risk_Measure} and \eqref{Eq_Deriv_gs_Numerical_study}, the true values of those derivatives are proportional with the same factor, implying that the relative errors of the
estimates of $\partial R(\Bs{\theta})/\partial \psi $ and $\partial R(%
\Bs{\theta})/\partial \kappa $ are equal. Consequently, we only
display the results for $\kappa $.

Figure \ref{Fig_Res_LRM_Kappa} shows that the LRM estimator is unbiased, as
expected. Moreover its variability substantially decreases when increasing $%
S $ from $10^{4}$ to $10^{5}$ and becomes very small (the relative error of
most estimates is less than $5\%$). This variability is the lowest for $%
\mathbf{x}_{2}=(9,9)^{\prime }$, perhaps due to the fact that this
configuration corresponds by far to the largest sensitivities and relative
sensitivities. To some extent, we may expect the accuracy of the estimation
to be an increasing function of the absolute value of the relative
sensitivity. Nevertheless, this seems to be more complex: the variability is
lower for $\mathbf{x}_{2}=(3,2)^{\prime }$ than for $\mathbf{x}%
_{2}=(1,1)^{\prime }$ 
and, for all site combinations, for $\beta =3$ than for $\beta =2$. Overall,
the LRM estimator is very satisfying in all configurations, whether the
sensitivity (or relative sensitivity) is low or large.

\begin{figure}[H]
\center
\includegraphics[scale=0.8]{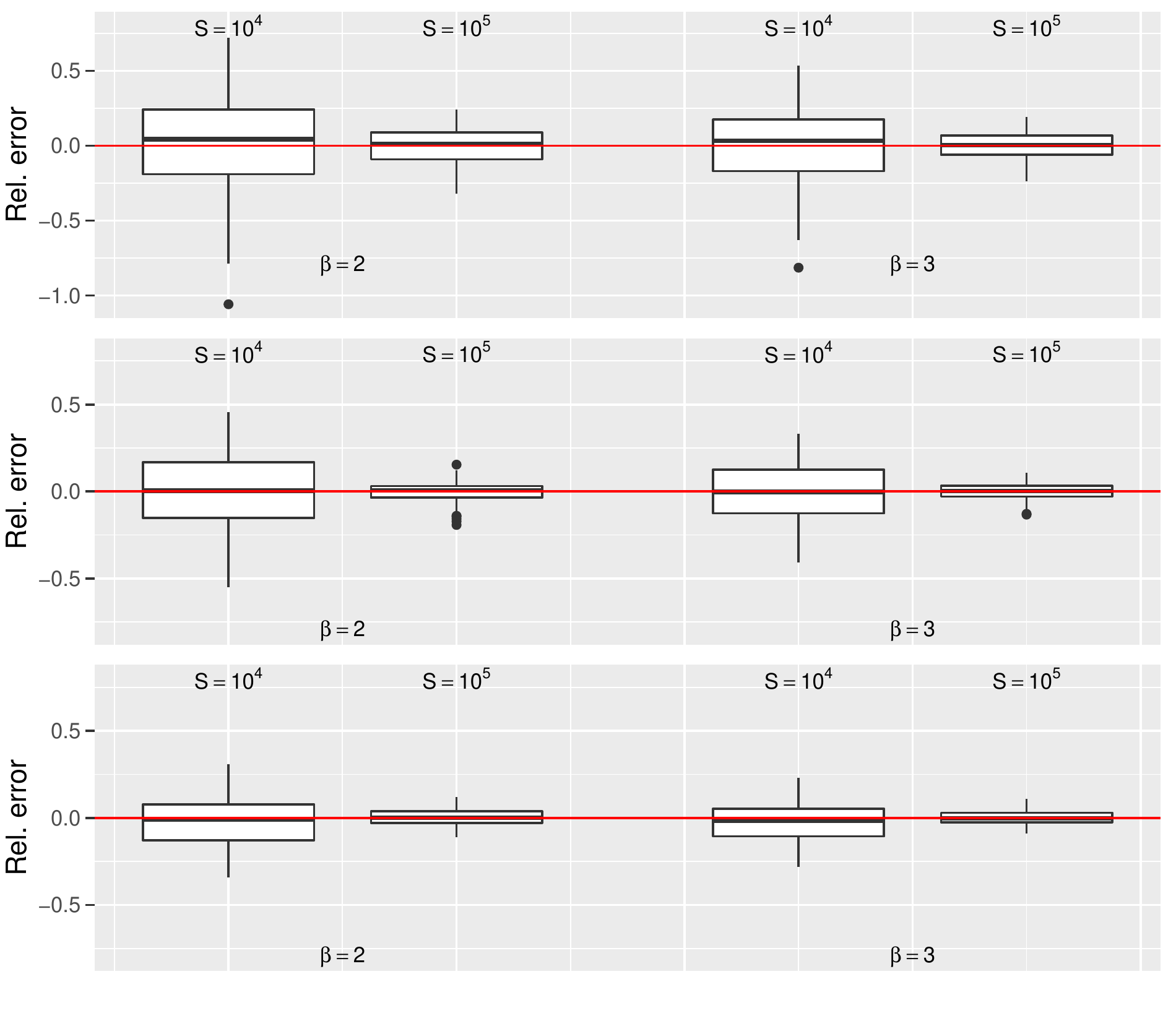}
\caption{Boxplots of the relative errors of each estimate for the derivative
with respect to $\protect\kappa$. Each row corresponds to a combination of
sites: from top to bottom, $\mathbf{x}_2=(1,1)^{\prime }$, $\mathbf{x}%
_2=(3,2)^{\prime }$ and $\mathbf{x}_2=(9,9)^{\prime }$.}
\label{Fig_Res_LRM_Kappa}
\end{figure}

\subsubsection{IPA for the Smith field}

We consider two combinations of sites: we set $\mathbf{x}%
_{1}=(0,0)^{\prime }$ and $\mathbf{x}_{2}=(1,1)^{\prime },(3,2)^{\prime }$.
Additionally, we take 
\begin{equation*}
\Sigma _{0}=%
\begin{pmatrix}
0.88\  & 0.07 \\ 
0.07\  & 2.43%
\end{pmatrix}%
,
\end{equation*}%
$\eta_1=\eta_2 =26.12$, $\tau_1 = \tau_2 =2.92$ and $\xi_1=\xi_2 =-0.10$, which are the estimates
obtained in Section \ref{Subsec_Application}. Unlike in Section \ref%
{Subsubsec_SimLMRBR}, we do not consider the case $\mathbf{x}%
_{2}=(9,9)^{\prime }$ as $R(\Sigma _{0})$ is approximately equal to $0$ in
that configuration. Note that we do not display the results for $\sigma _{21}
$ as they are exactly the same as those for $\sigma _{12}$, consistently
with the theory. 
The analytical computation of the terms $\partial
Y_{\Sigma }(\mathbf{x}_{j})/\partial \Sigma $, $j=1,2$, which is necessary
to implement IPA (see \eqref{Eq_Chain_Rule}), requires the coordinates of
the centers of the \textquotedblleft storms\textquotedblright\ 
\citep[see][for the
interpretation of the Smith field in terms of storms]{smith1990max}
realizing the maxima at the sites $\mathbf{x}_{j}$ (see %
\eqref{Eq_Deriv_Realis_Sig} in Section \ref{Sec_Proof_Th_MainResultIPA}). To
the best of our knowledge, these quantities cannot be obtained from the
simulation algorithms available on the Web (e.g., in R packages like
SpatialExtremes by \cite{PackageSpatialExtremes} or in the code by \cite%
{dombry2016exact} available on the Biometrika website). To overcome this
impediment and for other technical reasons, we programmed the simulation
algorithm of the Smith random field ourselves by adapting the approach of 
\citet[][Theorem
4]{schlather2002models}. Regarding the quantity $r$ appearing in that
approach, we took the value $r=15$ in order to ensure an accurate
simulation. The corresponding code will be available.

Table~\ref{TableValuesSensCorr_Smith} shows that, when increasing $\beta $
from $2$ to $3$, the sensitivities and relative sensitivities slightly
decrease whereas $R(\Sigma _{0})$ weakly increases. The relative
sensitivities are the highest for $\mathbf{x}_{2}=(3,2)^{\prime }$; they
take very high values for $\sigma _{11}$ and $\sigma _{12}$ (about $36\%$
for $\sigma _{12}$ and more than $160\%$ for $\sigma _{11}$), probably due
to the low value of $R(\Sigma _{0})$.

\begin{table}[H]
\center
\resizebox{\textwidth}{!}{
\begin{tabular}{c||cc|cc|cc||cc|cc|cc||cc}
 & \multicolumn{2}{c}{$\sigma_{11}$} & \multicolumn{2}{c}{$\sigma_{12}$} & \multicolumn{2}{c||}{$\sigma_{22}$} & \multicolumn{2}{c}{$\sigma_{11}$} & \multicolumn{2}{c}{$\sigma_{12}$} & \multicolumn{2}{c||}{$\sigma_{22}$} & \multicolumn{2}{c}{$R(\Sigma_0)$} \\ 
\hline 
\backslashbox{$\Mb{x}_2$}{$\beta$} & 2 & 3 & 2 & 3 & 2 & 3 & 2 & 3 & 2 & 3 & 2 & 3 & 2 & 3 \\ 
\hline 
$(1,
1)^{\prime}$ & 0.174 & 0.170 & 0.06 & 0.058 & 0.020 & 0.020 & 0.242 & 0.232 & 0.083 & 0.080 & 0.029 & 0.027 & 0.717 & 0.732 \\ 
$(3,
2)^{\prime}$ & 0.233 & 0.243 & 0.05 & 0.053 & 0.011 & 0.011 & 1.669 & 1.655 & 0.362 & 0.359 & 0.078 & 0.078 & 0.139 & 0.147  \\ 
\end{tabular}
}
\newline
\caption{Same as in Table \protect\ref{Table_Sensitivity_BR}, but for $%
\partial R(\Sigma )/\partial \protect\sigma _{11}|_{\Sigma =\Sigma _{0}}$, $%
\partial R(\Sigma )/\partial \protect\sigma _{12}|_{\Sigma =\Sigma _{0}}$
and $\partial R(\Sigma )/\partial \protect\sigma _{22}|_{\Sigma =\Sigma
_{0}} $.}
\label{TableValuesSensCorr_Smith}
\end{table}

Figures~\ref{Fig_Res_IPA_11}--\ref{Fig_Res_IPA_22} indicate that the IPA
estimator is unbiased (at least for $S$ large enough) and, as expected, the
variability decreases when increasing $S$ from $10^{4}$ to $10^{5}$. In most
cases, it reaches a low level (with a relative error of most estimates less than $5\%$). The variability is slightly lower in the case of $\mathbf{x%
}_{2}=(3,2)^{\prime }$, which can be explained by the fact that the relative
sensitivity is much larger in that case. It is however the converse for $%
\sigma _{22}$, coefficient for which the increase of relative sensitivity
compared to the case $\mathbf{x}_{2}=(1,1)^{\prime }$ is lower than for the
other coefficients. Moreover, the variability is systematically higher for $%
\sigma _{22}$ than for the other coefficients, and, whatever the coefficient
and the combination of sites, it tends to be lower for $\beta =3$ than for $%
\beta =2$. We intuitively expect the method to perform the best for high
relative sensitivities, but we see that it seems more complex.
Overall, IPA performs very well on this example, whether the sensitivity (or
relative sensitivity) is low or high.

\begin{figure}[H]
\center
\includegraphics[scale=0.8]{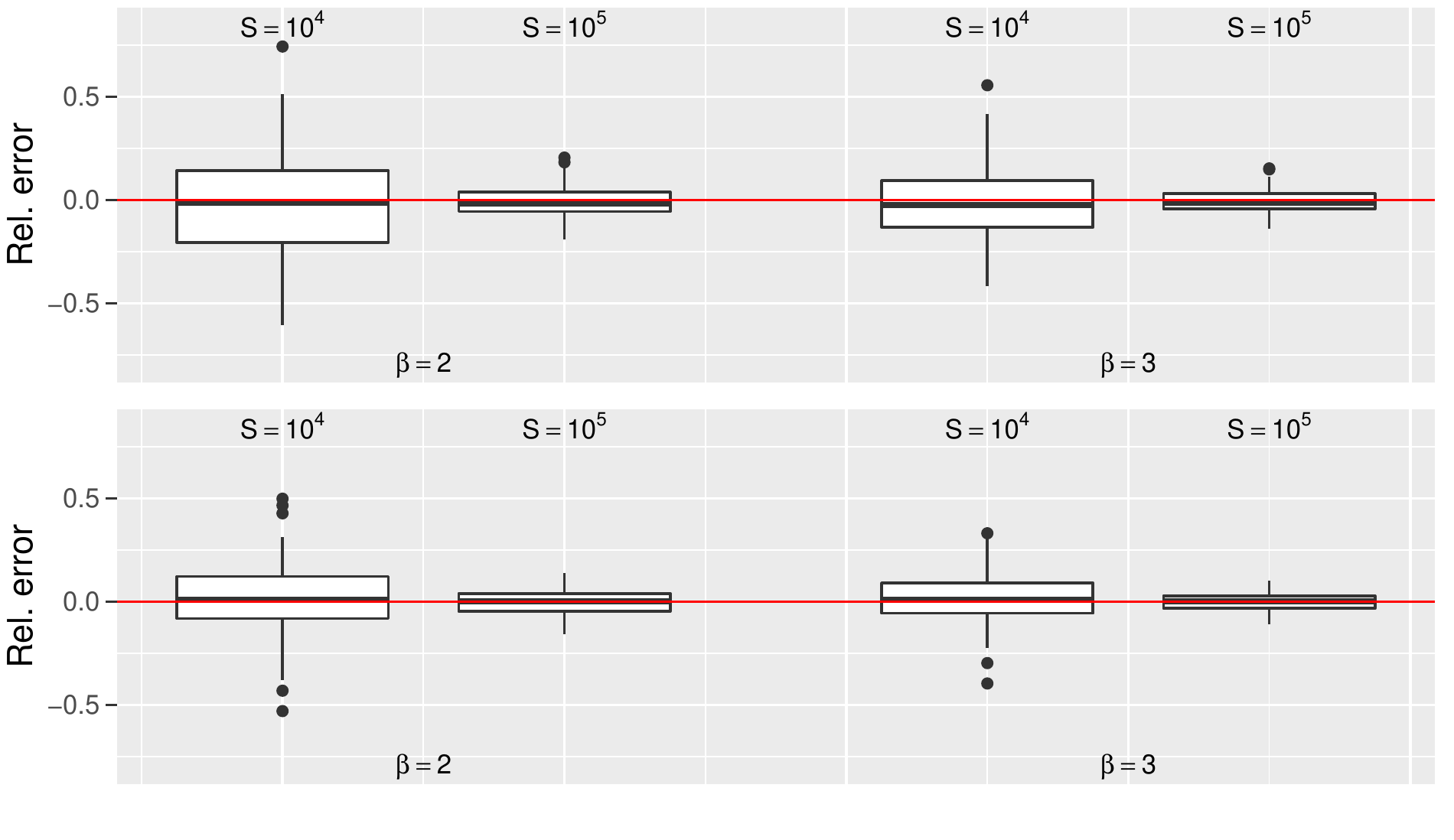}
\caption{Boxplots of the relative errors of each estimate for the derivative
with respect to $\protect\sigma _{11}$. Each row corresponds to a
combination of sites: from top to bottom, $\mathbf{x}_{2}=(1,1)^{\prime }$
and $\mathbf{x}_{2}=(3,2)^{\prime }$.}
\label{Fig_Res_IPA_11}
\end{figure}

\begin{figure}[H]
\center
\includegraphics[scale=0.8]{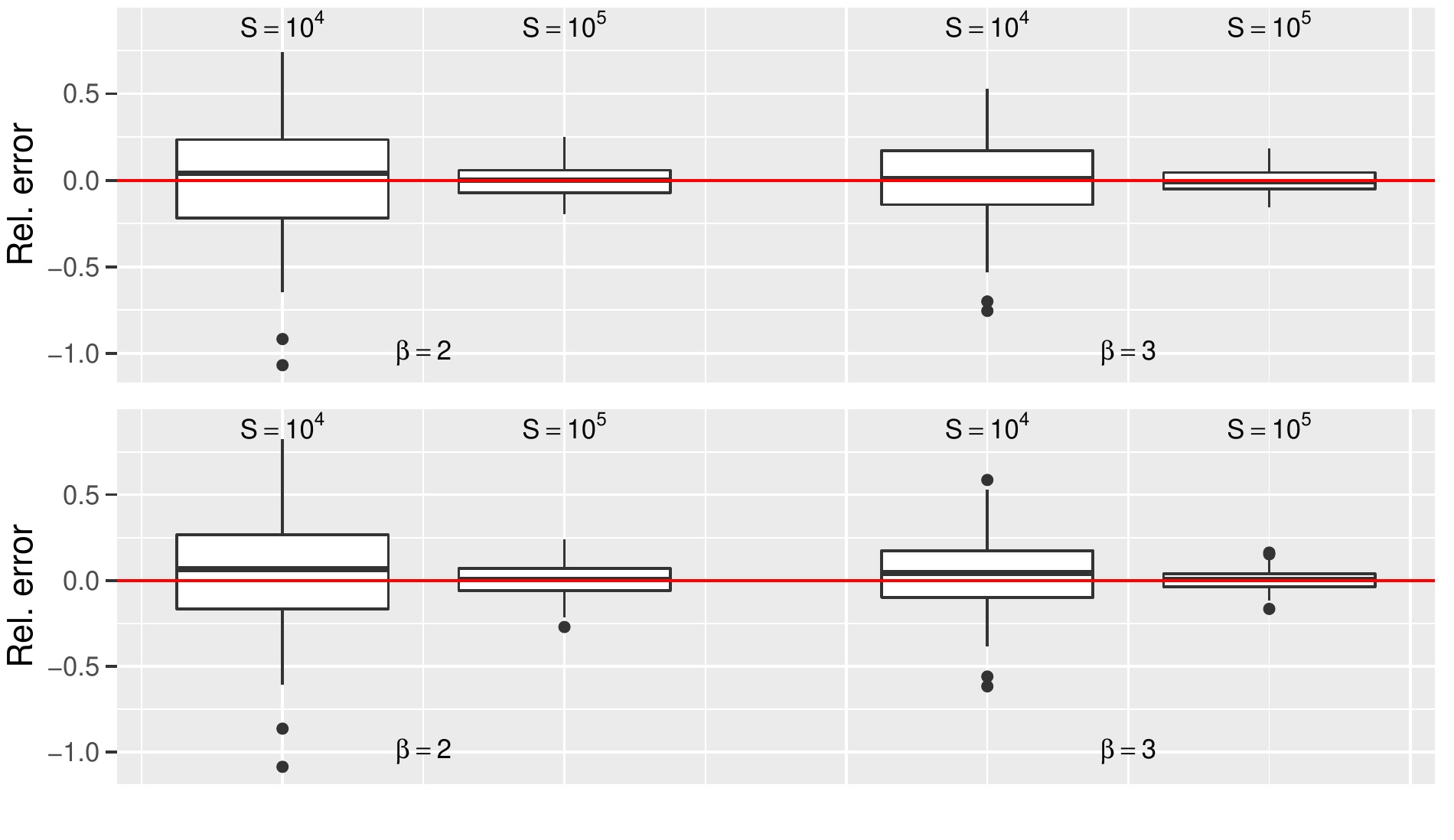}
\caption{Same as Figure \protect\ref{Fig_Res_IPA_11} for $\protect\sigma %
_{12}$.}
\label{Fig_Res_IPA_12}
\end{figure}

\begin{figure}[H]
\center
\includegraphics[scale=0.8]{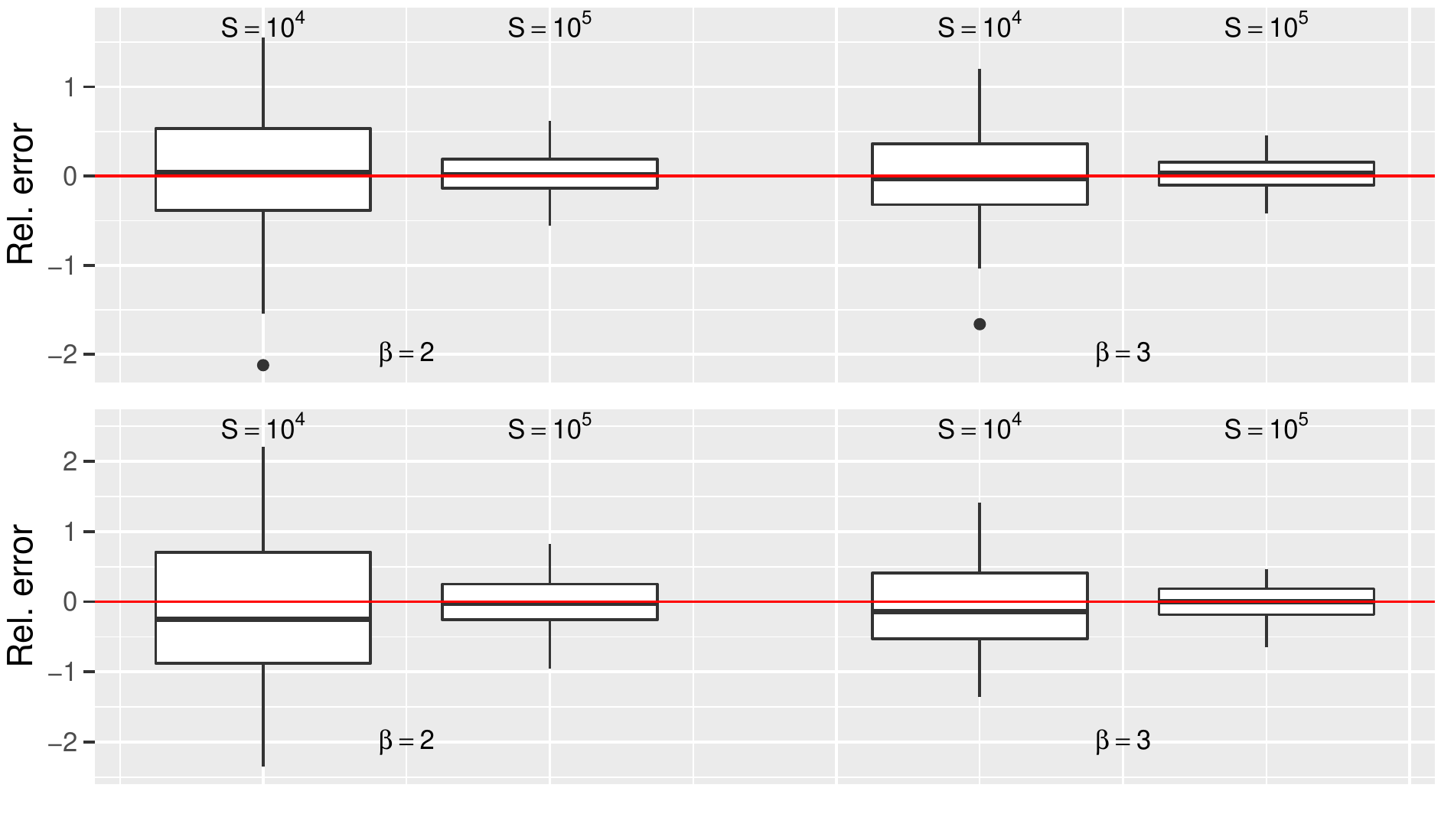}
\caption{Same as Figure \protect\ref{Fig_Res_IPA_11} for $\protect\sigma %
_{22}$.}
\label{Fig_Res_IPA_22}
\end{figure}

\subsection{Application}

\label{Subsec_Application}

This section is devoted to the computation of sensitivities of the measure $%
R(\Bs{\theta})$ (see \eqref{Eq_Ex_Relevant_Dependence_Measure}) in a
real case study that is valuable for the insurance/reinsurance industry. We
consider concrete wind speed data on a region where insurance/reinsurance
plays a significant role and for which power-law type damage functions for
losses stemming from extreme wind speeds are clearly documented in the
literature.

\subsubsection{Data}

We consider publicly available data from the European Centre for Medium
Range Weather Forecasting (ECMWF), more precisely a subset of the ERA5
(ECMWF Reanalysis 5th Generation) data set. More specifically the data we
study consist in hourly $10$ m wind gust time series from 1 January 1979 at
07:00 to 1 June 2019 at 00:00. The region we consider is a rectangle from $%
6^{\circ }$ to $9.75^{\circ }$ longitude and $49.75^{\circ }$ to $%
52.25^{\circ }$ latitude, and the resolution is $0.25^{\circ }$ longitude
and $0.25^{\circ }$ latitude; see Figure \ref{Fig_Plot_Region}. Thus, the
rectangle contains 176 grid points and is basically centred over the Ruhr
region in Germany. This area exhibits a high total residential insured value
per unit of surface, and, according to \citet[][Figure 2]{prahl2012applying}%
, a damage function of power-law type with an exponent around 8 seems
appropriate for insured losses (due to high wind speeds) on residential
buildings in that region. At each grid point, we compute the $41$ seasonal
(from October to March) maxima, leading to the dataset we use to fit
different max-stable models for $Y_{\Bs{\theta}}$. For the first
season, the maximum is computed over January--March. Considering the period
October--March allows us to remove seasonal non-stationarity in the wind
speed time series and to focus mainly on winter storms (and not on highly
localized extreme winds occurring during intense summer thunderstorms). 
\begin{figure}[H]
\center
\includegraphics[scale=1]{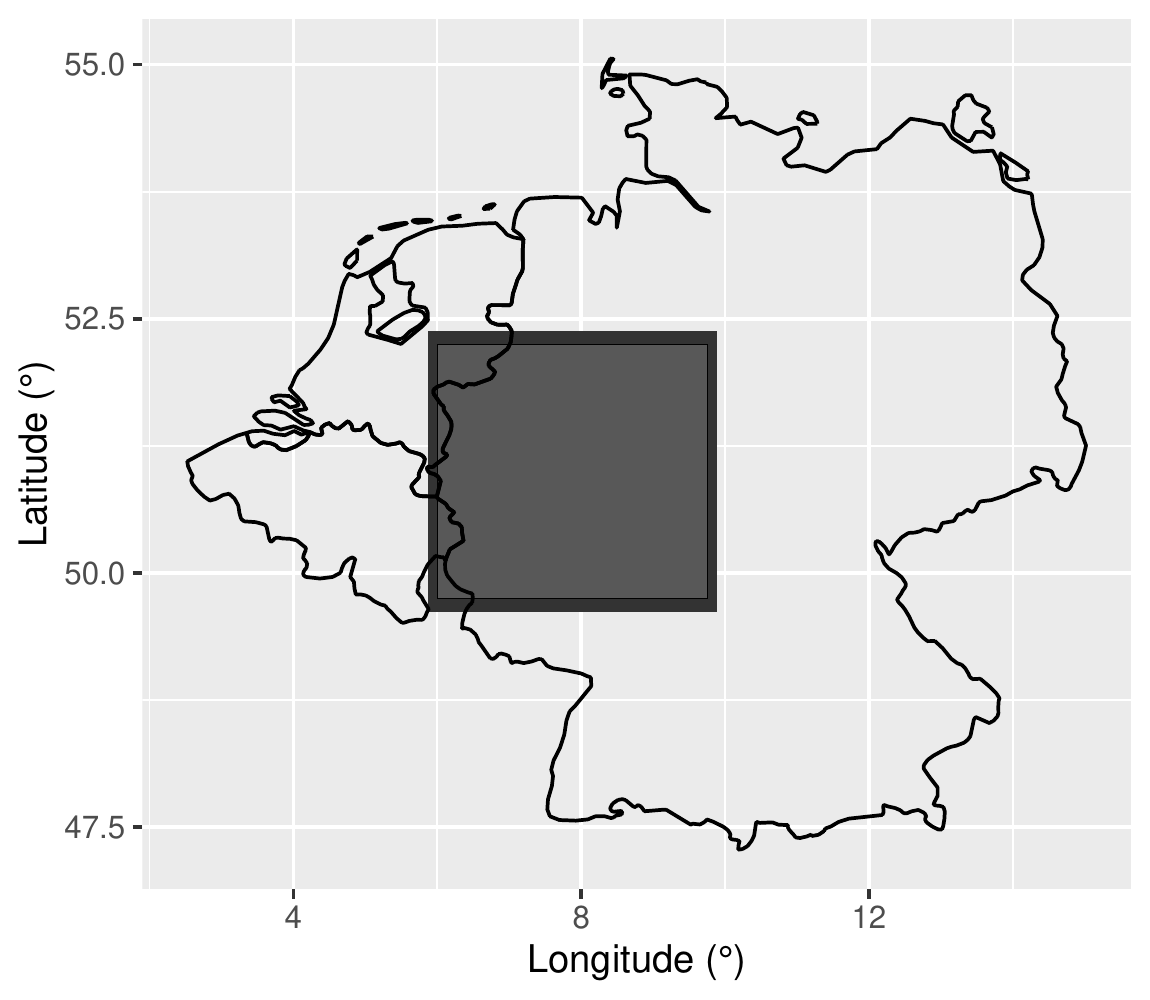}
\caption{Considered region (indicated by the shaded rectangle).}
\label{Fig_Plot_Region}
\end{figure}

\subsubsection{Results}

The max-stable models we consider are both the Brown--Resnick random field
with semi-variogram specified in \eqref{Eq_Power_Variogram} and the Smith
field. Regarding the location, scale and shape parameters, it is reasonable
to consider them as constant over the region (shown in a more detailed
analysis of the same data that will be soon available in a subsequent
paper). Modelling these parameters using trend surfaces (instead of
considering specific parameters at each grid point) is common as it reduces
the estimation time and allows prediction at sites where no observations are
available.

Both models are fitted using maximum pairwise likelihood estimation %
\citep[e.g.,][]{padoan2010likelihood, davison2012statistical} implemented in
the \texttt{fitmaxstab} function of SpatialExtremes \citep{PackageSpatialExtremes}.
The marginal and dependence parameters are estimated in a single step. Then
we perform model selection by minimization of the composite likelihood
information criterion (CLIC); see \cite{varin2005note}.

Table \ref{Table_CLIC_Param_Application} shows that, according to the CLIC,
the Brown--Resnick random field is more compatible with our data than the
Smith random field is. We also fitted the Schlather random field %
\citep{schlather2002models} with various correlation functions
(Whittle-Matérn, powered exponential and Cauchy) and the Brown--Resnick
random field is more appropriate according to the CLIC. This is in
agreement, e.g., with the results obtained by \cite{davison2012statistical}
in the case of rainfall.

\begin{table}[H]
\center
\resizebox{\textwidth}{!}{
\begin{tabular}{c|c|c|c|c|c|c|c}
\textbf{Brown--Resnick} & CLIC &  & $\kappa $ & $\psi $ & $\eta $ & $\tau $
& $\xi $ \\ 
& 6020169 &  & $3.05\ (0.98)$ & $0.86\ (0.06)$ & $26.11\ (0.41)$ & $2.90\
(0.24)$ & $-0.11\ (0.03)$ \\ \hline
\textbf{Smith} & CLIC & $\sigma _{11}$ & $\sigma _{12}$ & $\sigma _{22}$ & $%
\eta $ & $\tau $ & $\xi $ \\ 
& 6084169 & $0.88\ (0.17)$ & $0.07\ (0.03)$ & $2.43\ (0.47)$ & $26.12\
(0.38) $ & $2.92\ (0.22)$ & $-0.10\ (0.01)$
\end{tabular}
}
\newline
\caption{Values of the CLIC and fitted parameters of the Brown--Resnick and
Smith max-stable random fields. The values inside the brackets are the
standard errors.}
\label{Table_CLIC_Param_Application}
\end{table}

Figure \ref{Fig_GoodnessFitBR} shows that the theoretical extremal
coefficient function \citep[e.g.,][]{schlather2003dependence} of the fitted
Brown--Resnick model agrees reasonably well with the empirical extremal
coefficients, though it is slightly above their binned estimates. This low
underestimation of the spatial dependence may come from the choice of very
parsimonious trend surfaces for the location, scale and shape
parameters. The comparison done in Figure \ref{Fig_GoodnessFitBR}
constitutes a classical graphical goodness-of-fit diagnostic for max-stable
models and it shows here that the proposed model fits the data sufficiently
well for the purpose of this application.

\begin{figure}[H]
\center
\includegraphics[scale=0.7]{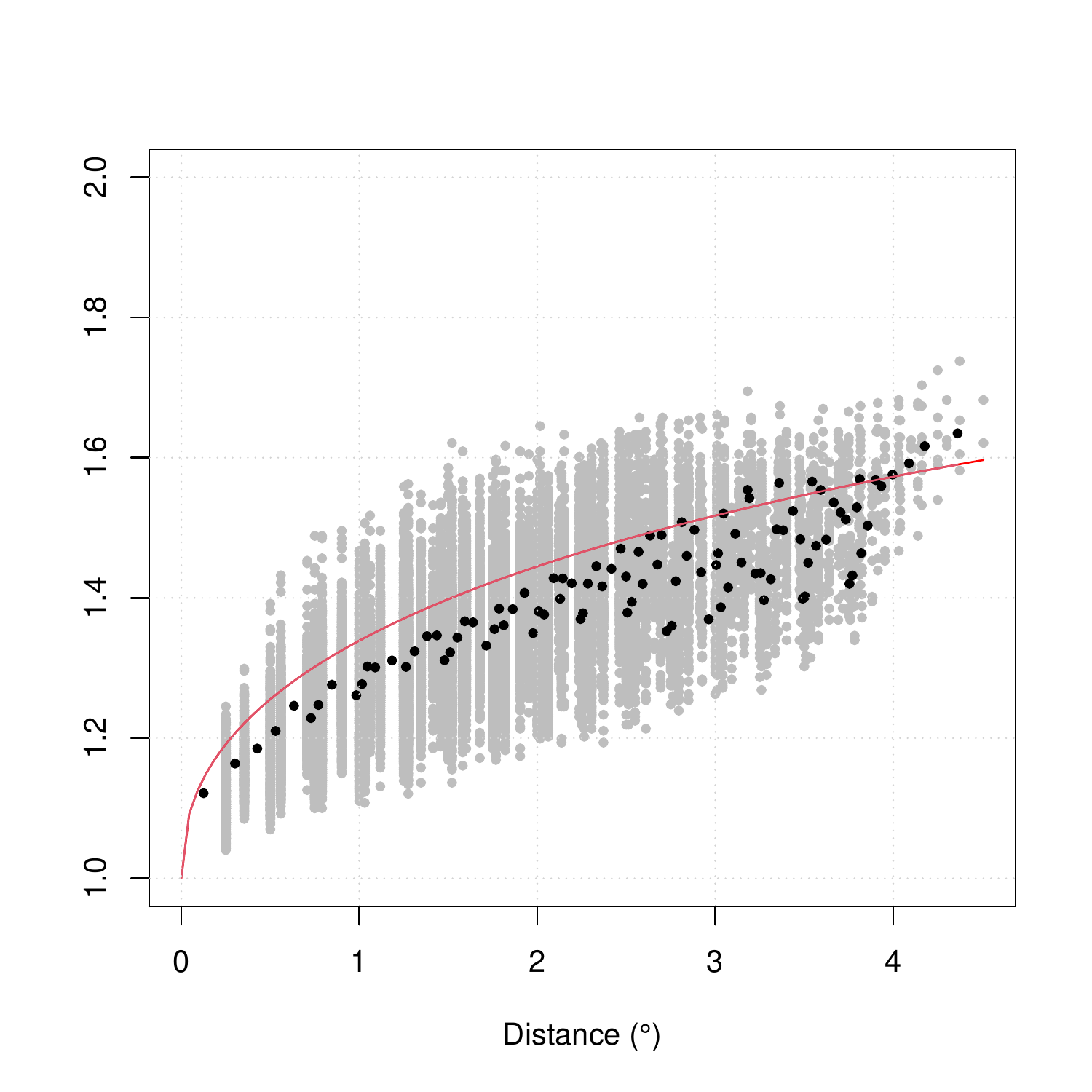}
\caption{Theoretical extremal coefficient function from the fitted model
(red line) and empirical extremal coefficients (points). The grey and black
points are pairwise and binned (with $1000$ bins) estimates, respectively.}
\label{Fig_GoodnessFitBR}
\end{figure}

The final dependence measure is 
\begin{equation}
R(\Bs{\theta})=\mbox{Corr}\left( X^{8}_{\Bs{\theta}}(\mathbf{x}%
_{1}),X^{8}_{\Bs{\theta}}(\mathbf{x}_{2}) \right) ,\quad \mathbf{x}%
_{1},\mathbf{x}_{2}\in \Mbb{R}^{2},  \label{Eq_FinalDepMeasure}
\end{equation}%
where $X_{\Bs{\theta}}$ is a general Brown--Resnick field having as
marginal parameters those in Table \ref{Table_CLIC_Param_Application}, and
we aim at providing its derivative at $\Bs{\theta}_{0}=(3.05,0.86)^{\prime }$. As the Brown--Resnick random field appears as a suitable model for wind
speed maxima and the power-law with $8$ as exponent is an appropriate damage
function, this dependence measure is well suited in practice for
(re)insurance companies on the region considered. The values of the
sensitivities obtained by LRM and their relative counterparts are given in
Table \ref{Table_FinalSensitivities} for the same combinations of sites as
in Section \ref{Subsubsec_SimLMRBR}. The combinations $\mathbf{x}%
_{1}=(0,0)^{\prime },\mathbf{x}_{2}=(1,1)^{\prime }$ and $\mathbf{x}%
_{1}=(0,0)^{\prime },\mathbf{x}_{2}=(3,2)^{\prime }$ correspond to pairs of
sites within our region of interest if we put the origin $(0,0)^{\prime }$
on the left part and in the lower-left corner of the rectangle,
respectively; by stationarity of the field, the origin can be chosen
arbitrarily. The site $\mathbf{x}_{2}=(9,9)^{\prime }$ is far to the East of
our region even we take as origin the lower-left corner of the rectangle,
but we chose it to emphasize that sensitivities can be very high for sites
that are highly distant. The relative sensitivities are slightly lower in
absolute value than in the cases $\beta =2$ and $\beta =3$ (see Table \ref%
{Fig_Res_LRM_Kappa}). However they can be quite substantial (e.g., for $%
\psi $ in the case $\mathbf{x}_{2}=(1,1)^{\prime }$ and $\kappa $ in the
case $\mathbf{x}_{2}=(3,2)^{\prime }$) or very high (for both $\kappa$
and $\psi$ in the case $\mathbf{x}_{2}=(9,9)^{\prime }$), which
highlights that looking at sensitivities is strongly recommended in concrete
risk assessment studies. 

\begin{table}[H]
\center
\begin{tabular}{c||c|c||c|c||c}
$\mathbf{x}_{2}$ & $\kappa$ & $\psi$ & $\kappa$ & $\psi$ & $R(\Bs{\theta}_{0})$ \\ \hline
$(1,1)^{\prime }$ & 0.039 & 0.106 & 0.046 & 0.126 & 0.840 \\ 
$(3,2)^{\prime }$ & 0.068 & -0.041 & 0.100 & -0.059 & 0.685 \\ 
$(9,9)^{\prime }$ & 0.096 & -0.486 & 0.278 & -1.409 & 0.345 
\end{tabular}
\newline
\caption{The left panel gives the values of $\partial R(\Bs{\protect \theta})/\partial \protect\kappa |_{\Bs{\protect \theta} =\Bs{%
\protect\theta }_{0}}$ and $\partial R(\Bs{\protect\theta })/\partial 
\protect\psi |_{\Bs{\protect\theta }=\Bs{\protect\theta }_{0}}$
for the different combinations of sites. The middle one displays the
previous values normalized by $R(\Bs{\protect\theta }_{0})$. The right
one gives $R(\Bs{\protect\theta }_{0})$. The sensitivities and
normalized sensitivities have been computed using the LRM method with $%
S=10^{6}$.}
\label{Table_FinalSensitivities}
\end{table}

\section{Discussion}

\label{Sec_Conclusion}

Max-stable random fields are
particularly suitable to model extreme spatial (e.g., environmental) events and, if applicable, LRM as well as IPA are powerful techniques
to estimate the derivatives of an expected performance with respect to
the model parameters. In this paper,
we introduce these two methods to extreme-value theory by applying them to expected performances based on the Brown--Resnick and the Smith max-stable random fields.
In a first part, we give convenient and tractable conditions on these fields that ensure the validity of LRM and IPA. Obtaining such tractable conditions is
non-trivial due to the complex structure of max-stable fields. Then, we focus on one of the several frameworks where our theoretical results can be valuable, more precisely the context of risk assessment.
We apply LRM and IPA to a specific example of dependence measure which is suited to assess the dependence between damage due to extreme wind speeds at two stations, and is thus worthwhile for actuarial practice. We show through a simulation study that both methods perform very well in various configurations. We finally present a concrete application involving reanalysis wind speed data; it emphasizes the importance of considering sensitivities in a risk assessment context when using max-stable fields.

It would be interesting to see how the LRM and IPA estimators behave in
other configurations.
In some situations (e.g., where the derivatives are very close to $0$),
appropriate improvements of the classical
Monte Carlo method (such as importance sampling) might be considered to
reduce the variability of the estimators. Future appealing work would also
consist in extending our conditions to other max-stable random fields and
possibly even proposing other methods for max-stable fields \citep[such as, e.g., the one by][]{peng2018new} for which LRM and IPA are not applicable. 

\section*{Acknowledgements}

Both authors gratefully acknowledge Anthony C. Davison and anonymous referees for insightful comments. Erwan Koch would like to thank the Swiss National Science Foundation (project number $200021\_178824$) for financial support.

\newpage 
\appendix

\section{Proofs of the main results}

\label{Sec_Proofs}

\subsection{Proof of Theorem \protect\ref{Th_Interchange_LRM}}

\begin{proof}
From Section 4.4 in \cite{dombry2017asymp}, we know that $\mathbf{Y}_{%
\Bs{\theta}}$ has an H\"{u}sler-Reiss distribution with a strictly
conditionally negative definite matrix given by $\Lambda _{\Bs{\theta}%
}=\left( \lambda _{\Bs{\theta}}(\mathbf{x}_{i},\mathbf{x}_{j})\right)
_{1\leq i,j\leq m}$. By Lemma B.4. in \cite{dombry2017asymp}, we deduce
that, for any neighbourhood $\mathcal{V}_{\Bs{\theta}_{0}}$ of $\Bs{%
\theta }_{0}$ and any $\alpha >0$, there exists a constant $A_{\mathcal{V}_{%
\Bs{\theta}_{0}}}>0$ such that, for all $\Bs{\theta}\in \mathcal{V}%
_{\Bs{\theta}_{0}}$, $\mathbf{y}\in (0,\infty )^{M}$ with $\left\Vert 
\mathbf{y}\right\Vert =1$ and $B\subset \mathcal{I}$, 
\begin{equation*}
\sup_{j=1,...,L}\left\vert \frac{\partial \log V_{\Bs{\theta}}\left( 
\mathbf{y}\right) }{\partial \theta _{j}}\right\vert \leq c_{\alpha }\left( 
\mathbf{y}\right) \qquad \text{and}\qquad \sup_{j=1,...,L}\left\vert \frac{%
\partial }{\partial \theta _{j}}\frac{\partial ^{|B|}}{\partial \mathbf{y}%
_{B}}\log V_{\Bs{\theta}}\left( \mathbf{y}\right) \right\vert \leq
c_{\alpha }\left( \mathbf{y}\right),
\end{equation*}
where
\begin{equation*}
c_{\alpha }\left( \mathbf{y}\right) =A_{\mathcal{V}_{\Bs{\theta}%
_{0}}}\sum_{i=1}^{M}y_{i}^{-\alpha }\text{.}
\end{equation*}%
Equation (14) in \cite{dombry2017asymp} then gives, for all $\mathbf{y}\in
(0,\infty )^{M}$,%
\begin{equation}
\sup_{j=1,...,L}\left\vert \frac{\partial f_{\Bs{\theta}}\left( \mathbf{%
y}\right) }{\partial \theta _{j}}\right\vert \leq \left(
M+\sum_{i=1}^{M}y_{i}^{-1}\right) c_{\alpha }\left( \frac{\mathbf{y}}{%
\left\Vert \mathbf{y}\right\Vert }\right) f_{\Bs{\theta}}\left( \mathbf{%
y}\right) . \label{Eq_SupDerivPart}
\end{equation}%
Now let us consider the ratio%
\begin{equation*}
\frac{f_{\Bs{\theta}}\left( \mathbf{y}\right) }{f_{\Bs{\theta}%
_{0}}\left( \mathbf{y}\right) }=\exp \left( -(V_{\Bs{\theta}}(\mathbf{y}%
)-V_{\Bs{\theta}_{0}}(\mathbf{y}))\right) \frac{\sum_{\pi \in \Pi
}\left( -1\right) ^{|\pi |}\prod\limits_{B\in \pi }\frac{\partial ^{|B|}}{%
\partial \mathbf{y}_{B}}V_{\Bs{\theta}}(\mathbf{y})}{\sum_{\pi \in \Pi
}\left( -1\right) ^{|\pi |}\prod\limits_{B\in \pi }\frac{\partial ^{|B|}}{%
\partial \mathbf{y}_{B}}V_{\Bs{\theta}_{0}}(\mathbf{y})}.
\end{equation*}%
By \eqref{Eq_B}, we have
\begin{equation*}
V_{\Bs{\theta}}(\mathbf{y})-V_{\Bs{\theta}_{0}}(\mathbf{y}%
)=\sum_{i=1}^{M}\frac{1}{y_{i}}\left[ \phi _{i}\left( \mathbf{y},\Bs{%
\theta }\right) -\phi _{i}\left( \mathbf{y},\Bs{\theta}_{0}\right) %
\right] \geq B_{\mathcal{V}_{\Bs{\theta}_{0}}}\sum_{i=1}^{M}\frac{1}{%
y_{i}},
\end{equation*}%
where $B_{\mathcal{V}_{\Bs{\theta}_{0}}}$ is given in \eqref{Eq_B}. 
From Equation (28) in \cite{asadi2015extremes}, we know that for all $B\subset 
\mathcal{I}$ and $\mathbf{y}\in (0,\infty )^{M}$ 
\begin{equation*}
\frac{\partial ^{|B|}}{\partial \mathbf{y}_{B}}V_{\Bs{\theta}}(\mathbf{y%
})=\frac{1}{\prod_{i\in B}y_{i}}\left( \frac{1}{|B|}\sum_{i\in
B}y_{i}^{-1}\varphi _{|B|-1}\left( \mathbf{\tilde{y}}_{B,i,\Bs{\theta}%
};R_{B,\Bs{\theta}}^{(i)}\right) \Phi _{M-|B|}\left( \mathbf{\tilde{y}}%
_{B^{c},i,\Bs{\theta}}-\mu _{B,\Bs{\theta}}^{(i)};P_{B,\Bs{%
\theta }}^{(i)}\right) \right),
\end{equation*}%
with%
\begin{eqnarray*}
\mathbf{\tilde{y}}_{B,i,\Bs{\theta}} &=&\left( \log \frac{y_{j}}{y_{i}}%
+2\lambda _{\Bs{\theta}}^{2}(\mathbf{x}_{i},\mathbf{x}_{j})\right)
_{j\in B,j\neq i} \\
R_{B,\Bs{\theta}}^{(i)} &=&\left( 2(\lambda _{\Bs{\theta}}^{2}(%
\mathbf{x}_{i},\mathbf{x}_{j})+\lambda _{\Bs{\theta}}^{2}(\mathbf{x}%
_{i},\mathbf{x}_{m})-\lambda _{\Bs{\theta}}^{2}(\mathbf{x}_{j},\mathbf{x%
}_{m}))\right) _{j,m\in B,j,m\neq i} \\
R_{B^{c},\Bs{\theta}} &=&\left( 2(\lambda _{\Bs{\theta}}^{2}(%
\mathbf{x}_{i},\mathbf{x}_{j})+\lambda _{\Bs{\theta}}^{2}(\mathbf{x}%
_{i},\mathbf{x}_{m})-\lambda _{\Bs{\theta}}^{2}(\mathbf{x}_{j},\mathbf{x%
}_{m}))\right) _{j,m\in B^{c}} \\
R_{B^{c},B,\Bs{\theta}}^{(i)} &=&\left( 2(\lambda _{\Bs{\theta}%
}^{2}(\mathbf{x}_{i},\mathbf{x}_{j})+\lambda _{\Bs{\theta}}^{2}(\mathbf{%
x}_{i},\mathbf{x}_{m})-\lambda _{\Bs{\theta}}^{2}(\mathbf{x}_{j},%
\mathbf{x}_{m}))\right) _{j\in B^{c},m\in B,m\neq i} \\
\mathbf{\tilde{y}}_{B^{c},i,\Bs{\theta}} &=&\left( \log \frac{y_{j}}{%
y_{i}}+2\lambda _{\Bs{\theta}}^{2}(\mathbf{x}_{i},\mathbf{x}%
_{j})\right) _{j\in B^{c},j\neq i} \\
\mu _{B,\Bs{\theta}}^{(i)} &=&R_{B^{c},B,\Bs{\theta}}^{(i)}\left(
R_{B,\Bs{\theta}}^{(i)}\right) ^{-1}\mathbf{\tilde{y}}_{B,i,\Bs{%
\theta }} \\
P_{B,\Bs{\theta}}^{(i)} &=&R_{B^{c},\Bs{\theta}}-R_{B^{c},B,%
\Bs{\theta}}^{(i)}\left( R_{B,\Bs{\theta}}^{(i)}\right)
^{-1}\left( R_{B^{c},B,\Bs{\theta}}^{(i)}\right) ^{\prime }.
\end{eqnarray*}%
Let $\Bs{\Theta}^{\left( \mathbf{y}\right) }=\mathbf{y}/\left\Vert 
\mathbf{y}\right\Vert $. It follows that, for $\pi \in \Pi$,%
\begin{eqnarray*}
&&\prod\limits_{B\in \pi }\frac{\partial ^{|B|}}{\partial \mathbf{y}_{B}}V_{%
\Bs{\theta}}(\mathbf{y}) \\
&=&\frac{1}{\prod_{i=1}^{M}y_{i}}\left\Vert \mathbf{y}\right\Vert ^{-|\pi
|}\prod\limits_{B\in \pi }\left( \frac{1}{|B|}\sum_{i\in B}\left( \Theta
_{i}^{\left( \mathbf{y}\right) }\right) ^{-1}\varphi _{|B|-1}\left( \mathbf{%
\tilde{\Theta}}_{B,i,\Bs{\theta}}^{\left( \mathbf{y}\right) };R_{B,%
\Bs{\theta}}^{(i)}\right) \Phi _{M-|B|}\left( \mathbf{\tilde{\Theta}}%
_{B^{c},i,\Bs{\theta}}^{\left( \mathbf{y}\right) }-\mu _{B,\Bs{%
\theta }}^{(i)};P_{B,\Bs{\theta}}^{(i)}\right) \right),
\end{eqnarray*}%
with%
\begin{eqnarray*}
\mathbf{\tilde{\Theta}}_{B,i,\Bs{\theta}}^{\left( \mathbf{y}\right) }
&=&\left( \log \left( \frac{\Theta _{j}^{\left( \mathbf{y}\right) }}{\Theta
_{i}^{\left( \mathbf{y}\right) }}\right) +2\lambda _{\Bs{\theta}}^{2}(%
\mathbf{x}_{i},\mathbf{x}_{j})\right) _{j\in B,j\neq i} \\
\mathbf{\tilde{\Theta}}_{B^{c},i,\Bs{\theta}}^{\left( \mathbf{y}\right)
} &=&\left( \log \left( \frac{\Theta _{j}^{\left( \mathbf{y}\right) }}{%
\Theta _{i}^{\left( \mathbf{y}\right) }}\right) +2\lambda _{\Bs{\theta}%
}^{2}(\mathbf{x}_{i},\mathbf{x}_{j})\right) _{j\in B^{c},j\neq i}.
\end{eqnarray*}%
Let%
\begin{equation*}
C\left( \pi ,\Bs{\Theta}^{\left( \mathbf{y}\right) },\Bs{\theta}%
\right) =\prod\limits_{B\in \pi }\left( \frac{1}{|B|}\sum_{i\in B}\left(
\Theta _{i}^{\left( \mathbf{y}\right) }\right) ^{-1}\varphi _{|B|-1}\left( 
\mathbf{\tilde{\Theta}}_{B,i,\Bs{\theta}}^{\left( \mathbf{y}\right)
};R_{B,\Bs{\theta}}^{(i)}\right) \Phi _{M-|B|}\left( \mathbf{\tilde{%
\Theta}}_{B^{c},i,\Bs{\theta}}^{\left( \mathbf{y}\right) }-\mu _{B,%
\Bs{\theta}}^{(i)};P_{B,\Bs{\theta}}^{(i)}\right) \right).
\end{equation*}%
We can note that $C\left( \pi ,\Bs{\Theta}^{\left( \mathbf{y}\right) },%
\Bs{\theta}\right) $ are positive and uniformly bounded for any
partition $\pi $, any $\Bs{\Theta}^{\left( \mathbf{y}\right) }\in 
\mathcal{S}_{+}=\{\mathbf{y}\in (0,\infty )^{M}:\left\Vert \mathbf{y}%
\right\Vert =1\}$, and any $\Bs{\theta}\in \mathcal{V}_{\Bs{\theta}%
_{0}}$. Then%
\begin{equation*}
\frac{\sum_{\pi \in \Pi }\left( -1\right) ^{|\pi |}\prod\limits_{B\in \pi }%
\frac{\partial ^{|B|}}{\partial \mathbf{y}_{B}}V_{\Bs{\theta}}(\mathbf{y%
})}{\sum_{\pi \in \Pi }\left( -1\right) ^{|\pi |}\prod\limits_{B\in \pi }%
\frac{\partial ^{|B|}}{\partial \mathbf{y}_{B}}V_{\Bs{\theta}_{0}}(%
\mathbf{y})}=\frac{\sum_{\pi \in \Pi }\left\Vert \mathbf{y}\right\Vert
^{-|\pi |}\left( -1\right) ^{|\pi |}C\left( \pi ,\Bs{\Theta}^{\left( 
\mathbf{y}\right) },\Bs{\theta}\right) }{\sum_{\pi \in \Pi }\left\Vert 
\mathbf{y}\right\Vert ^{-|\pi |}\left( -1\right) ^{|\pi |}C\left( \pi ,%
\Bs{\Theta}^{\left( \mathbf{y}\right) },\Bs{\theta}_{0}\right) }.
\end{equation*}%
The numerator and the denominator are polynomials of order $M$ in $%
\left\Vert \mathbf{y}\right\Vert^{-1} $. If $|\pi |=1$, then $\pi $ $%
=\{\{1,...,M\}\}$ and 
\begin{equation*}
C\left( \pi ,\Bs{\Theta}^{\left( \mathbf{y}\right) },\Bs{\theta}%
\right) =V_{\Bs{\theta}}(\Bs{\Theta}^{\left( \mathbf{y}\right) }).
\end{equation*}%
If $|\pi |=M$, then $\pi $ $=\{\{1\},...,\{M\}\}$ and%
\begin{equation*}
C\left( \pi ,\Bs{\Theta}^{\left( \mathbf{y}\right) },\Bs{\theta}%
\right) =\prod_{i=1}^{M}\left( \left( \Theta _{i}^{\left( \mathbf{y}\right)
}\right) ^{-1}\Phi _{M-1}\left( \mathbf{\tilde{\Theta}}_{\{1,...,M\}%
\backslash \{i\},i,\Bs{\theta}}^{\left( \mathbf{y}\right) }-\mu _{\{i\},%
\Bs{\theta}}^{(i)};P_{\{i\},\Bs{\theta}}^{(i)}\right) \right) .
\end{equation*}%
The ratio of both polynomials is bounded over any compact set included in $(0,\infty )$ and has limits as $\left\Vert \mathbf{y}\right\Vert$ tends to $0$ or $\infty$. Moreover these limits are also bounded for any $\Bs{\Theta}^{\left( \mathbf{y}\right) }\in  \mathcal{S}_{+}$. We can deduce that there exists a constant $C_{\mathcal{V}_{\Bs{\theta}%
_{0}}}$ such that%
\begin{equation*}
\sup_{\Bs{\theta}\in \mathcal{V}_{\Bs{\theta}_{0}}}\sup_{\mathbf{y}%
\in (0,\infty )^{M}}\left\vert \frac{\sum_{\pi \in \Pi }\left( -1\right)
^{|\pi |}\prod\limits_{B\in \pi }\frac{\partial ^{|B|}}{\partial \mathbf{y}%
_{B}}V_{\Bs{\theta}}(\mathbf{y})}{\sum_{\pi \in \Pi }\left( -1\right)
^{|\pi |}\prod\limits_{B\in \pi }\frac{\partial ^{|B|}}{\partial \mathbf{y}%
_{B}}V_{\Bs{\theta}_{0}}(\mathbf{y})}\right\vert <C_{\mathcal{V}_{%
\Bs{\theta}_{0}}}.
\end{equation*}%
It follows from \eqref{Eq_SupDerivPart} that%
\begin{equation*}
\sup_{j=1,...,L}\left\vert \frac{\partial f_{\Bs{\theta}}\left( \mathbf{%
y}\right) }{\partial \theta _{j}}\right\vert \leq A_{\mathcal{V}_{\Bs{%
\theta }_{0}}}C_{\mathcal{V}_{\Bs{\theta}_{0}}}\left( 1+\frac{1}{M}%
\sum_{i=1}^{M}y_{i}^{-1}\right) \left( \sum_{i=1}^{M}y_{i}^{-\alpha }\right)
\left\Vert \mathbf{y}\right\Vert ^{\alpha }\exp \left( \sum_{i=1}^{M}\frac{1%
}{y_{i}}B_{\mathcal{V}_{\Bs{\theta}_{0}}}\right) f_{\Bs{\theta}%
_{0}}\left( \mathbf{y}\right).
\end{equation*}%
Let us choose 
\begin{equation*}
\Psi \left( \mathbf{y}\right) =A_{\mathcal{V}_{\Bs{\theta}_{0}}}C_{%
\mathcal{V}_{\Bs{\theta}_{0}}}|H\left( \mathbf{y}\right) |\left( 1+%
\frac{1}{M}\sum_{i=1}^{M}y_{i}^{-1}\right) \left(
\sum_{i=1}^{M}y_{i}^{-\alpha }\right) \left\Vert \mathbf{y}\right\Vert
^{\alpha }\exp \left( -\sum_{i=1}^{M}\frac{1}{y_{i}}B_{\mathcal{V}_{\Bs{%
\theta }_{0}}}\right) f_{\Bs{\theta}_{0}}\left( \mathbf{y}\right) 
\end{equation*}%
to get%
\begin{equation*}
|H\left( \mathbf{y}\right) |\sup_{j=1,...,M}\left\vert \partial f_{\Bs{%
\theta }}\left( \mathbf{y}\right) /\partial \theta _{j}\right\vert \leq \Psi
\left( \mathbf{y}\right) 
\end{equation*}%
for all $\Bs{\theta}\in \mathcal{V}_{\Bs{\theta}_{0}}$ and almost
every $\mathbf{y}\in (0,\infty )^{M}$. The result follows by the dominated
convergence theorem. 

\end{proof}

\subsection{Proof of Theorem \protect\ref{Th_MainResultIPA}}

\label{Sec_Proof_Th_MainResultIPA}

Recall that $(U_{i},\mathbf{C}_{i})_{i\geq 1}$ are the points of a Poisson
point process on $(0,\infty )\times \mathbb{R}^{d}$ with intensity function $%
u^{-2} \mathrm{d}u \times \mathrm{d}\mathbf{c}$. For $i\geq 1$,
let $\varphi _{i,\Sigma }$ be the function from $\mathbb{R}^{d}$ to $\mathbb{%
R}$ defined by $\varphi _{i,\Sigma }\left( \mathbf{x}\right) =U_{i}\varphi
_{M}(\mathbf{x}-\mathbf{C}_{i},\Sigma )$.

Let us begin by noting that for each $j=1,\ldots,M$, the supremum $Y_{\Sigma }(%
\mathbf{x}_{j})=\bigvee_{i=1}^{\infty }U_{i}\varphi _{M}(\mathbf{x}_{j}-%
\mathbf{C}_{i},\Sigma )$ is a.s. attained by a unique function $\varphi
_{i,\Sigma }$ at $\mathbf{x}_{j}$.

\begin{Prop}
\label{Lem_Simplification_Smith} Let $\mathbf{x}_{1},\dots ,\mathbf{x}%
_{M}\in \mathbb{R}^{d}$, and define for $\mathbf{x}\in \mathbb{R}^{d}$, 
\begin{equation*}
\mathcal{I}_{\mathbf{x}}=\left\{ k:\bigvee_{i=1}^{\infty }U_{i}\varphi
_{M}(\mathbf{x}-\mathbf{C}_{i},\Sigma )=\varphi _{k,\Sigma }\left( \mathbf{%
x}\right) \right\} .
\end{equation*}%
Then a.s., we have $ |\mathcal{I}_{\mathbf{x}_{j}}|=1$ for all $j\geq 1$%
, where we recall that $|.|$ stands for the cardinality of a set.
\end{Prop}

\begin{proof}
Note that for every $\mathbf{x}\in \mathbb{R}^{d}$, $\left\{ \varphi
_{i,\Sigma }\left( \mathbf{x}\right) ,i=1,2,\ldots \right\} $ is a Poisson
random measure on $(0,\infty )$, which has no atoms. Therefore, the maxima
will have no ties with probability one.

\end{proof}

\begin{Rq}
The proof is also a direct consequence of Proposition 2.5 in \cite%
{Dombry2013}.
\end{Rq}

By Proposition \ref{Lem_Simplification_Smith}, we can define, for $j=1,\ldots,M$, the a.s. unique indexes $i_{\mathbf{x}_{j},\Sigma }$ satisfying 
\begin{equation*}
Y_{\Sigma }(\mathbf{x}_{j})=\varphi _{i_{\mathbf{x}_{j},\Sigma }}\left( 
\mathbf{x}_{j}\right) .
\end{equation*}

We now fix $\omega \in \Omega $ and consider one realization of the point
process $(U_{i},\mathbf{C}_{i})_{i\geq 1}$, denoted by $(U_{i}\left( \omega
\right) ,\mathbf{C}_{i}\left( \omega \right) )_{i\geq 1}$. For ease of
exposition, we do not mention $\omega $ further below. Let us consider the
set 
\begin{equation*}
\mathcal{B}=\left\{ \mathbf{x}\in \mathbb{R}^{d}:\exists j,k\geq 1,j\neq
k,\bigvee_{i=1}^{\infty }U_{i}\varphi _{M}(\mathbf{x}-\mathbf{C}%
_{i},\Sigma )=\varphi _{j,\Sigma }\left( \mathbf{x}\right) =\varphi
_{k,\Sigma }\left( \mathbf{x}\right) \right\} ,
\end{equation*}
i.e., the set of sites $\mathbf{x}\in \mathbb{R}^{d}$ for which the maximum $%
\bigvee_{i=1}^{\infty }U_{i}\varphi _{M}(\mathbf{x}-\mathbf{C}_{i},\Sigma
)$ is attained by at least two distinct functions $\varphi _{j,\Sigma }$
and $\varphi _{k,\Sigma }$.

\begin{Rq}
If $\Sigma =Id_{d}$ (identity matrix with dimension $d \times d$), the set $%
\mathcal{B}$ is the set of boundaries of the cells of a Poisson Laguerre
Tessellation \citep[e.g.,][]{Dombry2018}.
\end{Rq}

A key result is that $\mathcal{B}$ has a null Lebesgue measure.

\begin{Prop}
The set $\mathcal{B}$ has zero Lebesgue measure with probability one.
\end{Prop}

\begin{proof}
A point $\mathbf{x}\in \mathcal{B}$ is characterized as follows: 
\begin{align}
& \left. \mathbf{x}\in \mathcal{B}\right.  \notag \\
& \Leftrightarrow \exists j,k\geq 1,j\neq k:U_{k}\varphi _{M}(\mathbf{x}-%
\mathbf{C}_{k},\Sigma )=U_{j}\varphi _{M}(\mathbf{x}-\mathbf{C}_{j},\Sigma )
\notag \\
& \Leftrightarrow \exists j,k\geq 1,j\neq k:(\mathbf{x}-\mathbf{C}%
_{k})^{^{\prime }}\Sigma ^{-1}(\mathbf{x}-\mathbf{C}_{k})-2\log (U_{k})=(%
\mathbf{x}-\mathbf{C}_{j})^{^{\prime }}\Sigma ^{-1}(\mathbf{x}-\mathbf{C}%
_{j})-2\log (U_{j}).  \label{Eq_Equality_Storms}
\end{align}%
Let $\mathbf{h}\in \mathbb{R}^{d}$ such that $\mathbf{x}+\mathbf{h}$ is in a
neighbourhood of $\mathbf{x}$ but still belongs to $\mathcal{B}$, i.e., such
that 
\begin{equation*}
\left\Vert \mathbf{x}+\mathbf{h}-\mathbf{C}_{k}\right\Vert _{\Sigma
^{-1}}^{2}-2\log (U_{k})=\left\Vert \mathbf{x}+\mathbf{h}-\mathbf{C}%
_{j}\right\Vert _{\Sigma^{-1}}^{2}-2\log (U_{j}).
\end{equation*}%
Using \eqref{Eq_Equality_Storms}, we obtain that the previous equality is
equivalent to 
\begin{align*}
& \left. (\mathbf{x}+\mathbf{h}-\mathbf{C}_{k})^{^{\prime }}\Sigma ^{-1}(%
\mathbf{x}+\mathbf{h}-\mathbf{C}_{k})-2\log (U_{k})=(\mathbf{x}+\mathbf{h}-%
\mathbf{C}_{j})^{^{\prime }}\Sigma ^{-1}(\mathbf{x}+\mathbf{h}-\mathbf{C}%
_{j})-2\log (U_{j})\right. \\
& \Leftrightarrow (\mathbf{x}-\mathbf{C}_{k})^{^{\prime }}\Sigma ^{-1}(%
\mathbf{x}-\mathbf{C}_{k})+2(\mathbf{x}-\mathbf{C}_{k})^{^{\prime }}\Sigma
^{-1}\mathbf{h}+\mathbf{h^{^{\prime }}}\Sigma ^{-1}\mathbf{h}-2\log (U_{k})=(%
\mathbf{x}-\mathbf{C}_{j})^{^{\prime }}\Sigma ^{-1}(\mathbf{x}-\mathbf{C}%
_{j}) \\
& \ \ \ \ \ \ \ +2(\mathbf{x}-\mathbf{C}_{j})^{^{\prime }}\Sigma ^{-1}%
\mathbf{h}+\mathbf{h}^{^{\prime }}\Sigma ^{-1}\mathbf{h}-2\log (U_{j}) \\
& \Leftrightarrow (\mathbf{x}-\mathbf{C}_{k})^{^{\prime }}\Sigma ^{-1}%
\mathbf{h}=(\mathbf{x}-\mathbf{C}_{j})^{^{\prime }}\Sigma ^{-1}\mathbf{h} \\
& \Leftrightarrow (\mathbf{C}_{j}-\mathbf{C}_{k})^{^{\prime }}\Sigma ^{-1}%
\mathbf{h}=0.
\end{align*}%
Therefore, $\mathbf{x}+\mathbf{h}\in \mathcal{B}$ in a neighbourhood of $%
\mathbf{x}$ implies that $\mathbf{h}$ is orthogonal to the vector $(\mathbf{C%
}_{j}-\mathbf{C}_{k})$ for the inner product induced by $\Sigma ^{-1}$.
Thus, only one direction is suitable for $\mathbf{h}$. $\mathcal{B}$ is an
union of segments in $\mathbb{R}^{d}$. Moreover there is no ball around $%
\mathbf{x}$ belonging to $\mathcal{B}$, which implies that the interior of $%
\mathcal{B}$ is empty.

Let%
\begin{equation*}
\mathcal{B}_{j,k}=\left\{ \mathbf{y}\in \mathbb{R}^{d}:\inf_{\mathbf{x}\in 
\mathcal{B}}\left\Vert \mathbf{y}-\mathbf{x}\right\Vert <\frac{1}{j}\right\}
\cap \left\{ \mathbf{y}\in \mathbb{R}^{d}:\left\Vert \mathbf{y}\right\Vert
<k\right\} .
\end{equation*}%
It is easily seen that
\begin{equation*}
\lim_{j\rightarrow \infty }\mathcal{B}_{j,k}=\mathcal{B}\cap \left\{ \mathbf{%
y}\in \mathbb{R}^{d}:\left\Vert \mathbf{y}\right\Vert <k\right\}
\end{equation*}%
and hence that $\lim_{j\rightarrow \infty }\nu (\mathcal{B}_{j,k})=0$, where $\nu$ denotes the Lebesgue measure in $\Mbb{R}^d$. Therefore,
\begin{equation*}
\nu (\mathcal{B})=\lim_{k\rightarrow \infty }\nu \left(
\bigcup_{j=1}^{k}\left( \mathcal{B}\cap \left\{ \mathbf{y}\in \mathbb{R}%
^{d}:\left\Vert \mathbf{y}\right\Vert <j\right\} \right) \right) =0.\text{ }%
\end{equation*}
\end{proof}

Let $\mathbf{x}\in \mathbb{R}^{d}/\mathcal{B}$. We are now interested in the
existence of a neighbourhood of $\Sigma _{0}$ over which $i_{\mathbf{x}%
,\Sigma }$ is constant and thus $\Sigma \mapsto Z_{\Sigma }(\mathbf{x})$
becomes differentiable with respect to $\Sigma $.

\begin{Th}
\label{Lem_Derivability} Let $\mathbf{x}\in \mathbb{R}^{d}/\mathcal{B}$.
There exists a neighbourhood of $\Sigma _{0}$, $\mathcal{W}_{\Sigma _{0}}$,
such that $i_{\mathbf{x},\Sigma }$ is constant over this neighbourhood.
Moreover, the function $\Sigma \mapsto \log Y_{\Sigma }(\mathbf{x})$ is
differentiable over $\mathcal{W}_{\Sigma _{0}}$ and 
\begin{equation}
\left. \frac{\partial \log Y_{\Sigma }(\mathbf{x})}{\partial \Sigma }%
\right\vert _{\Sigma =\Sigma _{0}}=-\frac{1}{2}\left( \Sigma
_{0}^{-1}-\Sigma _{0}^{-1}(\mathbf{x}-\mathbf{C}_{i_{\mathbf{x}, \Sigma_{0}}})(\mathbf{x}-\mathbf{C}_{i_{\mathbf{x},\Sigma _{0}}})^{^{\prime
}}\Sigma _{0}^{-1}\right) .  \label{Eq_Deriv_Realis_Sig}
\end{equation}
\end{Th}

\begin{proof}
It follows from Proposition \ref{Lem_Simplification_Smith} that 
\begin{equation*}
Y_{\Sigma }(\mathbf{x})=\bigvee_{i=1}^{\infty }U_{i}\varphi _{M}(\mathbf{x}-%
\mathbf{C}_{i},\Sigma )=\varphi _{i_{\mathbf{x},\Sigma }}\left( \mathbf{x}%
\right)
\end{equation*}%
and, for $\mathbf{x}\in \mathbb{R}^{d}/\mathcal{B}$, we have, for all $j\neq
i_{\mathbf{x},\Sigma }$, 
\begin{equation*}
U_{i_{\mathbf{x},\Sigma }}\varphi _{M}(\mathbf{x}-\mathbf{C}_{i_{\mathbf{x}%
,\Sigma }},\Sigma )>U_{j}\varphi _{M}(\mathbf{x}-\mathbf{C}_{j},\Sigma ),
\end{equation*}%
or equivalently 
\begin{equation*}
2\log \left( U_{i_{\mathbf{x},\Sigma }}\right) -\left\Vert \mathbf{x}-%
\mathbf{C}_{i_{\mathbf{x},\Sigma }}\right\Vert _{\Sigma ^{-1}}^{2}>2\log
\left( U_{j}\right) -\left\Vert \mathbf{x}-\mathbf{C}_{j}\right\Vert
_{\Sigma ^{-1}}^{2},
\end{equation*}%
i.e., 
\begin{equation}
\left\Vert \mathbf{x}-\mathbf{C}_{j}\right\Vert _{\Sigma
^{-1}}^{2}-\left\Vert \mathbf{x}-\mathbf{C}_{i_{\mathbf{x},\Sigma
}}\right\Vert _{\Sigma ^{-1}}^{2}>2\log \left( U_{j}/U_{i_{\mathbf{x},\Sigma
}}\right) .  \label{Eq_Diff_Norms_Useful}
\end{equation}%
Let $\zeta >0$ and define 
\begin{eqnarray*}
\mathcal{I}_{1} &=&\left\{ j\neq i_{\mathbf{x},\Sigma }:\left\Vert \mathbf{x}%
-\mathbf{C}_{j}\right\Vert _{\Sigma ^{-1}}^{2}-\left\Vert \mathbf{x}-\mathbf{%
C}_{i_{\mathbf{x},\Sigma }}\right\Vert _{\Sigma ^{-1}}^{2}<\zeta \right\} ,
\\
\mathcal{I}_{2} &=&\left\{ j\neq i_{\mathbf{x},\Sigma }:\left\Vert \mathbf{x}%
-\mathbf{C}_{j}\right\Vert _{\Sigma ^{-1}}^{2}-\left\Vert \mathbf{x}-\mathbf{%
C}_{i_{\mathbf{x},\Sigma }}\right\Vert _{\Sigma ^{-1}}^{2}>\zeta ,2\log
(U_{j}/U_{i_{\mathbf{x},\Sigma }})>0\right\} , \\
\mathcal{I}_{3} &=&\left\{ j\neq i_{\mathbf{x},\Sigma }:\left\Vert \mathbf{x}%
-\mathbf{C}_{j}\right\Vert _{\Sigma ^{-1}}^{2}-\left\Vert \mathbf{x}-\mathbf{%
C}_{i_{\mathbf{x},\Sigma }}\right\Vert _{\Sigma ^{-1}}^{2}>\zeta >0>2\log
(U_{j}/U_{i_{\mathbf{x},\Sigma }})\right\} ,
\end{eqnarray*}%
such that $\mathcal{I}_{1}\cup \mathcal{I}_{2}\cup \mathcal{I}_{3}=\{j\geq
1,j\neq i_{\mathbf{x},\Sigma }\}$. Moreover it is easy to see, using the
definition of $i_{\mathbf{x},\Sigma }$ and the form of the intensity
function of the point process $(U_{i},\mathbf{C}_{i})_{i\geq 1}$, that $|\mathcal{I}_{1}|$ is finite, $|\mathcal{I}_{2}|$ is finite, but $|\mathcal{I}_{3}|$ is infinite.

Let $\Theta $ be a positive-definite matrix of size $d\times d$. We have,
for any $j\geq 1$, 
\begin{eqnarray*}
\left\Vert \mathbf{x}-\mathbf{C}_{j}\right\Vert _{\Sigma ^{-1}}^{2} &=&(%
\mathbf{x}-\mathbf{C}_{j})^{\prime }\Sigma ^{-1}(\mathbf{x}-\mathbf{C}_{j})
\\
&=&(\mathbf{x}-\mathbf{C}_{j})^{\prime }\left( \Sigma ^{-1}-\Theta
^{-1}\right) (\mathbf{x}-\mathbf{C}_{j})+\left\Vert \mathbf{x}-\mathbf{C}%
_{j}\right\Vert _{\Theta ^{-1}}^{2}.
\end{eqnarray*}%
In addition, 
\begin{eqnarray*}
|(\mathbf{x}-\mathbf{C}_{j})^{\prime }\left( \Sigma ^{-1}-\Theta
^{-1}\right) (\mathbf{x}-\mathbf{C}_{j})| &\leq &\left\Vert \mathbf{x}-%
\mathbf{C}_{j}\right\Vert \left\Vert \left( \Sigma ^{-1}-\Theta ^{-1}\right)
\left( \mathbf{x}-\mathbf{C}_{j}\right) \right\Vert \\
&\leq &\left\Vert \mathbf{x}-\mathbf{C}_{j}\right\Vert ^{2}\left\Vert \Sigma
^{-1}-\Theta ^{-1}\right\Vert .
\end{eqnarray*}%
Since $\left\Vert \cdot \right\Vert $ and $\left\Vert \mathbf{\cdot }%
\right\Vert _{\Sigma ^{-1}}$ are equivalent norms, there exists a positive
constant $D$ such that 
\begin{equation*}
|(\mathbf{x}-\mathbf{C}_{j})^{\prime }\left( \Sigma ^{-1}-\Theta
^{-1}\right) (\mathbf{x}-\mathbf{C}_{j})|\leq D\left\Vert \mathbf{x}-\mathbf{%
\ C}_{j}\right\Vert _{\Sigma ^{-1}}^{2}\left\Vert \Sigma ^{-1}-\Theta
^{-1}\right\Vert ,\quad \mathbf{x}\in \mathbb{R}^{d}.
\end{equation*}%
It follows that, for any $j\geq 1$, there exists a function $\mathbf{x}%
\mapsto a_{j}\left( \mathbf{x},\Sigma ,\Theta \right) $ such that 
\begin{equation}
\left\Vert \mathbf{x}-\mathbf{C}_{j}\right\Vert _{\Theta
^{-1}}^{2}=\left\Vert \mathbf{x}-\mathbf{C}_{j}\right\Vert _{\Sigma
^{-1}}^{2}\left( 1+a_{j}\left( \mathbf{x}-\mathbf{C}_{j},\Sigma ,\Theta
\right) \right)  \label{Eq_Def_aj}
\end{equation}%
and 
\begin{equation}
\sup_{\mathbf{x},\mathbf{C}_{j},j\geq 1}\left\vert a_{j}\left( \mathbf{x}-%
\mathbf{C}_{j},\Sigma ,\Theta \right) \right\vert \leq D\left\Vert \Sigma
^{-1}-\Theta ^{-1}\right\Vert .  \label{Eq_uniform_bound}
\end{equation}

Using \eqref{Eq_Def_aj}, we obtain 
\begin{eqnarray}
&&\left\Vert \mathbf{x}-\mathbf{C}_{j}\right\Vert _{\Theta
^{-1}}^{2}-\left\Vert \mathbf{x}-\mathbf{C}_{i_{\mathbf{x},\Sigma
}}\right\Vert _{\Theta ^{-1}}^{2}  \notag \\
&=&\left\Vert \mathbf{x}-\mathbf{C}_{j}\right\Vert _{\Sigma ^{-1}}^{2}\left(
1+a_{j}\left( \mathbf{x}-\mathbf{C}_{j},\Sigma ,\Theta \right) \right)
-\left\Vert \mathbf{x}-\mathbf{C}_{i_{\mathbf{x},\Sigma }}\right\Vert
_{\Sigma ^{-1}}^{2}\left( 1+a_{i_{\mathbf{x},\Sigma }}\left( \mathbf{x}-%
\mathbf{C}_{i_{\mathbf{x},\Sigma }},\Sigma ,\Theta \right) \right)  \notag \\
&=&\left\Vert \mathbf{x}-\mathbf{C}_{j}\right\Vert _{\Sigma
^{-1}}^{2}-\left\Vert \mathbf{x}-\mathbf{C}_{i_{\mathbf{x},\Sigma
}}\right\Vert _{\Sigma ^{-1}}^{2}+\left\Vert \mathbf{x}-\mathbf{C}%
_{j}\right\Vert _{\Sigma ^{-1}}^{2}a_{j}\left( \mathbf{x}-\mathbf{C}%
_{j},\Sigma ,\Theta \right)  \notag \\
&&-\left\Vert \mathbf{x}-\mathbf{C}_{i_{\mathbf{x},\Sigma }}\right\Vert
_{\Sigma ^{-1}}^{2}a_{i_{\mathbf{x},\Sigma }}\left( \mathbf{x}-\mathbf{C}%
_{i_{\mathbf{x},\Sigma }},\Sigma ,\Theta \right)  \notag \\
&=&\left( \left\Vert \mathbf{x}-\mathbf{C}_{j}\right\Vert _{\Sigma
^{-1}}^{2}-\left\Vert \mathbf{x}-\mathbf{C}_{i_{\mathbf{x},\Sigma
}}\right\Vert _{\Sigma ^{-1}}^{2}\right) \left( 1+a_{j}\left( \mathbf{x}-%
\mathbf{C}_{j},\Sigma ,\Theta \right) \right)  \notag \\
&&-\left\Vert \mathbf{x}-\mathbf{C}_{i_{\mathbf{x},\Sigma }}\right\Vert
_{\Sigma ^{-1}}^{2}\left( a_{i_{\mathbf{x},\Sigma }}\left( \mathbf{x}-%
\mathbf{C}_{i_{\mathbf{x},\Sigma }},\Sigma ,\Theta \right) -a_{j}\left( 
\mathbf{x}-\mathbf{C}_{j},\Sigma ,\Theta \right) \right) .
\label{Eq_Diff_Norms}
\end{eqnarray}%
Using \eqref{Eq_uniform_bound}, we see that there exists $\kappa >0$ such
that, for $\left\Vert \Sigma ^{-1}-\Theta ^{-1}\right\Vert <\kappa $, we
have, for all $j\geq 1$ such that $j\neq i_{\mathbf{x},\Sigma }$, that 
\begin{equation*}
\left\Vert \mathbf{x}-\mathbf{C}_{i_{\mathbf{x},\Sigma }}\right\Vert
_{\Sigma ^{-1}}^{2}|a_{i_{\mathbf{x},\Sigma }}\left( \mathbf{x}-\mathbf{C}%
_{i_{\mathbf{x},\Sigma }},\Sigma ,\Theta \right) -a_{j}\left( \mathbf{x}-%
\mathbf{C}_{j},\Sigma ,\Theta \right) |<\zeta /2,
\end{equation*}%
and, for all $j\in \mathcal{I}_{3}$, 
\begin{equation*}
\left( \left\Vert \mathbf{x}-\mathbf{C}_{j}\right\Vert _{\Sigma
^{-1}}^{2}-\left\Vert \mathbf{x}-\mathbf{C}_{i_{\mathbf{x},\Sigma
}}\right\Vert _{\Sigma ^{-1}}^{2}\right) \left( 1+a_{j}\left( \mathbf{x}-%
\mathbf{C}_{j},\Sigma ,\Theta \right) \right) >\zeta /2.
\end{equation*}%
Hence, using \eqref{Eq_Diff_Norms}, we obtain, for all $j\in \mathcal{I}_{3}$%
, 
\begin{equation}
\left\Vert \mathbf{x}-\mathbf{C}_{j}\right\Vert _{\Theta
^{-1}}^{2}-\left\Vert \mathbf{x}-\mathbf{C}_{i_{\mathbf{x},\Sigma
}}\right\Vert _{\Theta ^{-1}}^{2}>0>2\log (U_{j}/U_{i_{\mathbf{x},\Sigma }}).
\label{Eq_Before_Conc_1}
\end{equation}

Now, using \eqref{Eq_Diff_Norms_Useful}, the continuity of $\Sigma
^{-1}\mapsto \left\Vert \mathbf{\cdot }\right\Vert _{\Sigma ^{-1}}$ and the
fact that $|\mathcal{I}_{1}|$ and $|\mathcal{I}_{2}|$ are finite,
there exists $\kappa ^{\prime }>0$ such that, for $\left\Vert \Sigma
^{-1}-\Theta ^{-1}\right\Vert <\kappa ^{\prime }$, we have, for all $j\in $ $%
\mathcal{I}_{1}\cup \mathcal{I}_{2}$, that 
\begin{equation}
\left\Vert \mathbf{x}-\mathbf{C}_{j}\right\Vert _{\Theta
^{-1}}^{2}-\left\Vert \mathbf{x}-\mathbf{C}_{i_{\mathbf{x},\Sigma
}}\right\Vert _{\Theta ^{-1}}^{2}>2\log (U_{j}/U_{i_{\mathbf{x},\Sigma }}).
\label{Eq_Before_Conc_2}
\end{equation}

Combining \eqref{Eq_Before_Conc_1} and \eqref{Eq_Before_Conc_2}, we obtain
that, for all positive definite matrix $\Theta $ satisfying $\left\Vert
\Sigma ^{-1}-\Theta ^{-1}\right\Vert <\min \{\kappa ,\kappa ^{\prime }\}$,
we have, for all $j\geq 1$ such that $j\neq i_{\mathbf{x},\Sigma }$, 
\begin{equation*}
\left\Vert \mathbf{x}-\mathbf{C}_{j}\right\Vert _{\Theta
^{-1}}^{2}-\left\Vert \mathbf{x}-\mathbf{C}_{i_{\mathbf{x},\Sigma
}}\right\Vert _{\Theta ^{-1}}^{2}>2\log (U_{j}/U_{i_{\mathbf{x},\Sigma }}).
\end{equation*}%
Accordingly, $i_{\mathbf{x},\Theta }=i_{\mathbf{x},\Sigma }$. Hence we can
choose the neighbourhood of $\Sigma _{0}$ 
\begin{equation*}
\mathcal{W}_{\Sigma _{0}}=\left\{ \Theta \text{ positive definite}%
:\left\Vert \Sigma _{0}^{-1}-\Theta ^{-1}\right\Vert <\min \{\kappa ,\kappa
^{\prime }\}\right\} ,
\end{equation*}%
to define the derivative of $\Sigma \mapsto \log Y_{\Sigma }(\mathbf{x})$ at 
$\Sigma _{0}$.

We now compute the corresponding derivative. Using Proposition \ref%
{Lem_Simplification_Smith}, we have 
\begin{equation*}
\log Y_{\Sigma }(\mathbf{x})=\log \left( U_{i_{\mathbf{x},\Sigma }}\right) -%
\frac{d}{2}\log (2\pi )-\frac{1}{2}\log (\det (\Sigma ))-\frac{1}{2}(\mathbf{%
x}-\mathbf{C}_{i_{\mathbf{x},\Sigma }})^{^{\prime }}\Sigma ^{-1}(\mathbf{x}-%
\mathbf{C}_{i_{\mathbf{x},\Sigma }}),
\end{equation*}%
and hence 
\begin{equation*}
\frac{\partial \log Y_{\Sigma }(\mathbf{x})}{\partial \Sigma }=-\frac{1}{2}%
\left( \frac{\partial \log (\det (\Sigma ))}{\partial \Sigma }+\frac{%
\partial (\mathbf{x}-\mathbf{C}_{i_{\mathbf{x},\Sigma }})^{^{\prime }}\Sigma
^{-1}(\mathbf{x}-\mathbf{C}_{i_{\mathbf{x},\Sigma }})}{\partial \Sigma }%
\right) .
\end{equation*}%
Formula (11.7) in \cite{dwyer1967some} gives, for any symmetric matrix $%
\Sigma $, that 
\begin{equation}
\dfrac{\partial \log (\det (\Sigma ))}{\partial \Sigma }=\Sigma ^{-1}.
\label{Eq_Dwyer_1}
\end{equation}%
Moreover, since $i_{\mathbf{x},\Sigma }$ is constant over $\mathcal{W}%
_{\Sigma _{0}}$, Equation (11.8) in \cite{dwyer1967some} provides, for any
symmetric matrix $\Sigma $, 
\begin{equation}
\frac{\partial (\mathbf{x}-\mathbf{C}_{i_{\mathbf{x},\Sigma }})^{^{\prime
}}\Sigma ^{-1}(\mathbf{x}-\mathbf{C}_{i_{\mathbf{x},\Sigma }})}{\partial
\Sigma }=-\Sigma ^{-1}(\mathbf{x}-\mathbf{C}_{i_{\mathbf{x},\Sigma }})(%
\mathbf{x}-\mathbf{C}_{i_{\mathbf{x},\Sigma }})^{^{\prime }}\Sigma ^{-1}.
\label{Eq_Dwyer_2}
\end{equation}%
Combining \eqref{Eq_Dwyer_1} and \eqref{Eq_Dwyer_2}, we finally obtain 
\begin{equation*}
\left. \frac{\partial \log Y_{\Sigma }(\mathbf{x})}{\partial \Sigma }%
\right\vert _{\Sigma =\Sigma _{0}}=-\frac{1}{2}\left( \Sigma
_{0}^{-1}-\Sigma _{0}^{-1}(\mathbf{x}-\mathbf{C}_{i_{\mathbf{x}},_{\Sigma
_{0}}})(\mathbf{x}-\mathbf{C}_{i_{\mathbf{x}},_{\Sigma _{0}}})^{^{\prime
}}\Sigma _{0}^{-1}\right) .
\end{equation*}%

\end{proof}

We now prove that the derivative of $\Sigma \mapsto \log Y_{\Sigma }(\mathbf{x})$
can be uniformly bounded by an integrable random variable over a
neighbourhood of $\Sigma _{0}$.

\begin{Th}
\label{As_M1} There exists a non-random neighbourhood of $\Sigma _{0}$, $%
\mathcal{V}_{\Sigma _{0}}$, such that, for any $q>1$, there exists a random
variable $C_{\Sigma _{0}}(\mathbf{x},q)$ satisfying a.s. 
\begin{equation*}
\sup_{\Sigma \in \mathcal{V}_{\Sigma _{0}}}\left\Vert \frac{\partial \log
Y_{\Sigma }(\mathbf{x})}{\partial \Sigma }\right\Vert ^{q}\leq C_{\Sigma
_{0}}(\mathbf{x},q)
\end{equation*}%
and $\mathbb{E}\left[ C_{\Sigma _{0}}(\mathbf{x},q)\right] <\infty $.
\end{Th}

\begin{proof}
First recall that a.s. 
\begin{equation*}
\frac{\partial \log Y_{\Sigma }(\mathbf{x})}{\partial \Sigma }=-\frac{1}{2}%
\left( \Sigma ^{-1}-\Sigma ^{-1}(\mathbf{x}-\mathbf{C}_{i_{\mathbf{x},\Sigma
}})(\mathbf{x}-\mathbf{C}_{i_{\mathbf{x},\Sigma }})^{^{\prime }}\Sigma
^{-1}\right) ,
\end{equation*}%
which gives a.s. 
\begin{equation*}
\left\Vert \frac{\partial \log Y_{\Sigma }(\mathbf{x})}{\partial \Sigma }%
\right\Vert \leq \frac{1}{2}\left( \left\Vert \Sigma ^{-1}\right\Vert
+\left\Vert \Sigma ^{-1}(\mathbf{x}-\mathbf{C}_{i_{\mathbf{x},\Sigma }})(%
\mathbf{x}-\mathbf{C}_{i_{\mathbf{x},\Sigma }})^{^{\prime }}\Sigma
^{-1}\right\Vert \right) .
\end{equation*}%
Consequently, using the well-known fact that, for all $a,b\in \mathbb{R}$
and $q\geq 1$, $|a-b|^{q}\leq 2^{q-1}(|a|^{q}+|b|^{q})$, we obtain 
\begin{equation}
\left\Vert \frac{\partial \log Y_{\Sigma }(\mathbf{x})}{\partial \Sigma }%
\right\Vert ^{q}\leq \frac{1}{2}\left( \left\Vert \Sigma ^{-1}\right\Vert
^{q}+\left\Vert \Sigma ^{-1}(\mathbf{x}-\mathbf{C}_{i_{\mathbf{x},\Sigma }})(%
\mathbf{x}-\mathbf{C}_{i_{\mathbf{x},\Sigma }})^{^{\prime }}\Sigma
^{-1}\right\Vert ^{q}\right) .  \label{Eq_Norm_Powerq_Diff_ZSigma}
\end{equation}

Let $A=\Sigma ^{-1}(\mathbf{x}-\mathbf{C}_{i_{\mathbf{x},\Sigma }})$, such
that $\Sigma ^{-1}(\mathbf{x}-\mathbf{C}_{i_{\mathbf{x},\Sigma }})(\mathbf{x}%
-\mathbf{C}_{i_{\mathbf{x},\Sigma }})^{^{\prime }}\Sigma ^{-1}=AA^{^{\prime
}}$. Note that $AA^{^{\prime }}$ is a non-negative symmetric matrix of rank $%
1$ and thus it only has one positive eigenvalue given by 
\begin{equation*}
\lambda =A^{^{\prime }}A=\Vert \mathbf{x}-\mathbf{C}_{i_{\mathbf{x},\Sigma
}}\Vert _{\Sigma ^{-2}}^{2}
\end{equation*}%
(since we have that $(AA^{^{\prime }})A=A(A^{^{\prime }}A)=(A^{^{\prime
}}A)A $). It follows that 
\begin{equation}
\Vert AA^{^{\prime }}\Vert =\Vert \mathbf{x}-\mathbf{C}_{i_{\mathbf{x}%
,\Sigma }}\Vert _{\Sigma ^{-2}}^{2}.  \label{Eq_Majoration_Norm}
\end{equation}%
For a positive definite symmetric matrix $\Theta $, we denote by $\lambda
_{\max }\left( \Theta \right) $ and $\lambda _{\min }\left( \Theta \right) $
respectively the maximum and the minimum of its positive eigenvalues. We
have that 
\begin{equation*}
\Vert \mathbf{x}-\mathbf{C}_{i_{\mathbf{x},\Sigma }}\Vert _{\Sigma
^{-1}}^{2}=(\mathbf{x}-\mathbf{C}_{i_{\mathbf{x},\Sigma }})^{^{\prime
}}\Sigma ^{-1}(\mathbf{x}-\mathbf{C}_{i_{\mathbf{x},\Sigma }})\geq \lambda
_{\min }(\Sigma ^{-1})\Vert \mathbf{x}-\mathbf{C}_{i_{\mathbf{x},\Sigma
}}\Vert ^{2},
\end{equation*}%
and 
\begin{equation*}
\Vert \mathbf{x}-\mathbf{C}_{i_{\mathbf{x},\Sigma }}\Vert _{\Sigma
^{-2}}^{2}=(\mathbf{x}-\mathbf{C}_{i_{\mathbf{x},\Sigma }})^{^{\prime
}}\Sigma ^{-2}(\mathbf{x}-\mathbf{C}_{i_{\mathbf{x},\Sigma }})\leq \lambda
_{\max }\left( \Sigma ^{-2}\right) \Vert \mathbf{x}-\mathbf{C}_{i_{\mathbf{x}%
,\Sigma }}\Vert ^{2},
\end{equation*}%
which yield 
\begin{equation}
\Vert \mathbf{x}-\mathbf{C}_{i_{\mathbf{x},\Sigma }}\Vert _{\Sigma
^{-2}}^{2}\leq \frac{\left[ \lambda _{\max }(\Sigma ^{-1})\right] ^{2}}{%
\lambda _{\min }(\Sigma ^{-1})}\Vert \mathbf{x}-\mathbf{C}_{i_{\mathbf{x}%
,\Sigma }}\Vert _{\Sigma ^{-1}}^{2}.  \label{Eq_Quadratic_Form}
\end{equation}%
Combining \eqref{Eq_Norm_Powerq_Diff_ZSigma}, \eqref{Eq_Majoration_Norm} and %
\eqref{Eq_Quadratic_Form}, we have, for any $q>1$, that a.s. 
\begin{equation*}
\sup_{\Sigma \in \mathcal{V}_{\Sigma _{0}}}\left\Vert \frac{\partial \log
Y_{\Sigma }(\mathbf{x})}{\partial \Sigma }\right\Vert ^{q}\leq C_{\Sigma
_{0}}(\mathbf{x},q),
\end{equation*}%
where 
\begin{equation}
C_{\Sigma _{0}}(\mathbf{x},q)=\frac{1}{2}\left( \sup_{\Sigma \in \mathcal{V}%
_{\Sigma _{0}}}\left\Vert \Sigma ^{-1}\right\Vert ^{q}+\sup_{\Sigma \in 
\mathcal{V}_{\Sigma _{0}}}\left( \frac{\left[ \lambda _{\max }(\Sigma ^{-1})%
\right] ^{2}}{\lambda _{\min }(\Sigma ^{-1})}\right) ^{q}\sup_{\Sigma \in 
\mathcal{V}_{\Sigma _{0}}}\Vert \mathbf{x}-\mathbf{C}_{i_{\mathbf{x},\Sigma
}}\Vert _{\Sigma ^{-1}}^{2q}\right) .  \label{Eq_Bound_Norm}
\end{equation}

In order to control $C_{\Sigma _{0}}(\mathbf{x},q)$, it is sufficient to
control $\Vert \mathbf{x}-\mathbf{C}_{i_{\mathbf{x},\Sigma }}\Vert _{\Sigma
^{-1}}^{2q}$. As the center of the \textquotedblleft
storm\textquotedblright\ realizing the maximum at point $\mathbf{x}$, $%
\mathbf{C}_{i_{\mathbf{x},\Sigma }}$ is characterized by 
\begin{align}
& \quad \ \ U_{i_{\mathbf{x},\Sigma }}\ \varphi _{M}(\mathbf{x}-\mathbf{C}%
_{i_{\mathbf{x},\Sigma }},\Sigma )\geq U_{i}\ \varphi _{M}(\mathbf{x}-%
\mathbf{C}_{i},\Sigma ),\ \forall i\geq 1  \notag \\
& \Leftrightarrow \log (U_{i_{\mathbf{x},\Sigma }})-\frac{1}{2}(\mathbf{x}-%
\mathbf{C}_{i_{\mathbf{x},\Sigma }})^{^{\prime }}\Sigma ^{-1}(\mathbf{x}-\mathbf{C}_{i_{\mathbf{x},\Sigma }})\geq \log (U_{i})-\frac{1}{2}(\mathbf{x
}-\mathbf{C}_{i})^{^{\prime }}\Sigma ^{-1}(\mathbf{x}-\mathbf{C}_{i}),\
\forall i\geq 1  \notag \\
& \Leftrightarrow \left\Vert \mathbf{x}-\mathbf{C}_{i_{\mathbf{x},\Sigma
}}\right\Vert _{\Sigma ^{-1}}^{2}\leq 2\log (U_{i_{\mathbf{x},\Sigma
}})-2\log (U_{i})+\left\Vert \mathbf{x}-\mathbf{C}_{i}\right\Vert _{\Sigma
^{-1}}^{2},\ \forall i\geq 1,  \label{Eq_Characterization_Center}
\end{align}%
and therefore 
\begin{equation}
\sup_{\Sigma \in \mathcal{V}_{\Sigma _{0}}}\left\Vert \mathbf{x}-\mathbf{C}%
_{i_{\mathbf{x},\Sigma }}\right\Vert _{\Sigma ^{-1}}^{2}\leq 2\sup_{\Sigma
\in \mathcal{V}_{\Sigma _{0}}}\log (U_{i_{\mathbf{x},\Sigma }})+\sup_{\Sigma
\in \mathcal{V}_{\Sigma _{0}}}\left\Vert \mathbf{x}-\mathbf{\ C}%
_{i}\right\Vert _{\Sigma ^{-1}}^{2}-2\log (U_{i}),\ \forall i\geq 1.
\label{Eq_Majoration_Sup_Norm_ixsig}
\end{equation}%
Additionally, we have, for all $i\geq 1$, 
\begin{equation*}
\left\Vert \mathbf{x}-\mathbf{C}_{i}\right\Vert _{\Sigma
^{-1}}^{2}=\left\Vert \mathbf{x}-\mathbf{C}_{i}\right\Vert _{\Sigma
_{0}^{-1}}^{2}+(\mathbf{x}-\mathbf{C}_{i})^{\prime }\left( \Sigma
^{-1}-\Sigma _{0}^{-1}\right) (\mathbf{x}-\mathbf{C}_{i}),
\end{equation*}%
which yields, by equivalence of the norms $\Vert .\Vert $ and $\Vert .\Vert
_{\Sigma ^{-1}}$, that 
\begin{eqnarray}
\sup_{\Sigma \in \mathcal{V}_{\Sigma _{0}}}\left\Vert \mathbf{x}-\mathbf{C}%
_{i}\right\Vert _{\Sigma ^{-1}}^{2} &\leq &\left\Vert \mathbf{x}-\mathbf{C}%
_{i}\right\Vert _{\Sigma _{0}^{-1}}^{2}+\sup_{\Sigma \in \mathcal{V}_{\Sigma
_{0}}}\left\Vert \Sigma ^{-1}-\Sigma _{0}^{-1}\right\Vert \left\Vert \mathbf{x}-\mathbf{C}_{i}\right\Vert ^{2}  \notag \\
&\leq &\left( D_{0}+\sup_{\Sigma \in \mathcal{V}_{\Sigma _{0}}}\left\Vert
\Sigma ^{-1}-\Sigma _{0}^{-1}\right\Vert \right) \left\Vert \mathbf{x}-%
\mathbf{C}_{i}\right\Vert ^{2}  \label{Eq_Majoration_Sup_Norm_i}
\end{eqnarray}%
for some positive constant $D_{0}$. Hence let us now choose $\mathcal{V}%
_{\Sigma _{0}}$ such that 
\begin{equation}
\sup_{\Sigma \in \mathcal{V}_{\Sigma _{0}}}\left\Vert \Sigma ^{-1}-\Sigma
_{0}^{-1}\right\Vert <\infty .  \label{Eq_First_Condition_Nu0}
\end{equation}

Now, for two real-valued random variables $U$ and $V$, $U+V\geq \lambda $
implies that $U\geq \lambda /2$ or $V\geq \lambda /2$, giving $\mathbb{P}%
(U+V\geq \lambda )\leq \mathbb{P}\left( U\geq \lambda /2\right) +\mathbb{P}%
\left( V\geq \lambda /2\right) .$ Thus, using %
\eqref{Eq_Majoration_Sup_Norm_ixsig} and \eqref{Eq_Majoration_Sup_Norm_i},
we obtain 
\begin{align}
& \quad \ \mathbb{P}\left( \sup_{\Sigma \in \mathcal{V}_{\Sigma
_{0}}}\left\Vert \mathbf{x}-\mathbf{C}_{i_{\mathbf{x},\Sigma }}\right\Vert
_{\Sigma ^{-1}}^{2}\geq \lambda \right)  \notag \\
& \leq \mathbb{P}\left( \left( D_{0}+\sup_{\Sigma \in \mathcal{V}_{\Sigma
_{0}}}\left\Vert \Sigma ^{-1}-\Sigma _{0}^{-1}\right\Vert \right) \left\Vert 
\mathbf{x}-\mathbf{C}_{i}\right\Vert ^{2}-2\log (U_{i})+2\sup_{\Sigma \in 
\mathcal{V}_{\Sigma _{0}}}\log (U_{i_{\mathbf{x},\Sigma }})\geq \lambda \
\forall i\geq 1\right)  \notag \\
& \leq \mathbb{P}\left( \left( D_{0}+\sup_{\Sigma \in \mathcal{V}_{\Sigma
_{0}}}\left\Vert \Sigma ^{-1}-\Sigma _{0}^{-1}\right\Vert \right) \left\Vert 
\mathbf{x}-\mathbf{C}_{i}\right\Vert ^{2}-2\log (U_{i})\geq \frac{\lambda }{2%
}\ \forall i\geq 1\right)  \notag \\
& \quad +\mathbb{P}\left( 2\sup_{\Sigma \in \mathcal{V}_{\Sigma _{0}}}\log
(U_{i_{\mathbf{x},\Sigma }})\geq \frac{\lambda }{2}\right) .
\label{Eq_Majoration_Probability}
\end{align}

We first deal with the first term of the right-hand side of %
\eqref{Eq_Majoration_Probability}. Since the $(U_{i},\mathbf{C}_{i})_{i\geq
1}$ are the points of a Poisson process on $(0,\infty )\times \mathbb{R}^{d}$
with intensity function $u^{-2} \mathrm{d}u \times \mathrm{d}
\mathbf{c}$, we have 
\begin{align}
& \quad \ \mathbb{P}\left( \left( D_{0}+\sup_{\Sigma \in \mathcal{V}_{\Sigma
_{0}}}\left\Vert \Sigma ^{-1}-\Sigma _{0}^{-1}\right\Vert \right) \left\Vert 
\mathbf{x}-\mathbf{C}_{i}\right\Vert ^{2}-2\log (U_{i})\geq \frac{\lambda }{2%
}\ \forall i\geq 1\right)  \notag \\
& =\exp \left( -\mu \left\{ (u,\mathbf{c})\in (0,\infty )\times \mathbb{R}%
^{d}:\left( D_{0}+\sup_{\Sigma \in \mathcal{V}_{\Sigma _{0}}}\left\Vert
\Sigma ^{-1}-\Sigma _{0}^{-1}\right\Vert \right) \left\Vert \mathbf{x}-%
\mathbf{c}\right\Vert ^{2}-2\log (u)<\frac{\lambda }{2}\right\} \right) ,
\label{Eq_Exp_First_Probability}
\end{align}%
where 
\begin{align*}
& \quad \ \mu \left\{ (u,\mathbf{c})\in (0,\infty )\times \mathbb{R}%
^{d}:\left( D_{0}+\sup_{\Sigma \in \mathcal{V}_{\Sigma _{0}}}\left\Vert
\Sigma ^{-1}-\Sigma _{0}^{-1}\right\Vert \right) \left\Vert \mathbf{x}-%
\mathbf{c}\right\Vert ^{2}-2\log (u)<\frac{\lambda }{2}\right\} \\
& =\int_{e^{-\frac{\lambda }{4}}}^{\infty }\left( \int_{\left\Vert \mathbf{x}%
-\mathbf{c}\right\Vert ^{2}\leq \left( \frac{\lambda }{2}+2\log (u
)\right) \left( D_{0}+\sup_{\Sigma \in \mathcal{V}_{\Sigma _{0}}}\left\Vert
\Sigma ^{-1}-\Sigma _{0}^{-1}\right\Vert \right) ^{-1}}\ \mathrm{d}%
\mathbf{c}\right) u^{-2} \mathrm{d}u \\
& =\frac{\pi ^{d/2}}{\left( D_{0}+\sup_{\Sigma \in \mathcal{V}_{\Sigma
_{0}}}\left\Vert \Sigma ^{-1}-\Sigma _{0}^{-1}\right\Vert \right)
^{d/2}\Gamma \left( d/2+1\right) }\int_{e^{-\frac{\lambda }{4}}}^{\infty
}\left( \frac{\lambda }{2}+2\log (u)\right) ^{d/2}\ u^{-2} \mathrm{d}u.
\end{align*}

Making the change of variable $v=\lambda /2+2\log (u)\mbox{,
yielding }u =\exp \left( v/2-\lambda /4\right) $ and $\mathrm{d}
u=\exp \left( v/2-\lambda /4\right) \mathrm{d}v/2$, we obtain 
\begin{eqnarray}
&&\mu \left\{ (u,\mathbf{c})\in (0,\infty )\times \mathbb{R}^{d}:\left\Vert 
\mathbf{x}-\mathbf{c}\right\Vert _{\Sigma }^{2}-2\log (u)<\frac{\lambda }{2}%
\right\}  \notag \\
&=&\frac{\pi ^{d/2}\exp \left( \lambda /4\right) }{\left( D_{0}+\sup_{\Sigma
\in \mathcal{V}_{\Sigma _{0}}}\left\Vert \Sigma ^{-1}-\Sigma
_{0}^{-1}\right\Vert \right) ^{d/2}\Gamma \left( d/2+1\right) }%
\int_{0}^{\infty }v^{d/2}\exp \left( -\frac{v}{2}\right) \mathrm{d}v.
\label{Eq_Measure_Eq}
\end{eqnarray}%
The latter integral is finite and it follows from %
\eqref{Eq_Exp_First_Probability} that 
\begin{align}
& \qquad \ \mathbb{P}\left( \left( D_{0}+\sup_{\Sigma \in \mathcal{V}%
_{\Sigma _{0}}}\left\Vert \Sigma ^{-1}-\Sigma _{0}^{-1}\right\Vert \right)
\left\Vert \mathbf{x}-\mathbf{C}_{i}\right\Vert ^{2}-2\log (U_{i})\geq \frac{%
\lambda }{2},\ \forall i\geq 1\right)  \notag \\
& =\exp \left( -\left[ \frac{\pi ^{d/2}}{\left( D_{0}+\sup_{\Sigma \in 
\mathcal{V}_{\Sigma _{0}}}\left\Vert \Sigma ^{-1}-\Sigma
_{0}^{-1}\right\Vert \right) ^{d/2}\Gamma \left( d/2+1\right) }%
\int_{0}^{\infty }v^{d/2}\exp \left( -\frac{v}{2}\right) \mathrm{d}v
\right] \exp \left( \frac{\lambda }{4}\right) \right) .
\label{Eq_First_Term}
\end{align}

Let us now deal with the second term of the right-hand side of %
\eqref{Eq_Majoration_Probability}. Observe that 
\begin{equation*}
\mathbb{P}\left( 2\sup_{\Sigma \in \mathcal{V}_{\Sigma _{0}}}\log (U_{i_{%
\mathbf{x},\Sigma }})\geq \frac{\lambda }{2}\right) =\mathbb{P}\left(
\inf_{\Sigma \in \mathcal{V}_{\Sigma _{0}}}U_{i_{\mathbf{x},\Sigma
}}^{-1}\leq \exp \left( -\frac{\lambda }{4}\right) \right) .
\end{equation*}%
The mapping $u\rightarrow u^{-1}$ applied to the points of the Poisson point
process yields a new Poisson point process with intensity function $\mathrm{d}u$. Furthermore, such a Poisson point process on $(0,\infty )$ is
homogeneous\ and can be represented as the sum of independent standard
exponential random variables. Thus, we can write $\min_{i\geq
1}U_{i}^{-1} {\overset{d}{=}}\mathrm{Exp}(1)$. Moreover, $\inf_{\Sigma \in \mathcal{V}%
_{\Sigma _{0}}}U_{i_{\mathbf{x},\Sigma }}^{-1}\geq \min_{i\geq
1} U_{i}^{-1} $, implying 
\begin{align*}
\mathbb{P}\left( \inf_{\Sigma \in \mathcal{V}_{\Sigma _{0}}}U_{i_{\mathbf{x
},\Sigma }}^{-1}\leq \exp \left( -\frac{\lambda }{4}\right) \right) & \leq 
\mathbb{P}\left( \min_{i\geq 1} U_{i}^{-1} \leq \exp \left( -%
\frac{\lambda }{4}\right) \right) \\
& =1-\exp \left( -\exp \left( -\frac{\lambda }{4}\right) \right) \\
& \underset{\lambda \rightarrow \infty }{\sim }\exp \left( -\frac{\lambda }{4%
}\right) ,
\end{align*}%
which gives for large $\lambda $ 
\begin{equation}
\mathbb{P}\left( \sup_{\Sigma \in \mathcal{V}_{\Sigma _{0}}}2\log (U_{i_{%
\mathbf{x},\Sigma }})\geq \frac{\lambda }{2}\right) \leq \exp \left( -\frac{%
\lambda }{4}\right) .  \label{Eq_Second_Term}
\end{equation}

Now, for any $q>1$, we have 
\begin{equation*}
\mathbb{E}\left[ \sup_{\Sigma \in \mathcal{V}_{\Sigma _{0}}}\Vert \mathbf{x}-%
\mathbf{C}_{i_{\mathbf{x},\Sigma }}\Vert _{\Sigma ^{-1}}^{2q}\right]
=\int_{0}^{\infty }\mathbb{P}\left( \sup_{\Sigma \in \mathcal{\ V}_{\Sigma
_{0}}}\Vert \mathbf{x}-\mathbf{C}_{i_{\mathbf{x},\Sigma }}\Vert _{\Sigma
^{-1}}^{2q}\geq u\right) \mathrm{d}u.
\end{equation*}%
We carry out the change of variable $\lambda =u^{1/q}$, which gives $%
u=\lambda ^{q}$ and hence $\mathrm{d}u=q\lambda ^{q-1}\mathrm{d}\lambda $.
Accordingly, we obtain 
\begin{equation*}
\mathbb{E}\left[ \sup_{\Sigma \in \mathcal{V}_{\Sigma _{0}}}\Vert \mathbf{x}-%
\mathbf{C}_{i_{\mathbf{x},\Sigma }}\Vert _{\Sigma ^{-1}}^{2q}\right]
=q\int_{0}^{\infty }\lambda ^{q-1}\mathbb{P}\left( \sup_{\Sigma \in \mathcal{%
\ V}_{\Sigma _{0}}}\Vert \mathbf{x}-\mathbf{C}_{i_{\mathbf{x},\Sigma }}\Vert
_{\Sigma ^{-1}}^{2}\geq \lambda \right) \mathrm{d}\lambda ,
\end{equation*}%
and therefore, using \eqref{Eq_Majoration_Probability}, \eqref{Eq_First_Term}
and \eqref{Eq_Second_Term}, 
\begin{equation*}
\mathbb{E}\left[ \sup_{\Sigma \in \mathcal{V}_{\Sigma _{0}}}\Vert \mathbf{x}-%
\mathbf{C}_{i_{\mathbf{x},\Sigma }}\Vert _{\Sigma ^{-1}}^{2q}\right] <\infty
.
\end{equation*}

Finally, let us recall that 
\begin{equation*}
C_{\Sigma _{0}}(\mathbf{x},q)=\frac{1}{2}\left( \sup_{\Sigma \in \mathcal{V}%
_{\Sigma _{0}}}\left\Vert \Sigma ^{-1}\right\Vert ^{q}+\sup_{\Sigma \in 
\mathcal{V}_{\Sigma _{0}}}\left( \frac{\left[ \lambda _{\max }(\Sigma ^{-1})%
\right] ^{2}}{\lambda _{\min }(\Sigma ^{-1})}\right) ^{q}\sup_{\Sigma \in 
\mathcal{V}_{\Sigma _{0}}}\Vert \mathbf{x}-\mathbf{C}_{i_{\mathbf{x},\Sigma
}}\Vert _{\Sigma ^{-1}}^{2q}\right) .
\end{equation*}%
Consequently, using \eqref{Eq_First_Condition_Nu0}, we finally deduce that
any neighbourhood of $\Sigma _{0}$, $\mathcal{V}_{\Sigma _{0}}$, satisfying 
\begin{eqnarray*}
\sup_{\Sigma \in \mathcal{V}_{\Sigma _{0}}}\left\Vert \Sigma
^{-1}\right\Vert ^{q} &<&\infty , \\
\sup_{\Sigma \in \mathcal{V}_{\Sigma _{0}}} \frac{\left[ \lambda
_{\max }(\Sigma ^{-1})\right] ^{2}}{\lambda _{\min }(\Sigma ^{-1})}
&<&\infty , \\
\sup_{\Sigma \in \mathcal{V}_{\Sigma _{0}}}\left\Vert \Sigma ^{-1}-\Sigma
_{0}^{-1}\right\Vert &<&\infty ,
\end{eqnarray*}%
leads to 
\begin{equation*}
\mathbb{E}\left[ C_{\Sigma _{0}}(\mathbf{x},q)\right] <\infty. 
\end{equation*}

\end{proof}

We finally provide the proof of Theorem \ref{Th_MainResultIPA}.

\begin{proof}
Let us choose $\mathcal{V}_{\Sigma _{0}}$ as in the proof of Theorem \ref%
{As_M1}. It follows from \eqref{Eq_Chain_Rule} that 
\begin{equation*}
\left\Vert \frac{\partial H\left( \mathbf{Y}_{\Sigma }\right) }{\partial
\Sigma }\right\Vert \leq \sum_{i=1}^{M}\left\vert Y_{\Sigma }(\mathbf{x}_{i})%
\frac{\partial H\left( \mathbf{Y}_{\Sigma }\right) }{\partial y_{i}}%
\right\vert \left\Vert \frac{\partial \log Y_{\Sigma }(\mathbf{x}_{i})}{%
\partial \Sigma }\right\Vert ,
\end{equation*}%
which gives 
\begin{equation*}
\sup_{\Sigma \in \mathcal{V}_{\Sigma _{0}}}\left\Vert \frac{\partial H\left( 
\mathbf{Y}_{\Sigma }\right) }{\partial \Sigma }\right\Vert \leq
\sum_{i=1}^{M}\sup_{\Sigma \in \mathcal{V}_{\Sigma _{0}}}\left\vert
Y_{\Sigma }(\mathbf{x}_{i})\frac{\partial H\left( \mathbf{Y}_{\Sigma
}\right) }{\partial y_{i}}\right\vert \sup_{\Sigma \in \mathcal{V}_{\Sigma
_{0}}}\left\Vert \frac{\partial \log Y_{\Sigma }(\mathbf{x}_{i})}{\partial
\Sigma }\right\Vert .
\end{equation*}%
Hence we choose 
\begin{equation*}
B_{\Sigma _{0}}=\sum_{i=1}^{M}\sup_{\Sigma \in \mathcal{V}_{\Sigma
_{0}}}\left\vert Y_{\Sigma }(\mathbf{x}_{i})\frac{\partial H\left( \mathbf{Y}%
_{\Sigma }\right) }{\partial y_{i}}\right\vert \sup_{\Sigma \in \mathcal{V}%
_{\Sigma _{0}}}\left\Vert \frac{\partial \log Y_{\Sigma }(\mathbf{x}_{i})}{%
\partial \Sigma }\right\Vert .
\end{equation*}%
We have 
\begin{equation*}
\mathbb{E}\left[ B_{\Sigma _{0}}\right] =\sum_{i=1}^{M}\mathbb{E}\left[
\sup_{\Sigma \in \mathcal{V}_{\Sigma _{0}}}\left\vert Y_{\Sigma }(\mathbf{x}%
_{i})\frac{\partial H\left( \mathbf{Y}_{\Sigma }\right) }{\partial y_{i}}%
\right\vert \sup_{\Sigma \in \mathcal{V}_{\Sigma _{0}}}\left\Vert \frac{%
\partial \log Y_{\Sigma }(\mathbf{x}_{i})}{\partial \Sigma }\right\Vert %
\right] .
\end{equation*}%
Let $q>1$ such that $p^{-1}+q^{-1}=1$. By H\"{o}lder inequality, we have 
\begin{eqnarray*}
&&\mathbb{E}\left[ \sup_{\Sigma \in \mathcal{V}_{\Sigma _{0}}}\left\vert
Y_{\Sigma }(\mathbf{x}_{i})\frac{\partial H\left( \mathbf{Y}_{\Sigma
}\right) }{\partial y_{i}}\right\vert \sup_{\Sigma \in \mathcal{V}_{\Sigma
_{0}}}\left\Vert \frac{\partial \log Y_{\Sigma }(\mathbf{x}_{i})}{\partial
\Sigma }\right\Vert \right] \\
&\leq &\mathbb{E}\left[ \sup_{\Sigma \in \mathcal{V}_{\Sigma
_{0}}}\left\vert Y_{\Sigma }(\mathbf{x}_{i})\frac{\partial H\left( \mathbf{Y}%
_{\Sigma }\right) }{\partial y_{i}}\right\vert ^{p}\right] ^{1/p}\mathbb{E}%
\left[ \sup_{\Sigma \in \mathcal{V}_{\Sigma _{0}}}\left\Vert \frac{\partial
\log Y_{\Sigma }(\mathbf{x}_{i})}{\partial \Sigma }\right\Vert ^{q}\right]
^{1/q}.
\end{eqnarray*}%
By Theorem \ref{As_M1} and the assumptions, the result follows. 
\end{proof}

%

\subsection{Proof of Proposition \protect\ref{Prop_Validity_Assumptions}}

\subsubsection{LRM for the Brown--Resnick random field}

Let us first recall that $M=2$ and that $H$ is given by \eqref{Eq_Exp_R_CorrWind_H} with $\beta _{i}\in \mathbb{N}_{\ast }$, $i=1,2$%
. Let us now note that $\left\Vert \mathbf{Y}_{\Bs{\theta}%
_{0}}\right\Vert ^{\alpha }\leq 2^{-\alpha /2}\left( Y_{\Bs{\theta}%
_{0},1}\vee Y_{\Bs{\theta}_{0},2}\right) ^{\alpha }\leq 2^{-\alpha
/2}\left( Y_{\Bs{\theta}_{0},1}^{\alpha }+Y_{\Bs{\theta}%
_{0},2}^{\alpha }\right) $. Therefore there exists a positive constant $C$
such that 
\begin{eqnarray*}
&&|H\left( \mathbf{Y}_{\Bs{\theta}_{0}}\right) |\left( 1+\frac{1}{2}%
\sum_{i=1}^{2}Y_{\Bs{\theta}_{0},i}^{-1}\right) \left( \sum_{i=1}^{2}Y_{%
\Bs{\theta}_{0},i}^{-\alpha }\right) \exp \left( -B_{\mathcal{V}_{%
\Bs{\theta}_{0}}}\sum_{i=1}^{2}Y_{\Bs{\theta}_{0},i}^{-1}\right)
\left\Vert \mathbf{Y}_{\Bs{\theta}_{0}}\right\Vert ^{\alpha } \\
&\leq &C\sum_{\delta _{1}\in \Delta _{1},\delta _{2}\in \Delta _{2}}\left(
Y_{\Bs{\theta}_{0},1}^{\delta _{1}}\exp \left( -B_{\mathcal{V}_{\mathbf{%
\theta }_{0}}}Y_{\Bs{\theta}_{0},1}^{-1}\right) \right) \left( Y_{%
\Bs{\theta}_{0},2}^{\delta _{2}}\exp \left( -B_{\mathcal{V}_{\mathbf{%
\theta }_{0}}}Y_{\Bs{\theta}_{0},2}^{-1}\right) \right) 
\end{eqnarray*}%
where, for $i=1,2$, $\Delta _{i}$ is a set of constants whose the maximal
value is $\xi _{i}\beta _{i}+\alpha $. Moreover,
\begin{eqnarray*}
&&\mathbb{E}\left[ \left( Y_{\Bs{\theta}_{0},1}^{\delta _{1}}\exp
\left( -B_{\mathcal{V}_{\Bs{\theta}_{0}}}Y_{\Bs{\theta}%
_{0},1}^{-1}\right) \right) \left( Y_{\Bs{\theta}_{0},2}^{\delta
_{2}}\exp \left( -B_{\mathcal{V}_{\Bs{\theta}_{0}}}Y_{\Bs{\theta}%
_{0},2}^{-1}\right) \right) \right]  \\
&\leq &\left( \mathbb{E}\left[ \left( Y_{\Bs{\theta}_{0},1}^{2\delta
_{1}}\exp \left( -2B_{\mathcal{V}_{\Bs{\theta}_{0}}}Y_{\Bs{\theta}%
_{0},1}^{-1}\right) \right) \right] \mathbb{E}\left[ \left( Y_{\mathbf{%
\theta }_{0},2}^{2\delta _{2}}\exp \left( -2B_{\mathcal{V}_{\Bs{\theta}%
_{0}}}Y_{\Bs{\theta}_{0},2}^{-1}\right) \right) \right] \right) ^{1/2}.
\end{eqnarray*}%
Since $Y_{\Bs{\theta}_{0},i}$ has a standard Fr\'{e}chet distribution,
the expectation $\mathbb{E}\left[ \left( Y_{\Bs{\theta}_{0},i}^{2\delta
_{i}}\exp \left( -2B_{\mathcal{V}_{\Bs{\theta}_{0}}}Y_{\Bs{\theta}%
_{0},i}^{-1}\right) \right) \right] $ is finite, if $2\left( \xi _{i}\beta
_{i}+\alpha \right) <1$ and $-2B_{\mathcal{V}_{\Bs{\theta}_{0}}}<1$.
But it is possible to choose $\alpha $ and $\mathcal{V}_{\Bs{\theta}%
_{0}}$ to satisfy such constraints since $\xi _{i}\beta _{i}<1/2$ and $B_{%
\mathcal{V}_{\Bs{\theta}_{0}}}$ is defined by \eqref{Eq_B}.

Finally note that the derivatives $\partial \lambda _{\Bs{\theta}}(%
\mathbf{x}_{1},\mathbf{x}_{2})/\partial \psi $ and $\partial \lambda _{%
\Bs{\theta}}(\mathbf{x}_{1},\mathbf{x}_{2})/\partial \kappa $ are
easily obtained and therefore it is straightforward to conclude that the
condition in \eqref{Eq_bound_der_lambda} also holds with such a choice
of neighbourhood $\mathcal{V}_{\Bs{\theta}_{0}}$.

\subsubsection{IPA for the Smith random field}

From \eqref{Eq_Exp_R_CorrWind_H}, there exist positive constants $C_{i}$, $%
i=1,2$, such that%
\begin{equation*}
\left\vert Y_{\Sigma }(\mathbf{x}_{i})\frac{\partial H\left( \mathbf{Y}%
_{\Sigma }\right) }{\partial y_{i}}\right\vert \leq C_{i}\left( Y_{\Sigma }(%
\mathbf{x}_{1})\vee 1\right) ^{\xi _{1}\beta _{1}}\left( Y_{\Sigma }(\mathbf{%
x}_{2})\vee 1\right) ^{\xi _{2}\beta _{2}}.
\end{equation*}

Therefore, we have for $p>1$, 
\begin{equation}
\sup_{\Sigma \in \mathcal{V}_{\Sigma _{0}}}\left\vert Y_{\Sigma }(\mathbf{x}%
_{i})\frac{\partial H\left( \mathbf{Y}_{\Sigma }\right) }{\partial y_{i}}%
\right\vert ^{p}\leq C_{i}\sup_{\Sigma \in \mathcal{V}_{\Sigma _{0}}}\left(
Y_{\Sigma }(\mathbf{x}_{1})\vee 1\right) ^{p\xi _{1}\beta _{1}}\sup_{\Sigma
\in \mathcal{V}_{\Sigma _{0}}}\left( Y_{\Sigma }(\mathbf{x}_{2})\vee
1\right) ^{p\xi _{2}\beta _{2}}.  \label{Eq_Bound_sup_derivative}
\end{equation}%
Next result will allow us to prove that Condition %
\eqref{Eq_Condition_Main_Theorem} holds.

\begin{Prop}
\label{Pr_Uniform_bound}There exists a non-random neighbourhood of $\Sigma
_{0}$, $\mathcal{V}_{\Sigma _{0}}$, such that, for any $\beta <1$ and $%
\mathbf{x\in }\mathbb{R}^{d}$, 
\begin{equation*}
\mathbb{E}\left[ \sup_{\Sigma \in \mathcal{V}_{\Sigma _{0}}}\left( Y_{\Sigma
}(\mathbf{x})\vee 1\right) ^{\beta }\right] <\infty .
\end{equation*}
\end{Prop}

\begin{proof}
We have 
\begin{equation*}
Y_{\Sigma }(\mathbf{x})=\bigvee_{i=1}^{\infty }U_{i}\varphi _{M}(\mathbf{x}-%
\mathbf{C}_{i},\Sigma )\leq \frac{1}{\pi ^{d/2}\det \left( \Sigma \right)
^{1/2}}\bigvee_{i=1}^{\infty }U_{i}
\end{equation*}%
and thus 
\begin{equation*}
\sup_{\Sigma \in \mathcal{V}_{\Sigma _{0}}}\left( Y_{\Sigma }(\mathbf{x}%
)\vee 1\right) ^{\beta }\leq \sup_{\Sigma \in \mathcal{V}_{\Sigma
_{0}}}\left( \frac{1}{\pi ^{d/2}\det \left( \Sigma \right) ^{1/2}}\right)
^{\beta }\left( \bigvee_{i=1}^{\infty }U_{i}\vee \pi ^{d/2}\sup_{\Sigma \in 
\mathcal{V}_{\Sigma _{0}}}\det \left( \Sigma \right) ^{1/2}\right) ^{\beta }.
\end{equation*}%
It is well-know that $\bigvee_{i=1}^{\infty }U_{i}$ has a standard Fr\'{e}%
chet distribution and therefore $\mathbb{E[}(\bigvee_{i=1}^{\infty
}U_{i})^{\beta }]<\infty $ since $\beta <1$. We deduce that it suffices to
choose $\mathcal{V}_{\Sigma _{0}}$ such that $\sup_{\Sigma \in \mathcal{V}%
_{\Sigma _{0}}}\det \left( \Sigma \right) ^{-1/2}<\infty $ and $\sup_{\Sigma
\in \mathcal{V}_{\Sigma _{0}}}\det \left( \Sigma \right) ^{1/2}<\infty $.
This is possible since any symmetric invertible matrix $A$ admits a
neighbourhood $\mathcal{V}$ of matrices such that, for all $\tilde{A}\in 
\mathcal{V}$, $\det (A)-\varepsilon <\det (\tilde{A})<\det (A)+\varepsilon $%
, where $0<\varepsilon <\det (A)$. 
\end{proof}

By Proposition \ref{Pr_Uniform_bound}, $\left( \ref{Eq_Bound_sup_derivative}%
\right) $ and the H\"{o}lder inequality, there exists a non-random
neighbourhood of $\Sigma _{0}$, $\mathcal{V}_{\Sigma _{0}}$, such that for
some $p>1$ satisfying $2p\xi _{1}\beta _{1}<1$ and $2p\xi _{2}\beta _{2}<1$,
we have 
\begin{equation*}
\mathbb{E}\left[ \sup_{\Sigma \in \mathcal{V}_{\Sigma _{0}}}\left\vert
Y_{\Sigma }(\mathbf{x}_{i})\frac{\partial H\left( \mathbf{Y}_{\Sigma
}\right) }{\partial y_{i}}\right\vert ^{p}\right] \leq C_{i}\mathbb{E}\left[
\sup_{\Sigma \in \mathcal{V}_{\Sigma _{0}}}\left( Y_{\Sigma }(\mathbf{x}%
_{1})\vee 1\right) ^{2p\xi _{1}\beta _{1}}\right] ^{1/2}\mathbb{E}\left[
\sup_{\Sigma \in \mathcal{V}_{\Sigma _{0}}}\left( Y_{\Sigma }(\mathbf{x}%
_{2})\vee 1\right) ^{2p\xi _{2}\beta _{2}}\right] ^{1/2}<\infty .
\end{equation*}

\section{Analytical formulas}

\subsection{Expressions for the derivatives of $R(\Bs{\protect\theta })$
when $H$ is given by \eqref{Eq_Exp_R_CorrWind_H}}

\label{Sec_Anal_Deriv}

For $\tilde{\beta _{1}},\tilde{\beta _{2}}<1/2$, we introduce the function $%
g_{\tilde{\beta _{1}},\tilde{\beta _{2}}}$ defined by 
\begin{equation}
g_{\tilde{\beta _{1}},\tilde{\beta _{2}}}(h)=\left\{ 
\begin{array}{ll}
\Gamma (1-\tilde{\beta _{1}}-\tilde{\beta _{2}}), & \mbox{if}\quad h=0, \\ 
\displaystyle\int_{0}^{\infty }t^{\tilde{\beta _{2}}}\Big[C_{2}(t,h)\
C_{1}(t,h)^{\tilde{\beta _{1}}+\tilde{\beta _{2}}-2}\ \Gamma (2-\tilde{\beta
_{1}}-\tilde{\beta _{2}}) &  \\ 
\qquad +C_{3}(t,h)\ C_{1}(t,h)^{\tilde{\beta _{1}}+\tilde{\beta _{2}}-1}\
\Gamma (1-\tilde{\beta _{1}}-\tilde{\beta _{2}})\Big]\ \mathrm{d}t, & %
\mbox{if}\quad h>0,%
\end{array}%
\right.   \label{Eq_Def_g_beta1_beta2}
\end{equation}%
where, for $t,h>0$, 
\begin{align*}
C_{1}(t,h)& =\Phi \left( \frac{h}{2}+\frac{\log (t)}{h}\right) +\frac{1}{t}%
\Phi \left( \frac{h}{2}-\frac{\log (t)}{h}\right) , \\
C_{2}(t,h)& =\left[ \Phi \left( \frac{h}{2}+\frac{\log \left( t\right) }{h}%
\right) +\frac{1}{h}\varphi \left( \frac{h}{2}+\frac{\log (t)}{h}\right) -%
\frac{1}{ht}\varphi \left( \frac{h}{2}-\frac{\log \left( t\right) }{h}%
\right) \right]  \\
& \quad \ \times \left[ \frac{1}{t^{2}}\Phi \left( \frac{h}{2}-\frac{\log (t)%
}{h}\right) +\frac{1}{ht^{2}}\varphi \left( \frac{h}{2}-\frac{\log (t)}{h}%
\right) -\frac{1}{ht}\varphi \left( \frac{h}{2}+\frac{\log (t)}{h}\right) %
\right] , \\
C_{3}(t,h)& =\frac{1}{h^{2}t}\left( \frac{h}{2}-\frac{\log (t)}{h}\right) \
\varphi \left( \frac{h}{2}+\frac{\log (t)}{h}\right) +\frac{1}{h^{2}t^{2}}%
\left( \frac{h}{2}+\frac{\log (t)}{h}\right) \varphi \left( \frac{h}{2}-%
\frac{\log \left( t\right) }{h}\right) ,
\end{align*}%
with $\Phi $ and $\varphi $ denoting the standard univariate Gaussian
distribution and density functions, respectively.

Let $\mathbf{x}_{1},\mathbf{x}_{2}\in \mathbb{R}^{2}$ be our sites of
interest and $Y_{\Bs{\theta}}$ be a Brown--Resnick random field with
dependence parameters $\psi$ and $\kappa$ gathered in $\Bs{\theta}$, i.e., $\Bs{\theta}=\left( \psi ,\kappa
\right) ^{\prime }$. Let $\Bs{\theta}%
_{0}$ be the parameter at which we wish to compute the derivative of $R(%
\Bs{\theta})$. As before, let $\eta _{i},\tau _{i},\xi _{i}$ be the
location, scale and shape parameters of $X_{\Bs{\theta}}(\mathbf{x}%
_{i}) $, $i=1,2$, and $\beta _{i}=\beta (\mathbf{x}_{i})\in \mathbb{N}_{\ast
}$ such that $\beta _{i}\xi _{i}<1/2$. Theorem 2 in \cite%
{koch2018spatialpowers} yields 
\begin{align}
& \quad \ \mbox{Corr}\left( X^{\beta
_{1}}_{\Bs{\theta}}(\mathbf{x}_{1}),X^{\beta _{2}}_{\Bs{\theta}}(\mathbf{x}_{2})\right)  \notag \\
& =\frac{1}{\sqrt{D_{\beta _{1},\eta _{1},\tau _{1},\xi _{1}}D_{\beta
_{2},\eta _{2},\tau _{2},\xi _{2}}}}\Bigg[\sum_{k_{1}=0}^{\beta
_{1}}\sum_{k_{2}=0}^{\beta _{2}}B_{k_{1},\beta _{1},\eta _{1},\tau _{1},\xi
_{1},k_{2},\beta _{2},\eta _{2},\tau _{2},\xi _{2}}\ g_{(\beta
_{1}-k_{1})\xi _{1},(\beta _{2}-k_{2})\xi _{2}}\left( \sqrt{2\gamma _{%
\Bs{\theta}}(\mathbf{x}_{2}-\mathbf{x}_{1})}\right)  \notag \\
& \quad -\sum_{k_{1}=0}^{\beta _{1}}\sum_{k_{2}=0}^{\beta
_{2}}B_{k_{1},\beta _{1},\eta _{1},\tau _{1},\xi _{1},k_{2},\beta _{2},\eta
_{2},\tau _{2},\xi _{2}}\ \Gamma (1-[\beta _{1}-k_{1}]\xi _{1})\Gamma
(1-[\beta _{2}-k_{2}]\xi _{2})\Bigg],  \label{Eq_DetailedExprDepenMeasWind}
\end{align}%
where 
\begin{equation*}
B_{k_{1},\beta _{1},\eta _{1},\tau _{1},\xi _{1},k_{2},\beta _{2},\eta
_{2},\tau _{2},\xi _{2}}={\binom{\beta _{1}}{k_{1}}}\left( \eta _{1}-\frac{%
\tau _{1}}{\xi _{1}}\right) ^{k_{1}}\left( \frac{\tau _{1}}{\xi _{1}}\right)
^{\beta _{1}-k_{1}}{\binom{\beta _{2}}{k_{2}}}\left( \eta _{2}-\frac{\tau
_{2}}{\xi _{2}}\right) ^{k_{2}}\left( \frac{\tau _{2}}{\xi _{2}}\right)
^{\beta _{2}-k_{2}}.
\end{equation*}%
This yields 
\begin{equation}
\frac{\partial \mbox{Corr}\left( X^{\beta
_{1}}_{\Bs{\theta}}(\mathbf{x}_{1}),X^{\beta _{2}}_{\Bs{\theta}}(\mathbf{x}_{2})\right) }{\partial 
\Bs{\theta}}=\frac{\sum_{k_{1}=0}^{\beta _{1}}\sum_{k_{2}=0}^{\beta
_{2}}B_{k_{1},\beta _{1},\eta _{1},\tau _{1},\xi _{1},k_{2},\beta _{2},\eta
_{2},\tau _{2},\xi _{2}}\ \frac{\partial g_{(\beta _{1}-k_{1})\xi
_{1},(\beta _{2}-k_{2})\xi _{2}}\left( \sqrt{2\gamma _{\Bs{\theta}}(%
\mathbf{x}_{2}-\mathbf{x}_{1})}\right) }{\partial \Bs{\theta}}\Big \vert%
_{\Bs{\theta}=\Bs{\theta}_{0}}}{\sqrt{D_{\beta _{1},\eta _{1},\tau
_{1},\xi _{1}}D_{\beta _{2},\eta _{2},\tau _{2},\xi _{2}}}}.
\label{Eq_Deriv_Risk_Measure}
\end{equation}

We obtain from \eqref{Eq_Def_g_beta1_beta2} that 
\begin{align}
& \quad \frac{\partial g_{\tilde{\beta _{1}},\tilde{\beta _{2}}}\left( \sqrt{%
2\gamma _{\Bs{\theta}}(\mathbf{x}_{2}-\mathbf{x}_{1})}\right) }{%
\partial \Bs{\theta}}\Big \vert_{\Bs{\theta}=\Bs{\theta}_{0}}
\notag \\
& =\int_{0}^{\infty }\frac{\partial }{\partial h}\Big (t^{\tilde{\beta _{2}}}%
\Big[C_{2}(t,h)\ C_{1}(t,h)^{\tilde{\beta _{1}}+\tilde{\beta _{2}}-2}\
\Gamma (2-\tilde{\beta _{1}}-\tilde{\beta _{2}})+C_{3}(t,h)\ C_{1}(t,h)^{%
\tilde{\beta _{1}}+\tilde{\beta _{2}}-1}\ \Gamma (1-\tilde{\beta _{1}}-%
\tilde{\beta _{2}})\Big]\Big)\Big \vert_{h=h_{0}}\mathrm{d}t  \notag
\\
& \times \frac{\partial \sqrt{2\gamma _{\Bs{\theta}}(\mathbf{x}_{2}-%
\mathbf{x}_{1})}}{\partial \Bs{\theta}}\Big \vert_{\Bs{\theta}=\mathbf{%
\theta }_{0}},  \label{Eq_Deriv_gs_Numerical_study}
\end{align}%
where $h_{0}=\sqrt{2\gamma _{\Bs{\theta}_{0}}(\mathbf{x}_{2}-\mathbf{x}%
_{1})}$. The term 
\begin{equation*}
\frac{\partial }{\partial h}\Big (t^{\tilde{\beta _{2}}}\Big[C_{2}(t,h)\
C_{1}(t,h)^{\tilde{\beta _{1}}+\tilde{\beta _{2}}-2}\ \Gamma (2-\tilde{\beta
_{1}}-\tilde{\beta _{2}})+C_{3}(t,h)\ C_{1}(t,h)^{\tilde{\beta _{1}}+\tilde{%
\beta _{2}}-1}\ \Gamma (1-\tilde{\beta _{1}}-\tilde{\beta _{2}})\Big]\Big)%
\Big \vert_{h=h_{0}}
\end{equation*}%
has a closed (although complicated) expression (available upon request). The corresponding integral can be computed using numerical
methods such as adaptive quadrature. 
Finally, we obtain the true values of $\partial R(\Bs{\theta})/\partial 
\Bs{\theta}|_{\Bs{\theta}=\Bs{\theta}_{0}}$ by plugging the
appropiate values of \eqref{Eq_Deriv_gs_Numerical_study} in %
\eqref{Eq_Deriv_Risk_Measure}.

Now, if $Y_{\Sigma}$ is the Smith random field with covariance
matrix $\Sigma $ and $\Sigma _{0}$ is a symmetric positive-definite matrix
at which we want to compute the derivative of $R\left( \Bs{\theta}%
\right) $, exactly the same formulas and procedure can be applied upon
replacement of $\Bs{\theta}$ with $\Sigma $, $\Bs{\theta}_{0}$
with $\Sigma _{0}$, and $\gamma _{\Bs{\theta}}(\mathbf{x}_{2}-\mathbf{x}%
_{1})$ with $(\mathbf{x}_{2}-\mathbf{x}_{1})^{\prime }\Sigma ^{-1}(\mathbf{x}%
_{2}-\mathbf{x}_{1})/2$.

\subsection{Proportionality of the components of the score function for the
Brown--Resnick random field}

\label{Sect_Proport_BR}

The bivariate density of the (simple) Brown--Resnick random field (at $\mathbf{%
x}_{1}$ and $\mathbf{x}_{2} \in \Mbb{R}^2$) satisfies, for $y_{1},y_{2}>0$, 
\begin{align}
f_{\Bs{\theta}}(y_{1},y_{2})=\exp \left( -\frac{\Phi (w)}{y_{1}}-\frac{%
\Phi (v)}{y_{2}}\right) \times \bigg[\left( \frac{\Phi (w)}{y_{1}^{2}}+\frac{%
\varphi (w)}{hy_{1}^{2}}-\frac{\varphi (v)}{hy_{1}y_{2}}\right) & \times
\left( \frac{\Phi (v)}{y_{2}^{2}}+\frac{\varphi (v)}{hy_{2}^{2}}-\frac{%
\varphi (w)}{hy_{1}y_{2}}\right)   \notag \\
& +\left( \frac{v\varphi (w)}{h^{2}y_{1}^{2}y_{2}}+\frac{w\varphi (v)}{%
h^{2}y_{1}y_{2}^{2}}\right) \bigg],
\label{Chapter_RiskMeasure_Smith_Bivariate_Density}
\end{align}%
where 
\begin{equation*}
h=\sqrt{2\gamma _{\Bs{\theta}}(\mathbf{x}_{2}-\mathbf{x}_{1})}=\sqrt{2}%
\left( \frac{\Vert \mathbf{x}_{2}-\mathbf{x}_{1}\Vert }{\kappa }\right)
^{\psi /2},\quad w=\frac{h}{2}+\frac{\log \left( y_{2}/y_{1}\right) }{h}%
\quad \mbox{and}\quad v=\frac{h}{2}-\frac{\log \left( y_{2}/y_{1}\right) }{h}%
.
\end{equation*}%
This easily yields 
\begin{equation}
\frac{\partial \log f_{\Bs{\theta}}(y_{1},y_{2})}{\partial \psi }\Big
/\frac{\partial \log f_{\Bs{\theta}}(y_{1},y_{2})}{\partial \kappa }=%
\frac{\partial h}{\partial \psi }\Big/\frac{\partial h}{\partial \kappa }=-%
\frac{\kappa }{\psi }\log \left( \frac{\Vert \mathbf{x}_{2}-\mathbf{x}%
_{1}\Vert }{\kappa }\right) ,  \label{Eq_FacPropor}
\end{equation}%
and it follows from \eqref{Eq_ExpressionDerivativeLRM} that the LRM
estimates of $\partial R(\Bs{\theta})/\partial \psi $ and $\partial R(%
\Bs{\theta})/\partial \kappa $ are proportional by the factor given in
the right-hand side of \eqref{Eq_FacPropor}.

\newpage 

\bibliographystyle{apalike} 
\bibliography{References_Erwan}

\begin{thebibliography}{}

\bibitem[Asadi et~al., 2015]{asadi2015extremes}
Asadi, P., Davison, A.~C., and Engelke, S. (2015).
\newblock Extremes on river networks.
\newblock {\em The Annals of Applied Statistics}, 9(4):2023--2050.
\newblock \newline \url{https://doi.org/10.1214/15-AOAS863}.

\bibitem[Asmussen and Glynn, 2010]{asmussen2010}
Asmussen, S. and Glynn, P.~W. (2010).
\newblock {\em Stochastic Simulation: Algorithms and Analysis}.
\newblock Springer-Verlag New York.

\bibitem[Broadie and Glasserman, 1996]{broadie1996estimating}
Broadie, M. and Glasserman, P. (1996).
\newblock Estimating security price derivatives using simulation.
\newblock {\em Management Science}, 42(2):269--285.

\bibitem[Brown and Resnick, 1977]{brown1977extreme}
Brown, B.~M. and Resnick, S.~I. (1977).
\newblock Extreme values of independent stochastic processes.
\newblock {\em Journal of Applied Probability}, 14(4):732--739.
\newblock \newline \url{https://doi.org/10.2307/3213346}.

\bibitem[Chen and Fu, 2001]{chen2001efficient}
Chen, J. and Fu, M.~C. (2001).
\newblock Efficient sensitivity analysis of mortgage backed securities.
\newblock In {\em 12th Annual Derivatives Securities Conference, New York}.

\bibitem[Coles, 2001]{coles2001introduction}
Coles, S. (2001).
\newblock {\em An {I}ntroduction to {S}tatistical {M}odeling of {E}xtreme
  {V}alues}.
\newblock Springer London.

\bibitem[Davison et~al., 2012]{davison2012statistical}
Davison, A.~C., Padoan, S.~A., and Ribatet, M. (2012).
\newblock Statistical modeling of spatial extremes.
\newblock {\em Statistical Science}, 27(2):161--186.
\newblock \newline \url{https://doi.org/10.1214/11-STS376}.

\bibitem[de~Haan, 1984]{haan1984spectral}
de~Haan, L. (1984).
\newblock A spectral representation for max-stable processes.
\newblock {\em The Annals of Probability}, 12(4):1194--1204.
\newblock \newline \url{https://doi.org/10.1214/aop/1176993148}.

\bibitem[de~Haan and Ferreira, 2006]{de2007extreme}
de~Haan, L. and Ferreira, A. (2006).
\newblock {\em Extreme {V}alue {T}heory: {A}n {I}ntroduction}.
\newblock Springer-Verlag New York.
\newblock \newline \url{https://doi.org/10.1007/0-387-34471-3}.

\bibitem[Dombry et~al., 2016]{dombry2016exact}
Dombry, C., Engelke, S., and Oesting, M. (2016).
\newblock Exact simulation of max-stable processes.
\newblock {\em Biometrika}.

\bibitem[Dombry et~al., 2017]{dombry2017asymp}
Dombry, C., Engelke, S., and Oesting, M. (2017).
\newblock Asymptotic properties of the maximum likelihood estimator for
  multivariate extreme value distributions.
\newblock {\em arXiv preprint arXiv:1612.05178}.

\bibitem[Dombry and Eyi-Minko, 2013]{Dombry2013}
Dombry, C. and Eyi-Minko, F. (2013).
\newblock Regular conditional distributions of continuous max-infinitely
  divisible random fields.
\newblock {\em Electronic Journal of Probability}, 18(0).

\bibitem[Dombry et~al., 2013]{dombry2013conditional}
Dombry, C., Eyi-Minko, F., and Ribatet, M. (2013).
\newblock Conditional simulation of max-stable processes.
\newblock {\em Biometrika}, 100(1):111--124.
\newblock \newline \url{https://doi.org/10.1093/biomet/ass067}.

\bibitem[Dombry and Kabluchko, 2018]{Dombry2018}
Dombry, C. and Kabluchko, Z. (2018).
\newblock Random tessellations associated with max-stable random fields.
\newblock {\em Bernoulli}, 24(1):30--52.

\bibitem[Dwyer, 1967]{dwyer1967some}
Dwyer, P.~S. (1967).
\newblock Some applications of matrix derivatives in multivariate analysis.
\newblock {\em Journal of the American Statistical Association},
  62(318):607--625.

\bibitem[Emanuel, 2005]{emanuel2005increasing}
Emanuel, K. (2005).
\newblock Increasing destructiveness of tropical cyclones over the past 30
  years.
\newblock {\em Nature}, 436(4):686--688.
\newblock \newline \url{https://doi.org/10.1038/nature03906}.

\bibitem[Genton et~al., 2011]{genton2011likelihood}
Genton, M.~G., Ma, Y., and Sang, H. (2011).
\newblock On the likelihood function of {G}aussian max-stable processes.
\newblock {\em Biometrika}, 98(2):481--488.

\bibitem[Glasserman, 2003]{glasserman2013monte}
Glasserman, P. (2003).
\newblock {\em {M}onte {C}arlo {M}ethods in {F}inancial {E}ngineering}.
\newblock Springer-Verlag New York.

\bibitem[Glasserman and Liu, 2010]{glasserman2010sensitivity}
Glasserman, P. and Liu, Z. (2010).
\newblock Sensitivity estimates from characteristic functions.
\newblock {\em Operations Research}, 58(6):1611--1623.

\bibitem[Huser and Davison, 2013]{huser2013composite}
Huser, R. and Davison, A.~C. (2013).
\newblock Composite likelihood estimation for the {B}rown--{R}esnick process.
\newblock {\em Biometrika}, 100(2):511--518.
\newblock \newline \url{https://doi.org/10.1093/biomet/ass089}.

\bibitem[Huser et~al., 2019]{huser2019full}
Huser, R., Dombry, C., Ribatet, M., and Genton, M.~G. (2019).
\newblock Full likelihood inference for max-stable data.
\newblock {\em Stat}, 8(1):e218.
\newblock \newline \url{https://doi.org/10.1002/sta4.218}.

\bibitem[Kabluchko et~al., 2009]{kabluchko2009stationary}
Kabluchko, Z., Schlather, M., and de~Haan, L. (2009).
\newblock Stationary max-stable fields associated to negative definite
  functions.
\newblock {\em The Annals of Probability}, 37(5):2042--2065.
\newblock \newline \url{https://doi.org/10.1214/09-AOP455}.

\bibitem[Kantha, 2008]{kantha2008tropical}
Kantha, L. (2008).
\newblock Tropical cyclone destructive potential by integrated kinetic energy.
\newblock {\em Bulletin of the American Meteorological Society},
  89(2):219--221.

\bibitem[Koch, 2017]{koch2017spatial}
Koch, E. (2017).
\newblock Spatial risk measures and applications to max-stable processes.
\newblock {\em Extremes}, 20(3):635--670.
\newblock \newline \url{https://doi.org/10.1007/s10687-016-0274-0}.

\bibitem[Koch, 2019]{koch2018spatialpowers}
Koch, E. (2019).
\newblock Spatial risk measures induced by powers of max-stable random fields.
\newblock {\em In revision for Extremes. Available at
  \url{https://arxiv.org/pdf/1804.05694v1.pdf}}.

\bibitem[Lamb and Frydendahl, 1991]{lamb1991historic}
Lamb, H. and Frydendahl, K. (1991).
\newblock {\em Historic Storms of the North Sea, British Isles and Northwest
  Europe}.
\newblock Cambridge University Press.

\bibitem[Oesting et~al., 2018]{oesting2018exact}
Oesting, M., Schlather, M., and Zhou, C. (2018).
\newblock Exact and fast simulation of max-stable processes on a compact set
  using the normalized spectral representation.
\newblock {\em Bernoulli}, 24(2):1497--1530.
\newblock \newline \url{https://doi.org/10.3150/16-BEJ905}.

\bibitem[Opitz, 2013]{opitz2013}
Opitz, T. (2013).
\newblock Extremaltprocesses: Elliptical domain of attraction and a spectral
  representation.
\newblock {\em Journal of Multivariate Analysis}, 122:409--413.

\bibitem[Padoan et~al., 2010]{padoan2010likelihood}
Padoan, S.~A., Ribatet, M., and Sisson, S.~A. (2010).
\newblock Likelihood-based inference for max-stable processes.
\newblock {\em Journal of the American Statistical Association},
  105(489):263--277.
\newblock \newline \url{https://doi.org/10.1198/jasa.2009.tm08577}.

\bibitem[Peng et~al., 2018]{peng2018new}
Peng, Y., Fu, M.~C., Hu, J.-Q., and Heidergott, B. (2018).
\newblock A new unbiased stochastic derivative estimator for discontinuous
  sample performances with structural parameters.
\newblock {\em Operations Research}, 66(2):487--499.
\newblock \newline \url{https://doi.org/10.1287/opre.2017.1674}.

\bibitem[Prahl et~al., 2015]{prahl2015comparison}
Prahl, B.~F., Rybski, D., Burghoff, O., and Kropp, J.~P. (2015).
\newblock Comparison of storm damage functions and their performance.
\newblock {\em Natural Hazards and Earth System Sciences}, 15:769--788.
\newblock \newline \url{https://doi.org/10.5194/nhess-15-769-2015}.

\bibitem[Prahl et~al., 2012]{prahl2012applying}
Prahl, B.~F., Rybski, D., Kropp, J.~P., Burghoff, O., and Held, H. (2012).
\newblock Applying stochastic small-scale damage functions to german winter
  storms.
\newblock {\em Geophysical Research Letters}, 39(6).
\newblock \newline \url{https://doi.org/10.1029/2012GL050961}.

\bibitem[Ribatet et~al., 2018]{PackageSpatialExtremes}
Ribatet, M., Singleton, R., and {R Core team} (2018).
\newblock Spatial{E}xtremes: {M}odelling {S}patial {E}xtremes.
\newblock {\em R package version 2.0-7}.

\bibitem[Schlather, 2002]{schlather2002models}
Schlather, M. (2002).
\newblock Models for stationary max-stable random fields.
\newblock {\em Extremes}, 5(1):33--44.

\bibitem[Schlather and Tawn, 2003]{schlather2003dependence}
Schlather, M. and Tawn, J.~A. (2003).
\newblock A dependence measure for multivariate and spatial extreme values:
  Properties and inference.
\newblock {\em Biometrika}, 90(1):139--156.
\newblock \newline \url{https://doi.org/10.1093/biomet/90.1.139}.

\bibitem[Simiu and Scanlan, 1996]{simiu1996wind}
Simiu, E. and Scanlan, R.~H. (1996).
\newblock {\em Wind Effects on Structures: Fundamentals and Applications to
  Design}.
\newblock John Wiley \& Sons, Inc.

\bibitem[Smith, 1990]{smith1990max}
Smith, R.~L. (1990).
\newblock Max-stable processes and spatial extremes.
\newblock {\em Unpublished manuscript}, University of North Carolina.

\bibitem[Varin and Vidoni, 2005]{varin2005note}
Varin, C. and Vidoni, P. (2005).
\newblock A note on composite likelihood inference and model selection.
\newblock {\em Biometrika}, 92(3):519--528.
\newblock \newline \url{https://doi.org/10.1093/biomet/92.3.519}.

\end{thebibliography}

\end{document}